\documentclass[onecolumn]{IEEEtran}

\usepackage{macro}

\title{Quadratically Constrained Two-way Adversarial  Channels}
\author{
\IEEEauthorblockN{
Yihan Zhang\IEEEauthorrefmark{1},
Shashank Vatedka\IEEEauthorrefmark{2}
Sidharth Jaggi\IEEEauthorrefmark{1}
}\\
\IEEEauthorblockA{
\IEEEauthorrefmark{1}Dept.\ of Information Engineering, The Chinese University of Hong Kong, Hong Kong SAR \\
\{\href{mailto:zy417@ie.cuhk.edu.hk}{zy417},\href{mailto:jaggi@ie.cuhk.edu.hk}{jaggi}\}@ie.cuhk.edu.hk \\
\IEEEauthorrefmark{2}Dept. of Electrical Engineering, Indian Institute of Technology, Hyderabad \\
\href{mailto:shashankvatedka@iith.ac.in}{shashankvatedka@iith.ac.in} 
}
}

\begin{document}
\maketitle

\begin{abstract}
	We study achievable rates of reliable communication in a power-constrained two-way additive interference channel over the real alphabet where communication is disrupted by a power-constrained jammer. This models the wireless communication scenario where two users Alice and Bob, operating in the full duplex mode, wish to exchange messages with each other in the presence of a jammer, James. Alice and Bob simultaneously transmit their encodings $ \underline{x}_A $ and $\underline{x}_B $ over $ n $ channel uses. It is assumed that James can choose his jamming signal $ \underline{s} $ as a noncausal randomized function of $ \underline{x}_A $ and $ \underline{x}_B $, and the codebooks used by Alice and Bob. Alice and Bob observe $ \underline{x}_A+\underline{x}_B +\underline{s}$, and must recover each others' messages reliably. In this article, we provide upper and lower bounds on the capacity of this channel which match each other and equal $ \frac{1}{2}\log\paren{\frac{1}{2} + \mathsf{SNR}} $ in the high-$\mathsf{SNR}$ regime (where $\mathsf{SNR}$, \emph{signal to noise ratios}, is defined as the ratio of the power constraints of the users to the power constraint of the jammer). We give a code construction based on lattice codes, and derive achievable rates for large $\mathsf{SNR}$. We also present upper bounds based on two specific attack strategies for James. Along the way, sumset property of lattices for the achievability and general properties of capacity-achieving codes for memoryless channels for the converse are proved, which might be of independent interest.
\end{abstract}


\section{Introduction}

Our work is motivated by jamming in multiuser wireless channels. Consider two users Alice and Bob who wish to exchange independent messages (assumed to be uniformly distributed in a set of size $ 2^{nR} $) with each other over the wireless medium. The communications is disrupted by an adversarial jammer, James, who injects additive noise into the channel. We assume that all three parties operate in the full-duplex mode, which means that they are able to transmit and receive simultaneously. Alice and Bob encode their messages into $ n $-length sequences $ \vbfx_A $ and $ \vbfx_B $ with real valued components and are simultaneously transmitted across the channel. At the same time, James transmits a jamming sequence $ \vbfs $. The channel is additive, and each user gets to observe $ \vbfy = \vbfx_A+\vbfx_B+\vbfs $. The goal of the two users is to recover each others' message reliably from this observation.

The signals transmitted by Alice, Bob and James are required to satisfy quadratic power constraints of $ nP,nP $, and $ nN $ respectively, i.e.,
\[
\Vert \vbfx_A\Vert^2\leq nP, \quad \Vert \vbfx_B\Vert^2\leq nP,\quad \Vert \vbfs\Vert^2\leq nN.
\]
We assume that James can select his jamming signal $ \vbfs $ as a noncausal function of $\vbfz\coloneq \vbfx_A+\vbfx_B $, and also the codebooks/coding strategies used by Alice and Bob. However, James has no additional information about the messages or the transmitted signals in addition to that revealed by $ \vbfx_A+\vbfx_B $ and the users' codebooks. We call this the $ (P,N) $ quadratically constrained two-way adversarial channel problem. This is illustrated in Fig.~\ref{fig:two-way-gaussian-adv-ic}.

The goal is to design sequences of encoders and decoders for Alice and Bob such that the probability of error of decoding the respective messages is vanishing in $ n $. Here, the randomness is over the encoding processes used by Alice and Bob, as well as the jamming signal. We say that a rate $ R $ is achievable if there exist sequences of codes for which the associated probabilities of decoding error is vanishing in $ n $, and the capacity is the supremum of all achievable rates.

In this paper, we give an upper bound on the capacity. We show that reliable communication is impossible for $ N\geq 3P/4 $. For $ N<3P/4 $, we show that the capacity is upper bounded by $ \frac{1}{2}\log \left(\frac{1}{2}+\frac{P}{N}\right) $. We also describe a coding scheme which shows that for sufficiently large values of $ P/N $, this bound is achievable.


The problem considered in this paper falls under the general setup of arbitrarily varying channels (AVCs), introduced by Blackwell et al.~\cite{blackwell-avc-1960}. This framework is a good model for channels where the noise statistics are arbitrary and unknown, and also where communication is disrupted by active adversaries.
Much of the literature has focused on point-to-point communication where Alice wants to send a message to Bob, and James attempts to jam the transmission. 
The quadratically constrained point-to-point AVC (also called the Gaussian AVC) was studied by Blachman~\cite{blachman-1962}, who gave upper and lower bounds on the capacity of the channel under the assumption that
James observes a noiseless version of the transmitted codeword (a.k.a.\ the \emph{omniscient} adversary). 
Later,  Hughes and Narayan~\cite{hughes-narayan-it1987}, and Csisz\'ar and Narayan~\cite{csiszar-narayan-it1991}, studied the problem with an ``oblivious'' James, who knows the codebook, but does not see the transmitted codeword.  
They showed that under an average probability of error metric, the capacity of the oblivious adversarial channel is equal to $ \frac{1}{2}\log\left(1+\frac{P}{N}\right) $ when $ P>N $ and zero otherwise.

Successive works have characterized the error exponent of the oblivious Gaussian AVC~\cite{thomas1991exponential}, capacity of the oblivious vector Gaussian AVC~\cite{hughes1988vectorgaussAVC}, and the Gaussian AVC with an unlimited amount of shared secret key between Alice and Bob~\cite{sarwate-gastpar-isit2006}. Sarwate~\cite{sarwate-spcom2012}, and later Zhang et al.~\cite{zhang-quadratic-isit}  studied the myopic AVC,  where James can choose his jamming vector as a function of the codebooks and a noisy copy of the transmitted signal. A related model was studied by Haddadpour et al.~\cite{haddadpour-isit2013}, who assumed that James knows the message, but not the exact codeword transmitted by Alice. Game-theoretic versions of the problems have also been considered in the literature, including the point-to-point case~\cite{medard}, with multiple antennas at the transmitter and receiver~\cite{baker-chao-siam1996}, and also the two-sender scenario~\cite{shafiee}. The list decoding capacity under the oblivious and omniscient cases were studied by Hosseinigoki and Kosut~\cite{hosseinigoki-kosut-2018-oblivious-gaussian-avc-ld} and Zhang and Vatedka~\cite{zhang-vatedka-2019-ld-real} respectively.


Multiuser AVCs have received attention only very recently. Multiple access channels with adversarial jamming were studied in~\cite{pereg2019capacitymac,sangwan2019byzantine}.  The capacity of the relay channel was analyzed in~\cite{pereg2019arbitrarily}, while~\cite{pereg2019arbitrarilybroadcast} gave inner and outer bounds on the capacity region of the degraded broadcast channel with side information at the encoder.

The work most related to our paper is that on the discrete-alphabet two-way additive channel with an adversarial jammer which was studied by Jaggi and Langberg~\cite{jaggi-langberg-2017-two-way}. They showed that for discrete additive channels over $ \bF_q $ where James' transmissions must satisfy a Hamming weight constraint of $ p $, the capacity is equal to $ 1-H_q(p) $. In other words, James can do no worse than transmitting random noise. Many of our ideas were inspired by this work, and we will elaborate on this in the coming sections. However, the conclusions that we can draw about the quadratically constrained case are different. In particular, the capacity is \emph{lower} that what we would get if the noise vector were Gaussian. A game-theoretic version of the quadratically constrained case we study here was studied by McDonald et al.~\cite{mcdonald2019two}.

\section{Overview of our results and techniques}
\subsection{Overview of results}
For a $ (P,N) $ quadratically constrained two-way adversarial channel, let $ \snr\coloneqq P/N $ be the \emph{signal-to-noise ratio}.
\begin{theorem}[Achievability]\label{thm:achievability}
For a $ (P,N) $ quadratically constrained two-way adversarial channel, given any sufficiently small constant $ \delta>0 $, if $ \snr>g(\delta) $ for some function $ g $ such that $ g(\delta)\xrightarrow{\delta\to0}\infty $, then both users can achieve rate $ \frac{1}{2}\log\paren{\frac{1}{2}+\frac{P}{N}} -\delta $. That is, $ C_A\ge\sqrbrkt{\frac{1}{2}\log\paren{\frac{1}{2}+\frac{P}{N}}}^+ $ and $ C_B\ge\sqrbrkt{\frac{1}{2}\log\paren{\frac{1}{2}+\frac{P}{N}}}^+ $.
\end{theorem}

\begin{theorem}[Converse]\label{thm:converse_scaleandbabble}
For a $ (P,N) $ quadratically constrained two-way adversarial channel, for any  sufficiently small constant $ \delta>0 $, neither of the users can achieve  rate larger than $ \frac{1}{2}\log\paren{\frac{1}{2}+\frac{P}{N}} + \delta $. That is,  $ C_A\le\sqrbrkt{\frac{1}{2}\log\paren{\frac{1}{2}+\frac{P}{N}}}^+ $ and $ C_B\le\sqrbrkt{\frac{1}{2}\log\paren{\frac{1}{2}+\frac{P}{N}}}^+ $.
\end{theorem}

\begin{corollary}[Capacity]\label{cor:capacity}
For a $ (P,N) $ quadratically constrained two-way adversarial channel, given any  sufficiently small constant $ \delta>0 $, if $ \snr>g(\delta) $ for some function $ g $ such that $ g(\delta)\xrightarrow{\delta\to0}\infty $, then $ C_A = C_B = \sqrbrkt{\frac{1}{2}\log\paren{\frac{1}{2}+\frac{P}{N}}}^+ $. 
\end{corollary}

Both our achievability and converse results can be trivially generalized to the asymmetric case, where the transmissions of Alice and Bob must satisfy 
\[
\Vert \vbfx_A\Vert^2\leq nP_A, \quad \Vert \vbfx_B\Vert^2\leq nP_B,
\]
James can independently jam the received vectors of Alice and Bob with jamming signals $ \vbfs_A $ and $ \vbfs_B $ which must satisfy
\[
 \Vert \vbfs_A\Vert^2\leq nN_A, \quad  \Vert \vbfs_B\Vert^2\leq nN_B.
\] 
Here Alice and Bob respectively receive $ \vbfy_A =\vbfx_A+\vbfx_B+\vbfs_A $ and $ \vbfy_B =\vbfx_A+\vbfx_B+\vbfs_B $.
For this $ (P_A,P_B,N_A,N_B) $ quadratically constrained two-way adversarial channel, let $ \snr_A\coloneqq P_B/N_A $ and $ \snr_B\coloneqq P_A/N_B $ be the SNRs of user one and two, respectively. Then we have
\begin{corollary}[Capacity, asymmetric case]\label{cor:capacity_asymm}
For a $ (P_A,P_B,N_A,N_B) $ quadratically constrained two-way adversarial channel, given any sufficiently small constants  $ \delta_1,\delta_2>0 $, if $ \snr_A>g_1(\delta_1) $ and $ \snr_B>g_2(\delta_2) $ for some functions $ g_1 $ and $ g_2 $ such that $ g_1(\delta_1)\xrightarrow{\delta_1\to0}\infty $ and $ g_2(\delta_2)\xrightarrow{\delta_2\to0}\infty $, then $ C_A =  \sqrbrkt{\frac{1}{2}\log\paren{ \frac{P_A}{P_A+P_B} + \frac{P_B}{N_A} }}^+ $ and $ C_B = \sqrbrkt{\frac{1}{2}\log\paren{ \frac{P_B}{P_A+P_B} + \frac{P_A}{N_B} }}^+ $. 
\end{corollary}

Note that the capacity $ \frac{1}{2}\log\paren{\frac{1}{2}+\frac{P}{N}} $ vanishes when $ N\ge 2P $ or $ \snr\le1/2 $. 
Though the capacity theorem indicates that $  \frac{1}{2}\log\paren{\frac{1}{2}+\frac{P}{N}} $ is the capacity in high-$\snr$  regime, we do not believe that this is tight in all regimes. 
Our intuition comes from the following improved converse result.  
We are able to push the boundary of zero-rate regime inward via   certain symmetrization strategy which we call \emph{$ \vbfz $-aware symmetrization}. 
\begin{theorem}[Converse]\label{thm:converse_symm}
For a $ (P,N) $ quadratically constrained two-way adversarial channel,  neither of the users can achieve positive rate if $ N>3P/4 $, or $ \snr < 4/3 $. 
\end{theorem}

Again, the above theorem can be trivially generalized to the asymmetric case which reads as follows.
\begin{theorem}[Converse]\label{thm:converse_symm_asymm}
For a $ (P_A,P_B,N_A,N_B) $ quadratically constrained two-way adversarial channel,  user one cannot achieve positive rate if $ N_A>\frac{2P_B+P_A}{4} $; user two cannot achieve positive rate if $ N_B>\frac{2P_A+P_B}{4} $. 
\end{theorem}

\subsection{Overview of proof techniques and related work}
Our ideas are inspired by~\cite{jaggi-langberg-2017-two-way}, which characterized the capacity of the discrete additive two-way channel with a jammer. They showed that using randomly expurgated \emph{linear} codebooks for Alice and Bob achieves the symmetric capacity $ 1-H_q(p) $, where $ H_q(p) $ denotes the $ q $-ary entropy of $ p $. This implies that James can do no worse than transmitting random noise.
It was also observed that neither linear codes nor uniformly random codebooks can achieve the capacity of this channel. Indeed, our codebook design closely mimics~\cite{jaggi-langberg-2017-two-way}: we use randomly expurgated lattice codebooks.

Unlike the discrete case studied in~\cite{jaggi-langberg-2017-two-way}, the setup we study in this paper poses additional challenges. In our setup, if the additive noise were random Gaussian with independent and identically distributed (i.i.d.) $ \cN(0,N) $ components, then the capacity  is equal to $ \frac{1}{2}\log_2\left(1+\frac{P}{N}\right) $. However, we give a converse to show that the capacity is in fact strictly below this. An important observation is that the capacity of the \emph{discrete} additive adversarial two-way channel is equal to the list decoding capacity (which also turns out to be the capacity with random noise). Unlike the discrete case, we show that the capacity of the $ (P,N) $ quadratically constrained two-way adversarial channel is (for large values of $ P/N $) strictly above the list decoding capacity.

\subsubsection{Proof techniques for upper bound} We provide three separate converse bounds for this problem by providing three attack strategies for James:
\begin{itemize}
	\item Clearly, if $ P\leq N $, then James can transmit a random codeword from Alice's (resp.\ Bob's) codebook chosen independently of everything else. Over the randomness in the choice of the codeword, Bob (resp.\ Alice) will then be unable to distinguish between the codewords transmitted by Alice (resp.\ Bob) and James. Hence, the capacity is zero.
	\item We can improve this to show that the capacity is zero for $ P\leq 3N/4 $. James independently selects a random codeword $ \vbfx_A' $ from Alice's codebook and transmits $ -\frac{1}{2}(\vbfz-\vbfx_A') $ whenever he has enough power. With high probability (w.h.p.), this attack vector satisfies the power constraint, and Bob receives $ 0.5\vbfx_B + 0.5(\vbfx_A+\vbfx_A') $. Bob cannot decide whether $ \vbfx_A $ or $ \vbfx_A' $ was transmitted, and therefore the probability of error is bounded away from zero.
	\item In the regime when $ N\leq 3P/4 $, we define a different attack for James. He can transmit $ \vbfs = -\alpha\vbfz+\vbfg $, where $ \vbfg\sim \cN(0,\gamma^2\bfI_n) $ and $ \alpha,\gamma $ are constants that can be optimized over. This instantiates an effective AWGN channel for Bob (resp.
	Alice) which implies that the capacity cannot exceed that of this effective AWGN channel. Upon optimizing the constants, we get that the capacity cannot be any larger than $ \frac{1}{2}\log\left(\frac{1}{2}+\frac{P}{N}\right) $. 
	To prove this, we analyze  general properties of the empirical properties of capacity-achieving codes for the AWGN channel (which we call AWGN-good codes) which we believe are novel results and might be of independent interest.
	We show that independent codewords chosen uniformly from any AWGN-good code are approximately orthogonal with high probability.
\end{itemize}

\subsubsection{Proof techniques for lower bound}
Let us briefly summarize the main elements of the achievability proof in~\cite{jaggi-langberg-2017-two-way}. 
A key step used is that even after expurgation, James is sufficiently confused about the transmitted codeword: if $ \cC_A $ and $ \cC_B $ are the codebooks obtained by independent random expurgations of the original linear code $ \cC $, then $ |\cC_A+\cC_B|\approx |\cC_A| $ and leaks very little information about the individual codewords to James. As a consequence, James cannot ``push'' the transmitted codeword to the nearest codeword in the corresponding codebook. The final step is to show that as long as the original code is list decodable with small list sizes, the expurgated code is uniquely decodable w.h.p. (over the randomness in the code expurgation).

Unlike the discrete case, we are not able to prove a matching lower bound on the capacity for all values of $ P,N $. We show that for sufficiently large $ P/N $, the capacity is $ C = \frac{1}{2}\log\left(\frac{1}{2}+\frac{P}{N}\right) $. The code for Alice and Bob is obtained by independently expurgating a lattice code with spherical shaping (to satisfy power constraint). What makes the quadratically constrained case more challenging than the discrete one is that due to the power constraint, the sum of two codewords leaks information about the individual codewords. However, if the original lattice code is suitably chosen, then we can show that James is sufficiently confused. Even then, following the approach in~\cite{jaggi-langberg-2017-two-way} gets us to only the list decoding capacity of $ \frac{1}{2}\log\frac{P}{N} $. To improve the rate, we introduce a proof technique inspired by~\cite{zhang-quadratic-arxiv}. We show that for every attack vector that James can instantiate, the effective decoding region is significantly smaller than $ \cB^n(\vbfy,\sqrt{nN}) $\footnote{Here $\cB^n\paren{\vu,r} $ denotes an $n$-dimensional Euclidean ball centered around $\vu$ of radius $r$.} w.h.p. (over the randomness in the choice of message).

To prove the lower bound, we show the following results which may be of independent interest:
\begin{itemize}
	\item Given any ``good'' lattice $ \Lf $ and the associated lattice codebook $ \cC \coloneq \Lf\cap \cB^n(\vzero,\sqrt{nP}) $, the sum of two independently and uniformly chosen codewords from $ \cC $ lies in a thin shell of radius $ \sqrt{2nP} $. We call this the \emph{typical sumset} of the lattice code. 
	\item For any vector $ \vv\in\bR^n $, a uniformly chosen codeword from $ \cC $ is almost orthogonal to $ \vv $. Consequently, two random codewords are almost orthogonal to each other.
	\item The above points reinforce the idea that codewords from a good lattice code have many properties similar to those chosen from random Gaussian codebooks.
	\item For any vector in the typical sumset, most pairs of codewords that sum to this vector respectively lie in a thin strip (See Fig. \ref{fig:sumset} and Fig. \ref{fig:strip}). This implies that given James' observation, the actual $ \vbfx_A $ (resp.\  $ \vbfx_B $) is uniformly distributed in a thin strip. 
	\item As a result of the above property, for every attack vector $ \vs $ that James can instantiate, the \emph{effective} decoding radius turns out to be $ \sqrt{n\frac{2PN}{2P - N}} = \sqrt{n\frac{N}{1 - \frac{N}{2P}}} $ which is even larger than $ \sqrt{nN} $.
	However, the (effective) decoding ball actually makes a relatively small intersection with the coding ball with high probability. We show that with high probability, the \emph{average/typical} effective decoding radius is $ \sqrt{n\frac{2PN}{2PN+N}} = \sqrt{n\frac{N}{1 + \frac{N}{2P}}} $ which is  smaller than $ \sqrt{n\frac{N}{1 - \frac{N}{2P}}} $, and actually also smaller than  $ \sqrt{nN} $ as one would naively assume. 
	We can then use the list decoding argument followed by the analysis of unique decodability as in~\cite{jaggi-langberg-2017-two-way}.
\end{itemize}

\section{Organization of the paper}
The rest of the paper is organized as follows. 
The notational conventions that we follow throughout the paper is fixed in Sec. \ref{sec:notation}.
Basics on concentration inequalities, high-dimensional geometry, information/coding theory and background on lattices are provided in Sec. \ref{sec:prelim} and Appendix \ref{sec:primer_lattices}.
We formally define the problem treated in this paper in Sec. \ref{sec:problem_formulation}.
In Sec. \ref{sec:beyond_listdec_cap}, to motivate our posterior estimation-style decoding rules, we provide intuition as to why in the high-$\snr$ regime, the capacity turns out to be $ \frac{1}{2}\log\paren{\frac{1}{2} + \frac{P}{N}} $, lower than the AWGN$(P,N)$ channel capacity $ \frac{1}{2}\log\paren{1+\frac{P}{N}} $, higher than the list-decoding capacity $ \frac{1}{2}\log\frac{P}{N} $. 
Sec. \ref{sec:achievability} contains a full proof of the achievability results.
Specifically,
\begin{enumerate}
	\item Our code construction based on expurgated lattice code is described in Sec. \ref{sec:code_design};
	\item Various error events to be considered in subsequent sections are defined in Sec. \ref{sec:error_events};
	\item In Sec. \ref{sec:sumset}, we prove sumset property of lattices which is useful in  the rest of the analysis and might be of independent interest elsewhere;
	\item In Sec. \ref{sec:estimating_alpha}, we show that $ \alpha $, the component of $ \vbfs $ that is parallel to $ \vbfz $ can be well estimated by the receiver;
	\item In Sec. \ref{sec:estimating_effdecrad}, we show that the effective decoding radius can also be well estimated by the receiver;
	\item The rate is properly set in Sec. \ref{sec:settingrate};
	\item  The \emph{average}  effective decoding radius is computed in Sec. \ref{sec:compute_avgdecrad};
	\item Finally, the probability of decoding error is bounded in Sec. \ref{sec:bouding_pe} using McDiarmid's inequality;
	\item Additionally, as a bonus section, in Sec. \ref{sec:improved_sumset}, we provide improved analysis of sumset property which yields bounds independent of $ \rcov(\Lf) $.
\end{enumerate}
Converse results are proved in Sec. \ref{sec:converse}.
Specifically,
\begin{enumerate}
	\item The scale-and-babble strategy that yields a tight outer bound in the high-$\snr$ regime is described in Sec. \ref{sec:scale_babble_strategy}; 
	\item The strategy is analyzed in Sec. \ref{sec:converse} using information inequalities;
	\item The bounding procedure of certain term $ Q $ (the probability that the scale-and-babble jamming vector violates James' power constraint) is deferred to Sec. \ref{sec:bound_q};
	\item Being useful in the converse argument and of independent interest, the proof of an empirical independence property that is universal to \emph{any} AWGN capacity-achieving code is further deferred to Sec. \ref{sec:empirical_properties_awgn};
	\item A symmetrization-type attack strategy which we call $ \vbfz $-aware symmetrization is described and analyzed in Sec. \ref{sec:z-aware_symm}.
\end{enumerate}
The paper is concluded in Sec. \ref{sec:concl_rk_open_problems} with some final remarks and open questions of future interests.

\section{Notation}\label{sec:notation}
\noindent\textbf{Conventions.}
Sets are denoted by capital letters in calligraphic typeface, e.g., $\cC,\cI$, etc. 
Random variables are denoted by lower case letters in boldface or capital letters in plain typeface, e.g., $\bfm,\bfx,\bfs,U,W$, etc. Their realizations are denoted by corresponding lower case letters in plain typeface, e.g., $m,x,s,u,w$, etc. Vectors (random or fixed) of length $n$, where $n$ is the blocklength without further specification, are denoted by lower case letters with  underlines, e.g., $\vbfx,\vbfs,\vx,\vs$, etc. The $i$-th entry of a vector $\vx\in\cX^n$ is denoted by $\vx(i)$  since we can alternatively think $\vx$ as a function from $[n]$ to $\cX$. Same for a random vector $\vbfx$. Matrices are denoted by capital letters in boldface, e.g., $\bfP,\mathbf{\Sigma}$, etc. Similarly, the $(i,j)$-th entry of a matrix $\bfG\in\bF^{n\times m}$ is denoted by $\bfG(i,j)$. We sometimes write $\bfG_{n\times m}$ to explicitly specify its dimension. For square matrices, we write $\bfG_n$ for short. Letter $\bfI$ is reserved for identity matrix.  

\noindent\textbf{Functions.}
We use the standard Bachmann--Landau (Big-Oh) notation for asymptotics of real-valued functions in positive integers. 

For two real-valued functions $f(n),g(n)$ in positive integers, we say that $f(n)$ \emph{asymptotically equals} $g(n)$, denoted $f(n)\asymp g(n)$, if 
\[\lim_{n\to\infty}\frac{f(n)}{g(n)} = 1.\]
For instance, $2^{n+\log n}\asymp2^{n+\log n}+2^n$, $2^{n+\log n}\not\asymp2^n$.
 We write $f(n)\doteq g(n)$ (read $f(n)$ dot equals $g(n)$) if the coefficients of the dominant terms in the exponents of $f(n)$ and $g(n)$ match,
\[\lim_{n\to\infty}\frac{\log f(n)}{\log g(n)} = 1.\]
For instance, $2^{3n}\doteq2^{3n+n^{1/4}}$, $2^{2^n}\not\doteq2^{2^{n+\log n}}$. Note that $f(n)\asymp g(n)$ implies $f(n)\doteq g(n)$, but the converse is not true.

For any $q\in\bR_{>0}$, we write $\log_q(\cdot)$ for the logarithm to the base $q$. In particular, let $\log(\cdot)$ and $\ln(\cdot)$ denote logarithms to the base two and $e$, respectively.

For any $\cA\subseteq\Omega$, the indicator function of $\cA$ is defined as, for any   $x\in\Omega$, 
\[\one{\cA}(x)\coloneqq\begin{cases}1,&x\in \cA\\0,&x\notin \cA\end{cases}.\]
At times, we will slightly abuse notation by saying that $\one{\sfA}$ is $1$ when event $\sfA$ happens and 0 otherwise. Note that $\one{\cA}(\cdot)=\indicator{\cdot\in\cA}$.

\noindent\textbf{Sets.}
For any two sets $\cA$ and $\cB$ with additive and multiplicative structures, let $\cA+\cB$ and $\cA\cdot\cB$ denote the Minkowski sum and Minkowski product of them which are defined as
\[\cA+\cB\coloneqq\curbrkt{a+b\colon a\in\cA,b\in\cB},\quad\cA\cdot\cB\coloneqq\curbrkt{a\cdot b\colon a\in\cA,b\in\cB},\]
respectively.
If $\cA=\{x\}$ is a singleton set, we write $x+\cB$ and $x\cdot\cB$ for $\{x\}+\cB$ and $\{x\}\cdot\cB$.

For $M\in\bZ_{>0}$, we let $[M]$ denote the set of first $M$ positive integers $\{1,2,\cdots,M\}$.

\noindent\textbf{Geometry.}
Let $ \normtwo{\cdot} $ denote the Euclidean/$L^2$-norm. Specifically, for any $\vx\in\bR^n$,
\[ \normtwo{\vx} \coloneqq\paren{\sum_{i=1}^n\vx_i^2}^{1/2}.\]

Let $\vol_n(\cdot)$ denote the $n$-dimensional Lebesgue volume of an Euclidean body (set with nonempty interior). Specifically, for any Euclidean body $\cA\subseteq\bR^n$,
\[\vol_n(\cA)=\intgover\cA\diff \vx=\intgover{\bR^n}\one{\cA}(\vx)\diff\vx,\]
where $\diff\vx$ denotes the differential of $\vx$ with respect to (w.r.t.) the Lebesgue measure on $\bR^n$. 
For convenience, the  subscript for dimension  will be dropped if no confusion will be caused. 

An $(n-1)$-dimensional Euclidean sphere centered at $\vx$ of radius $r$ is denoted by
\[\cS^{n-1}(\vx,r) \coloneqq \curbrkt{\vy\in\bR^n\colon \normtwo{\vy} =r}.\]
An $n$-dimensional Euclidean ball centered at $\vx$ of radius $r$ is denoted by
\[\cB^n(\vx,r) \coloneqq \curbrkt{\vy\in\bR^n\colon \normtwo{\vy} \le r}.\]
We will drop the  superscript for  dimension  when they are clear from the context.
When the center of the ball or sphere is not important, we also drop the first argument. 

Let $ V_n \coloneqq \vol_n(\cB^n(\vzero,1)) $. 

\noindent\textbf{Information theory.}
We use $H(\cdot)$ to interchangeably denote  the binary entropy function and the (differential or discrete) Shannon entropy; the exact meaning  will  be clear from the context.
In particular, if $P_{\vbfx}\colon \bR^n\to\bR_{\ge0}$ is the p.d.f. of a random vector $\vbfx$ in $\bR^n$, $H(\vbfx)$ denotes the differential entropy of $\vbfx\sim P_{\vbfx}$,
\[H(\vbfx)=-\intgover{\bR^n} P_{\vbfx}(\vx)\log P_{\vbfx}(\vx)\diff \vx.\]
If $ \cX $ is a discrete set, and $ P_{\vbfx}\colon \cX^n\to[0,1] $ is the p.m.f. of a random vector $ \vbfx $ on $ \cX^n $, $ H(\vbfx) $ denotes the Shannon entropy of $ \vbfx\sim P_{\vbfx} $,
\begin{align}
H(\vbfx)\coloneqq& \sum_{\vx\in\cX^n}P_{\vbfx}(\vx)\log\frac{1}{P_{\vbfx}(\vx)}. \notag
\end{align}
For any $p\in[0,1]$, $H(p)$ denotes the binary entropy 
\[H(p)=p\log\frac{1}{p}+(1-p)\log\frac{1}{1-p}.\]
The same convention is followed for mutual information.

\section{Preliminaries}\label{sec:prelim}

\noindent\textbf{Algebraic inequalities.}
\begin{fact}\label{fact:log_ineq}
For any $ x\ge0 $, $ \log(1-x)\le-x $ and $ \log(1+x)\le2x $. For any $ 0\le x\le1/2 $, $ \log(1-x)\ge-2x $. For any $ 0\le x\le 1 $, $ \log(1+x)\ge x $. 
\end{fact}
\begin{fact}\label{fact:exp_ineq}
For any $ x\ge-1 $ and $ L\notin(0,1) $, $ (1-x)^L\ge 1-Lx $. For any $ x\in[0,1] $ and any $ L\in\bZ_{\ge0} $, $ (1-x)^L\le\frac{1}{1+Lx} $. 
\end{fact}
\begin{corollary}\label{cor:log_ratio_ineq}
For any $ a,b\ge0 $, $ 0\le \eps\le a/2 $ and $ \delta\ge0 $, we have
\begin{align}
\log\frac{a-\eps}{b+\delta}\ge& \log\frac{a}{b} - \frac{2\eps}{a} - \frac{2\delta}{b}. \notag
\end{align}
\end{corollary}
\begin{proof}
\begin{align}
\log\frac{a - \eps}{b+\delta} =& \log(a-\eps) - \log(b+\delta) \notag \\
=& \log a + \log\paren{1 - \frac{\eps}{a}} - \log b - \log\paren{1 + \frac{\delta}{b}} \notag \\
\ge& \log\frac{a}{b} - \frac{2\eps}{a} - \frac{2\delta}{b}, \label{eqn:apply_log_ineq}
\end{align}
where Inequality \eqref{eqn:apply_log_ineq} follows from Fact \ref{fact:log_ineq} since $ 2\eps/a\le1/2 $  by assumption. 
\end{proof}


\noindent\textbf{Probability.}
\begin{fact}\label{fact:error_decomp}
For any events $\cA$ and $\cE$, $\prob{\cA}\le\prob{\cE} + \prob{\cA\cap\cE^c}$.
\end{fact}

\begin{lemma}[Markov's inequality]\label{lem:markov}
If $X$ is a nonnegative random variable, then for any $a>0$, $ \prob{X\ge a}\le \expt{X}/a $. 
\end{lemma}

\begin{lemma}[Chernoff bound]\label{lem:chernoff}
Suppose $ X_1,\cdots,X_N $ is a sequence of $N$ $ \curbrkt{0,1} $-valued independent random variables. 
Let $ X\coloneqq\sum_{i=1}^NX_i $. Then for any $ \delta\in[0,1] $,
\begin{align}
\prob{ X\ge(1+\delta)\expt{X} }\le& \exp\paren{ -\frac{\delta^2}{3}\expt{X} }, \notag \\ 
\prob{ X\le(1-\delta)\expt{X} }\le& \exp\paren{ -\frac{\delta^2}{2}\expt{X} }, \notag \\ 
\prob{ X\notin(1\pm\delta)\expt{X} }\le& 2\exp\paren{ -\frac{\delta^2}{3}\expt{X} }. \notag 
\end{align}
\end{lemma}

\begin{corollary}
\label{cor:expurgation}
Suppose a codebook $ \cC $ consists of $ 2^{nK} $ ($K>0$) codewords  in a set $ \cV $ and $ \cW $ is a subset of $ \cV $. 
Let $ \cC' $ denote the codebook obtained by independently removing  each codeword in $ \cC $ with probability $ 1-2^{-n\gamma} $ ($\gamma<K$). Then
\begin{align}
\prob{ \card{ \cC'\cap \cW }\notin(1\pm1/2)\expt{\card{\cC'\cap \cW}} }\le& 2\exp\paren{ -\frac{1}{12}2^{n(K-\gamma)} }. \notag
\end{align}
\end{corollary}

\begin{lemma}[Gaussian tail]\label{lem:gaussian_tail}
If $\bfg\sim\cN(0,\sigma^2)$, then for any $ \delta\ge0 $,  $\prob{\bfg\ge\delta}\le\exp\left(-\frac{\delta^2}{2\sigma^2}\right)$.
\end{lemma}
\begin{lemma}[$\chi^2$ tail]\label{lem:chisquared_tail}
If $\vbfg\sim\cN(\vzero,\sigma^2\bfI_n)$, then $\|\vbfg\|_2^2$ has (scaled) $\chi^2$-distribution and
\begin{align}
\prob{\normtwo{\vbfg}^2\ge n\sigma^2(1+\delta)}\le& \exp\left(-\frac{\delta^2}{4}n\right), \notag \\
\prob{\normtwo{\vbfg}^2\le n\sigma^2(1-\delta)}\le& \exp\left(-\frac{\delta^2}{2}n\right), \notag \\
\prob{\normtwo{\vbfg}^2\notin n\sigma^2(1\pm\delta)}\le& 2\exp\left(-\frac{\delta^2}{4}n\right). \notag
\end{align}
\end{lemma}

\begin{lemma}[McDiarmid's inequality]\label{lem:mcdiarmid_ineq}
Suppose $ X_1,\cdots,X_N $ is a sequence of $\cX$- valued random variables.
Let $ f\colon\cX^N\to\bR $ be a function. 
Define the Lipschitz constant of $f$ at the $i$-th input as
\begin{align}
\lip_i(f)\coloneqq& \max_{x_1,\cdots,x_i,x_i',\cdots,x_N\in\cX}\abs{ f(x_1,\cdots,x_i,\cdots,x_N) - f(x_1,\cdots,x_i',\cdots,x_N) }. \notag
\end{align}
Define the Lipschitz constant of $f$ as 
\begin{align}
\lip(f)\coloneqq& \max_{i\in[N]}\lip_i(f). \notag
\end{align}
Then we have
\begin{align}
\prob{ f(X_1,\cdots,X_N)>(1+\delta)\expt{f} }\le& \exp\paren{ -\frac{2\delta^2\expt{f}^2}{N\lip(f)^2} }, \notag \\
\prob{ f(X_1,\cdots,X_N)<(1-\delta)\expt{f} }\le& \exp\paren{ -\frac{2\delta^2\expt{f}^2}{N\lip(f)^2} }, \notag \\
\prob{ f(X_1,\cdots,X_N)\notin(1\pm\delta)\expt{f} }\le& 2\exp\paren{ -\frac{2\delta^2\expt{f}^2}{N\lip(f)^2} }. \notag
\end{align}
\end{lemma}

\begin{lemma}[First mean value theorem for integrals]\label{lem:mean_value_thm}
Let $ \Omega\subset\bR^n $ be a closed set. Let $ f\colon\Omega\to\bR $ be a continuous function and $ g\colon \Omega\to\bR $ be a integrable function that does not change sign. Then there exists $ \vx\in\Omega $ such that
\begin{align}
\intgover{\Omega} f(\vxi)g(\vxi)\diff\vxi = f(\vx)\intgover{\Omega}g(\vxi)\diff\vxi. \notag
\end{align}
\end{lemma}

\noindent\textbf{Geometry.}
\begin{fact}
$\frac{V_{n-1}}{V_{n}}\asymp\sqrt{\frac{n}{2\pi}}$.
\end{fact}
\begin{fact}
$\vol(\cB^n\paren{r}) = V_nr^n $ and $ V_n\asymp\frac{1}{\sqrt{\pi n}}\paren{{2\pi e}/{n}}^{n/2} $.
\end{fact}

The following lemma can be used to estimate the number of lattice points in any convex body, whose proof is along the lines of \cite{ordentlich2016simple}.
\begin{lemma}\label{lem:lattice_pt_bd}
For any  body $\cK\subset\bR^n$ and a lattice $\Lf$, the number of lattice point in $\cK$ is upper and lower bounded by
\begin{align*}
    \frac{\vol(\wc\cK)}{\vol(\Lf)}\le&\card{ \Lf \cap\cK}\le\frac{\vol(\wh\cK)}{\vol(\Lf)},
\end{align*}
where
\begin{align*}
    {\wc\cK}\coloneqq&\curbrkt{\vx\in\cK\colon d(\vx,\partial \cK)\ge\rcov(\Lf)},\\
    {\wh\cK}\coloneqq&\cK+\cV(\Lf).
\end{align*}
\end{lemma}

\begin{lemma}\label{lem:improved_lattice_pt_bd}
Let $\Lc\le\bR^n$ be a full rank lattice. Then for any $\vy\in\bR^n$ and any $r>\rcov(\Lc)/q$,
\begin{align}
\frac{1}{q^n}\card{\frac{1}{q}\Lc\cap\cB^n(\vy,r)}\in&\paren{\frac{r}{\reff(\Lc)}}^n\paren{1\pm\frac{\rcov(\Lc)}{qr}}^n. \notag
\end{align}
\end{lemma}

\begin{lemma}\label{lem:expurgation}
Let $\Lf$ be a full rank lattice in $ \bR^n $. Let $ \cV $ be a convex body and $ \cW $ be a convex subset of $ \cV $. Assume $ \vol(\wc\cV)>0 $ and $ \vol(\wc\cW)>0 $. 
Independently remove each lattice point in $\Lf$ with probability $ 1-2^{-n\gamma} $ and let $\Lf'$ denote the resulting configuration. 
Let $ \vbfx $ be a lattice point uniformly distributed on $ \Lf'\cap\cV $. Then
\begin{align}
\probover{\Lf'}{\probover{\vbfx\sim\Lf'\cap\cV}{\vbfx\in\cW}>3\frac{\vol(\wh\cW)}{\vol(\wc\cV)}}\le& \exp\paren{ -\frac{2^{-n\gamma}}{12}\frac{\vol(\wc\cW)}{\vol(\Lf)} }+\exp\paren{ -\frac{2^{-n\gamma}}{12}\frac{\vol(\wc\cV)}{\vol(\Lf)} }. \notag
\end{align}
\end{lemma}
\begin{proof}
\begin{align}
\probover{\Lf'}{\probover{\vbfx\sim\Lf'\cap\cV}{\vbfx\in\cW}>3\frac{\vol(\wh\cW)}{\vol(\wc\cV)}}=& \probover{\Lf'}{
	\frac{\card{\Lf'\cap\cW}}{\card{\Lf'\cap\cV}}>3\frac{\vol(\wh\cW)}{\vol(\wc\cV)}
} \notag \\
\le& \probover{\Lf'}{ \card{\Lf'\cap\cW}>\frac{3}{2}2^{-n\gamma}\frac{\vol(\wh\cW)}{\vol(\Lf)} } + \probover{\Lf'}{ \card{\Lf'\cap\cV}<\frac{1}{2}2^{-n\gamma}\frac{\vol(\wc\cV)}{\vol(\Lf)}} \notag \\
\le& \probover{\Lf'}{ \card{\Lf'\cap\cW}>\frac{3}{2}2^{-n\gamma}\card{\Lf\cap\cW} } + \probover{\Lf'}{ \card{\Lf'\cap\cV}<\frac{1}{2}2^{-n\gamma}\card{\Lf\cap\cV}} \label{eqn:bound1} \\
\le& \exp\paren{ -\frac{2^{-n\gamma}}{12}\card{\Lf\cap\cW} }+\exp\paren{ -\frac{2^{-n\gamma}}{12}\card{\Lf\cap\cV} } \label{eqn:expectation} \\
\le& \exp\paren{ -\frac{2^{-n\gamma}}{12}\frac{\vol(\wc\cW)}{\vol(\Lf)} }+\exp\paren{ -\frac{2^{-n\gamma}}{12}\frac{\vol(\wc\cV)}{\vol(\Lf)} }, \label{eqn:bound2} 
\end{align}
where Inequality \eqref{eqn:expectation} follows from Corollary \ref{cor:expurgation} by noting that
\begin{align}
\exptover{\Lf'}{\card{\Lf'\cap\cW}}=& 2^{-n\gamma}\card{\Lf\cap\cW}, \quad\exptover{\Lf'}{\card{\Lf'\cap\cV}}= 2^{-n\gamma}\card{\Lf\cap\cV}. \notag
\end{align}
Inequalities \eqref{eqn:bound1} and \eqref{eqn:bound2} are by Lemma \ref{lem:lattice_pt_bd}. 
\end{proof}

\noindent\textbf{Information theory.}
The following inequalities are standard in information theory. 
\begin{lemma}[Cardinality bound]\label{lem:cardinality_bd}
	If $X$ is a random variable distributed on a finite set $ \cX $, then $ H(X)\le\log|\cX| $.
\end{lemma}
\begin{lemma}[Entropy vs. variance bound]\label{lem:entropy_vs_var}
If $ X $ is a real-valued random variable, then $ H(X)\le\frac{1}{2}\log(2\pi e\var{X}) $. 
\end{lemma}
\begin{lemma}[Fano's inequality]\label{lem:fano_ineq}
If $ X\to Y\to \wh X $ is a Markov chain where $ X $ is distributed on $ [2^{nR}] $, then $ H(X|Y)\le H(X|\wh X)\le 1 + \prob{\wh X\ne X}nR $. 
\end{lemma}
\begin{lemma}[Data processing inequality]\label{lem:dpi}
If $ X\leftrightarrow Y \leftrightarrow  Z $ form a Markov chain, then $ I(X;Z)\le I(Y;Z) $. 
\end{lemma}

\begin{definition}[Quadratically constrained myopic adversarial channel]\label{def:myopic_channel}
A $ (P,\sigma^2,N) $-quadratically constrained myopic adversarial channel takes as input $ \vx_m $ which encodes $ m\in\cM $ subject to power constraint $ \normtwo{\vx}\le \sqrt{nP} $.
The transmitted codeword $ \vx $ is also broadcast through a AWGN$ (P,\sigma^2) $ channel and James receives $ \vbfz = \vx + \vbfsz $ where $ \vbfsz\sim\cN(\vzero,\sigma^2\bfI_n) $.
Based on $ \vbfz $ and the codebook $ \cC = \curbrkt{\vx_m}_{m\in\cM} $ (which is known to every party), James designs an adversarial noise vector $ \vs $ subject to power constrant  $ \normtwo{\vs}\le\sqrt{nN} $.
Once $ \vs $ is transmitted, the channel adds it to $ \vx $ and outputs $ \vbfy = \vx + \vs\in\cB^n\paren{\vzero,\sqrt{nP} + \sqrt{nN}} $.
Bob receiving $ \vbfy $ is required to reliably decode to the message corresponding to $ \vs $. 
\end{definition} 

\begin{definition}[List decodability of Euclidean codes]\label{def:list_dec}
A code $ \cC = \curbrkt{\vx_i}_{i \in\cM}\subseteq\bR^n $ is said to be $ (P,N,L) $-list decodable for some $ P,N>0 $ and $ L\in\bZ_{>0} $ if $ \normtwo{\vx_i}\le\sqrt{nP} $ for every $ i\in\cM $ and for any $ \vy\in\bR^n $, $ \card{ \cC\cap\cB^n\paren{\vy,\sqrt{nN}} }\le L $. 
\end{definition}

\begin{definition}[List decodability of infinite lattices]\label{def:list_dec}
An infinite lattice $ \Lf\le\bR^n $ is said to be $ (N,L) $-list decodable for some $ N>0 $ and $ L\in\bZ_{>0} $ if for every $ \vy\in\bR^n $, $ \card{\Lf\cap\cB^n\paren{\vy,\sqrt{nN}}}\le L $. 
\end{definition}

\begin{definition}[Normalized logarithmic density]\label{def:nld}
Let $\Lf \le\bR^n $ be an infinite lattice. 
The density of $\Lf$ is defined as 
\begin{align}
\Delta(\Lf) \coloneq& \limsup_{n\to\infty}\frac{\card{\Lf\cap[0,a)^n}}{a^n}. \notag
\end{align}
With slight abuse of notation, the normalized logarithmic density (NLD) of $\Lf$ is defined as 
$
R(\Lf) \coloneq\frac{1}{n} {\log\Delta(\Lf)}
$. 
NLD measure the ``rate'' of a lattice. 
\end{definition}

\section{Problem formulation}\label{sec:problem_formulation}
This paper is concerned with the following communication scenario. 
Two transmitters Alice and Bob want to exchange their messages $ \bfm_A $ and $ \bfm_B $ over a noisy channel governed by an adversary James. 
Specifically, we assume $ \bfm_A $ and $ \bfm_B $ are uniformly distributed in Alice's and Bob's message sets $ \cM $ and $ \cW $, respectively.
To fight against the adversarial noise to be introduced by James, Alice encodes her message into a length-$n$ real-valued vector $ \vbfx_A\coloneqq \vx_{\bfm_A} $ satisfying $ \normtwo{\vbfx_A}\le \sqrt{nP_A} $ for some channel parameter $ P_A>0 $. 
Similarly, Bob is allowed to encode his message into a codeword $ \vbfx_B \coloneqq \vx_{\bfm_B} $ satisfying $ \normtwo{\vbfx_B}\le\sqrt{nP_B} $ for some $ P_B>0 $.
By Kerckhoffs's principle, we assume that codebooks (collection of codewords) used by Alice and Bob are known to every party in the system. 
Codewords $ \vbfx_A $ and $ \vbfx_B $ are transmitted and added in the channel. 
James gets to know the sum $ \vbfz \coloneqq \vbfx_A+\vbfx_B $. 
Based on his observation, James designs adversarial vectors $ \vs_A $ and $ \vs_B $ such that $ \normtwo{\vs_A}\le\sqrt{nN_A} $ and $ \normtwo{\vs_B}\le\sqrt{nN_B} $ for some $ N_A>0 $ and $ N_B>0 $, respectively.
Once $ \vs_A $ and $ \vs_B $ are fed into the channel, Alice receives $ \vbfy_A \coloneqq \vbfz+\vs_A $ and Bob receives $ \vbfy_B \coloneqq \vbfz+\vs_B $.
The goal for Alice/Bob is to reliably decode the other transmitter Bob's/Alice's message w.h.p. over $ \bfm_A $ and $ \bfm_B $.

The channel model is depicted in Fig. \ref{fig:two-way-gaussian-adv-ic}.
\begin{figure}[htbp]
    \centering
    \includegraphics[width = \textwidth]{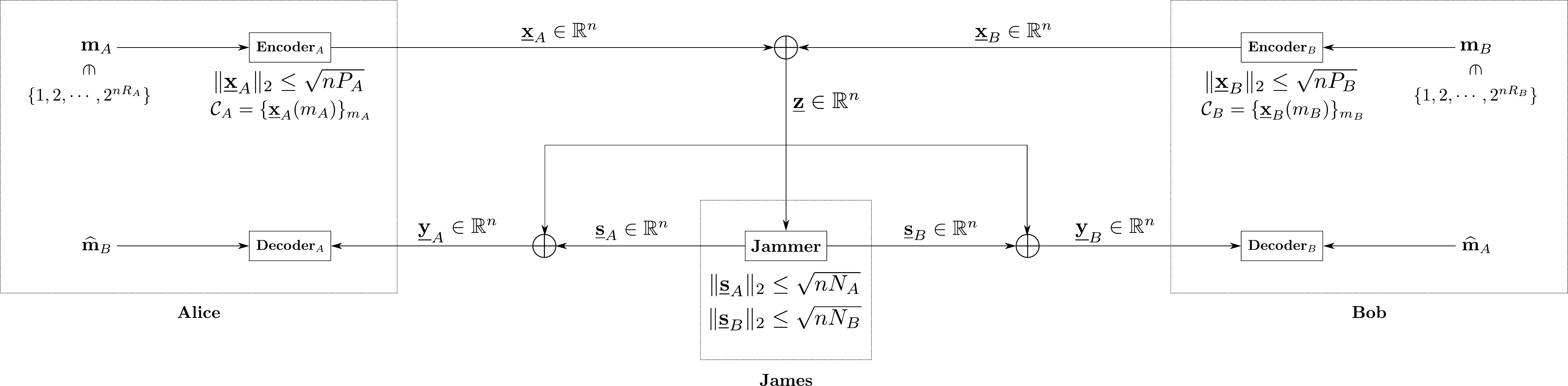}
    \caption{A quadratically constrained two-way adversarial  channel. In our proof, we assume $ P_A=P_B=P $ and $ N_A = N_B = N $. All of our results can be easily extended to the general asymmetric case.}
    \label{fig:two-way-gaussian-adv-ic}
\end{figure}

\begin{definition}[Quadratically constrained two-way adversarial  channel]
A $ (P_A,P_B,N_A,N_B) $-quadratically constrained two-way adversarial  channel is a function pair $ (W_A, W_B) $,
\begin{align}
&\begin{array}{rlll}
W_A\colon& \cB^n\paren{ \vzero,\sqrt{nP_A} } \times\cB^n\paren{ \vzero,\sqrt{nP_B} } \times\cB^n\paren{ \vzero,\sqrt{nN_A} } &\to &\cB^n\paren{ \vzero, \sqrt{nP_A} + \sqrt{nP_B} + \sqrt{nN_A} } \\
& (\vx_A,\vx_B,\vs_A) &\mapsto &\vy_A \coloneqq \vx_A+\vx_B+\vs_A
\end{array},
\notag \\
&\begin{array}{rlll}
W_B\colon& \cB^n\paren{ \vzero,\sqrt{nP_A} } \times\cB^n\paren{ \vzero,\sqrt{nP_B} } \times\cB^n\paren{ \vzero,\sqrt{nN_B} } &\to &\cB^n\paren{ \vzero, \sqrt{nP_A} + \sqrt{nP_B} + \sqrt{nN_B} } \\
& (\vx_A,\vx_B,\vs_B) &\mapsto &\vy_B \coloneqq \vx_A+\vx_B+\vs_B
\end{array}.
\notag
\end{align}
Here $ \vs_A $ and $ \vs_B $ 
are outputs of an \emph{arbitrary} jamming map pair $ (\jam_A,\jam_B) $ of the following form,
\begin{align}
&\begin{array}{rlll}
\jam_{A,\cC_A,\cC_B}\colon& \cC_A+\cC_B &\to &\cB^n\paren{\vzero,\sqrt{nN_A}} \\
& \vz &\mapsto & \vs_A
\end{array}, \notag \\
&\begin{array}{rlll}
\jam_{B,\cC_A,\cC_B}\colon& \cC_A+\cC_B &\to &\cB^n\paren{\vzero,\sqrt{nN_B}} \\
& \vz &\mapsto & \vs_B
\end{array}. \notag 
\end{align}
Note that both $ \jam_A $ and $ \jam_B $ can depend on $ (\cC_A,\cC_B) $.
\end{definition}

\begin{remark}
Throughout this paper, we focus on the symmetric case where $ P_A = P_B = P $ and $ N_A = N_B = N $. 
Such channels are denoted by $ (P,N) $-quadratically constrained two-way adversarial channels for short. 
All results can be trivially extended to the asymmetric case. 
We will state the extension without proof. 
\end{remark}

\begin{definition}[Code]
A code $ (\enc_A,\enc_B,\dec_A,\dec_B) $ 
consists of 
\begin{itemize}
	\item Alice's encoder 
	\begin{align}
	\begin{array}{rlll}
	\enc_A\colon& \cM&\to &\cB^n\paren{\vzero,\sqrt{nP_A}} \\
	& m & \mapsto & \vx_{A,m}
	\end{array};
	\notag
	\end{align}
	\item Bob's encoder 
	\begin{align}
	\begin{array}{rlll}
	\enc_B\colon& \cW &\to &\cB^n\paren{\vzero,\sqrt{nP_B}} \\
	& w & \mapsto & \vx_{B,w}
	\end{array};
	\notag
	\end{align}
	\item Alice's decoder 
	\begin{align}
	\begin{array}{rlll}
	\dec_A\colon &\cB^n\paren{\vzero, \sqrt{nP_A} + \sqrt{nP_B} + \sqrt{nN_A}} &\to &\cW \\
	& \vy_A & \mapsto & {\wh w}
	\end{array};
	\notag
	\end{align}
	\item Bob's decoder 
	\begin{align}
	\begin{array}{rlll}
	\dec_B\colon &\cB^n\paren{\vzero,\sqrt{nP_A} + \sqrt{nP_B} + \sqrt{nN_B}} &\to &\cW \\
	& \vy_B & \mapsto & {\wh m}
	\end{array}.
	\notag
	\end{align}
\end{itemize}
The dimension $n$ is called the blocklength of the code.

Let $ M\coloneqq|\cM| $ and $ W\coloneqq|\cW| $. 
The message sets $ \cM $ and $ \cW $ are identified with $ \sqrbrkt{M} $ and $ \sqrbrkt{W} $, respectively. 
The rate of a code $ (\cC_A,\cC_B) $ is defined as a pair $ (R_A,R_B) $ where $ R_A = R(\cC_A) \coloneqq\frac{\log M}{n} $ and $ R_B = R(\cC_B) \coloneqq\frac{\log W}{n} $. 

At times, we also abuse the notation and call the collection of codewords (images of the encoding maps) codebooks, i.e., $ \cC_A\coloneqq\curbrkt{\vx_{A,m}}_{m = 1}^M $, and $ \cC_B \coloneqq\curbrkt{\vx_{B,w}}_{w = 1}^{W} $. 
\end{definition}

\begin{definition}[Average probability of error]
The average probability of error of a codebook pair $ (\cC_A,\cC_B) $ associated with $ (\enc_A,\enc_B,\dec_A,\dec_B) $ used over a $ (P_A,P_B,N_A,N_B) $-quadratically constrained two-way adversarial  channel is defined as
\begin{align}
P_{\e,A}(\cC_A,\cC_B) \coloneqq& \max_{\jam_{A,\cC_A,\cC_B}} \probover{\substack{\bfm_A\sim\cM\\\bfm_B\sim\cW}}{ \dec_A\paren{ \vx_{A,\bfm_A} + \vx_{B,\bfm_B} + \jam_{A,\cC_A,\cC_B}(\vx_{A,\bfm_A} + \vx_{B,\bfm_B}) }\ne\bfm_B } \notag \\
=& \max_{\jam_{A,\cC_A,\cC_B}} \frac{1}{MW} \sum_{m_A\in\cM} \card{\curbrkt{ m_B\in\cW\colon \dec_A\paren{ \vx_{A,m_A} + \vx_{B, m_B} + \jam_{A,\cC_A,\cC_B}(\vx_{A, m_A} + \vx_{B, m_B}) }\ne m_B }}, \notag \\
P_{\e,B}(\cC_A,\cC_B) \coloneqq& \max_{\jam_{B,\cC_A,\cC_B}} \probover{\substack{\bfm_A\sim\cM\\\bfm_B\sim\cW}}{ \dec_B\paren{ \vx_{A,\bfm_A} + \vx_{B,\bfm_B} + \jam_{B,\cC_A,\cC_B}(\vx_{A,\bfm_A} + \vx_{B,\bfm_B}) }\ne\bfm_A } \notag \\
=& \max_{\jam_{B,\cC_A,\cC_B}} \frac{1}{MW} \sum_{m_B\in\cW} \card{\curbrkt{ m_A\in\cM\colon \dec_B\paren{ \vx_{A,m_A} + \vx_{B, m_B} + \jam_{B,\cC_A,\cC_B}(\vx_{A, m_A} + \vx_{B, m_B}) }\ne m_A }}, \notag
\end{align}
where the probabilities are taken over uniform selection of $ \bfm_A $ and $ \bfm_B $. 
\end{definition}

\begin{definition}[Achievable rate]
A rate pair $ (R_A,R_B) $ is said to be achievable if for any constant $ \beta_1,\beta_2 >0 $ and $ \eps_1,\eps_2>0 $, there exists a sequence of codes $ \curbrkt{(\cC_{A,n},\cC_{B,n})}_{n} $ for infinitely many $n  $ such that, there is an $ n_0 $, for every $ n>n_0 $,
\begin{itemize}
	\item   $ R_{A,n} \ge R_A - \beta_1 $ and $ R_{B,n} \ge R_B - \beta_2 $;
	\item the probabilities of Alice's and Bob's decoding errors vanish in $ n $, 
	\begin{align}
	P_{\e,A}(\cC_{A,n},\cC_{B,n}) \le& \eps_1, \notag \\
	P_{\e,B}(\cC_{A,n},\cC_{B,n}) \le& \eps_2. \notag 
	\end{align}
\end{itemize}
\end{definition}

\begin{definition}[Capacity]
The capacity $ (C_A,C_B) $ of a $ (P_A,P_B,N_A,N_B) $ quadratically constrained two-way adversarial  channel is defined as the supremum of all achievable rates, 
\begin{align}
C_A \coloneqq& \limsup_{\substack{\eps_1\downarrow0,\eps_2\downarrow0}} \limsup_{n\uparrow\infty} \max_{ \substack{\cC_{A,n} \subset\cB^n\paren{\vzero,\sqrt{nP_A}} \\ \cC_{B,n}\subset\cB^n\paren{\vzero,\sqrt{nP_B}} \\P_{\e,A}(\cC_{A,n},\cC_{B,n})\le\eps_1\\P_{\e,B}(\cC_{A,n},\cC_{B,n})\le\eps_2} } R( \cC_{A,n} ), \notag \\
C_B \coloneqq& \limsup_{\substack{\eps_1\downarrow0,\eps_2\downarrow0}} \limsup_{n\uparrow\infty} \max_{ \substack{\cC_{A,n} \subset\cB^n\paren{\vzero,\sqrt{nP_A}} \\ \cC_{B,n}\subset\cB^n\paren{\vzero,\sqrt{nP_B}}\\P_{\e,A}(\cC_{A,n},\cC_{B,n})\le\eps_1\\P_{\e,B}(\cC_{A,n},\cC_{B,n})\le\eps_2} } R( \cC_{B,n} ). \notag 
\end{align}
\end{definition}

\section{Beyond list decoding capacity: modified decoding rules} \label{sec:beyond_listdec_cap}
{Naively following the proof strategy in \cite{jaggi-langberg-2017-two-way}, we cannot prove any achievable rate that is larger than the list decoding capacity $\frac{1}{2}\log\frac{P}{N}$ above which the list size of \emph{any} code has to be exponential in the \emph{blocklength} $n$.}

\subsection{Decoding rule (informal)}
Bob computes 
\begin{align}
{\wh\alpha} \coloneq& 1-\frac{\langle\vbfy_B,\vbfx_B \rangle}{nP}, \notag \\
{\wt{\vbfy}_B}\coloneq& \vbfy_B-(1-\wh\alpha)\vbfx_B, \notag \\
r_{\text{dec}} \coloneq& \normtwo{\vbfy_B}^2 - 2(1-\wh\alpha)\inprod{\vbfy_B}{\vbfx_B}. \notag
\end{align}
If there is a single codeword 
\[\vbfx_A\in  \cC_A\cap \frac{1}{(1-\wh\alpha)}\cB^n(\widetilde{\vbfy}_B,r_{\text{dec}}),\]
then the decoder outputs the message associated to $\vbfx_A $. Otherwise, it declares an error. Alice's decoder operates likewise.

\subsection{Intuition}
We provide intuition behind our posterior-estimation-style decoding rule. 
All slack factors will be omitted in the rough calculations in this section.

Before proceeding, we would like to remind the readers of a fact from high dimensional geometry: as long as the $ \rcov(\Lf) $ is sufficiently small,  a random lattice point in a ball is concentrated near the surface of the ball and is approximately orthogonal to any given vector. 



Suppose a random pair of $\vbfx_A$ and $\vbfx_B$ is transmitted. They are concentrated in a thin shell near the sphere $\cS^{n-1} (\vzero,\sqrt{nP})$ and are almost orthogonal with high probability. Consider Bob trying to decode. Alice's decoding rule is symmetric. Bob receives $\vbfy_B=\vbfx_A+\vbfx_B+\vbfs_B$. From James' view, he observes $\vbfz=\vbfx_A+\vbfx_B$ which has norm about $\sqrt{2nP}$ w.h.p. There is a large number of pairs of $(\vbfx_A,\vbfx_B)$ which sums up to $\vbfz$. Moreover, each pair is approximately orthogonal and each of $\vbfx_A$ and $\vbfx_B$ is approximately uniformly distributed in a thin strip of radius $\sqrt{nP/2}$ perpendicular to $\vbfz$. 
James' jamming vector can be generically decomposed into directions parallel and perpendicular to $\vbfz$, 
\[
\vbfs_B=-\alpha\vbfz+\vbfsperp =-\alpha(\vbfx_A+\vbfx_B)+\vbfsperp 
,\]
where $\vbfsperp = \proj_{\vbfz^\perp}(\vbfs_B) $ is orthogonal to $\vbfz$. He has to choose $\alpha$ so that $\vbfs_B$ does not violate his power constraint, 
\[
\normtwo{-\alpha\vbfx_A-\alpha\vbfx_B+\vbfsperp }^2\approx2\alpha^2nP+\normtwo{\vbfsperp }^2\le nN
.\]
This imposes a constrain on $\alpha$: $ \abs{\alpha}\le\sqrt{\frac{N}{2P}} $. 
Under this decomposition,  Bob's received word can be written as $ \vbfy_B\coloneq (1-\alpha)\vbfx_A + (1-\alpha)\vbfx_B + \vbfsperp $.
From James' view, if $ \vbfz $ is typical (i.e., $ \normtwo{\vbfz}\in\sqrt{2\alpha^2P(1\pm\delta)} $), there is a large number of pairs of codewords $ (\vbfx_A,\vbfx_B) $ that were potentially transmitted (i.e., $ \vbfx_A+\vbfx_B = \vbfz $). 
Furthermore, these codewords are uniformly distributed in  a thin strip $ \cT $ near the surface of $ \cB^n\paren{\vzero,\sqrt{nP}} $, orthogonal to $ \vbfz $, of radius approximately $ \sqrt{nP/2} $. (See Fig. \ref{fig:eff_dec} for the geometry.) 
Hence the value of 
\[
\inprod{\vbfy_B}{\vbfx_B}=\inprod{(1-\alpha)\vbfx_A+(1-\alpha)\vbfx_B+\vbfsperp }{\vbfx_B}
\]
is well concentrated around
\[
0+(1-\alpha)\normtwo{\vbfx_B}^2+0\approx(1-\alpha)nP
\]
w.h.p. over  message selection. 
Thereby the value of $\alpha$ that was chosen by James  can be well estimated by Bob  via estimator $\wh\alpha\coloneqq1-\frac{\inprod{\vbfy_B}{\vbfx_B}}{nP}$. 
Then Bob computes $\vbfy_B-(1-\wh\alpha)\vbfx_B$ which in turn well approximates $(1-\alpha)\vbfx_A+\vbfsperp $. 
We now observe that, once James receives $ \vbfz $ and instantiates his jamming vector $ \vbfs $ based on  $ \vbfz $,
the effective channel to Bob is essentially $ \wt\vbfy_B = \wt\vbfx_A+\vbfsperp $ where  $\vbfsperp $ is fixed, $ \wt\vbfx_A\coloneq(1-\alpha)\vbfx_A $ and $\vbfx_A$ is  uniformly distributed in the strip $\cT \cap \cC_A $. 
It turns out that $\vbfx_A$ and $\vbfsperp $ are almost orthogonal w.h.p. 
Let $\wt P\coloneqq(1-\alpha)^2P$. Assuming James used up all his power (which is the worst case for Bob), let $\wt N\coloneqq N-2\alpha^2P$. 
For any $ \cA\subset\bR^n $, let $ \wt\cA $ denote $ (1-\alpha)\cA $.
One can compute the \emph{typical}  radius  of the decoding region induced by $ \wt\vbfx_A\in\wt\cT\cap\wt\cC_A $ under the translation of $ \vbfsperp $, which turns out to be approximately $ \sqrt{n\frac{\wt P\wt N}{\wt P+\wt N}} $. 
Now invoking techniques in \cite{jaggi-langberg-2017-two-way} allows us to show that as long as $\wt\Lf$ is $(\wt N,L)$ list-decodable with constant (independent of $n$) list size $L$, then $ \cC_A $ is uniquely decodable with probability $ 1-2^{-2^{\Omega(n)}} $ over expurgation. 
Hence the $ (\wt P,\wt N) $-list-decoding capacity $\frac{1}{2}\log\paren{\frac{\wt P}{\wt P\wt N/(\wt P+\wt N)}}=\frac{1}{2}\log\paren{1+\frac{\wt P}{\wt N}}$ can be achieved. 
Minimizing over James' choice of $\alpha$ subject to $\abs{\alpha}\le\sqrt{\frac{N}{2P}}$ gives that under the worst jamming strategy that James can impose, the rate $\frac{1}{2}\log\paren{\frac{1}{2}+\frac{P}{N}}$ can be  achieved. (The maximizer $\alpha_*$ turns out to be $\frac{N}{2P}$.)
This optimization problem coincides with the one that shows up in our converse.

\subsection{Some remarks}
\begin{enumerate}
	\item We assume $ N<P $ (otherwise the capacity is obviously 0). Hence $ \alpha\le\sqrt{\frac{N}{2P}}\le\sqrt{1/2}<1 $. 
	\item Let us examine what Bob gains by running the above decoder. Consider the worst channel to Bob that James could instantiate, which, in hindsight, corresponds to $ \alpha $ being $ \alpha_* \coloneq \frac{N}{2P} $.
	The original channel is $ \vbfy_B = \vbfx_A+\vbfx_B+\vbfs_B $. Naively cancelling his signal $\vbfx_B $, Bob gets $ \wt{\wt \vbfy}_B \coloneq \vbfx_A+\vbfs_B $. 
	Being over-pessimistic and assuming worst-case $\vbfs_B $, one would expect the SNR to be $ \wt{\wt \snr} \coloneq P/N $ and only $ \frac{1}{2}\log\frac{P}{N} $ (which coincides with the $(P,N)$-list-decoding capacity) could be achieved. 
	However, by running the above decoder, Bob in fact gets the effective channel 
	\[\wt\vbfy_B \coloneq \wt\vbfx_A + \vbfsperp = (1-\alpha_*)\vbfx_A+\vbfsperp = \frac{2P-N}{2P}\vbfx_A + \vbfsperp.\]
	Scaling everything back by $ \frac{1}{1-\alpha_*} $, Bob gets $ \vbfx_A + \frac{2P}{2P-N}\vbfsperp $.
	Note that $ \vbfsperp $ typically has power
	\[ \wt N \coloneq N-2\alpha_*^2P = \frac{N(2P-N)}{2P} . \]
	The effective SNR is hence 
	\[ \wt{\snr} = \frac{P}{ \paren{\frac{2P}{2P-N}}^2 \wt N} = \frac{P}{N} - \frac{1}{2} < \frac{P}{N} .\]
	At this point, it seems that our reduction can only lead to achievable rate $ \frac{1}{2}\log{\wt\snr} = \frac{1}{2}\log\paren{\frac{P}{N} - \frac{1}{2}} $, which is, somewhat counterintuitively, even less than the naive $ \frac{1}{2}\log\frac{P}{N} $. 
	However, it turns out that though $ \vbfsperp $ comes from an adversarial noise  $ \vbfs_B $, Bob can actually achieve $ \frac{1}{2}\log\paren{1+ \wt{\snr}} $, as if $ \vbfsperp $ was a Gaussian of the same variance. 
	The miracle is essentially due to the fact that James only gets to observe $ \vbfz $, rather than individual signals $ \vbfx_A $ and $ \vbfx_B $.
	As a consequence of measure concentration, the \emph{average/typical} effective decoding radius $ \overline{\vbfsperp} $ has much lower power: 
	\[ \overline{N} \coloneq \frac{\wt P\wt N}{\wt P+\wt N} = \frac{N(2P-N)^2}{2P(2P+N)}, \]
	where 
	\[\wt P \coloneq (1 - \alpha_*)^2P = \frac{(2P - N)^2}{4P}.\]
	Therefore the \emph{average/typical} effective SNR is 
	\[ \overline{\snr} \coloneq \frac{P}{\paren{\frac{2P}{2P-N}}^2\overline{N}} = \frac{P}{N} + \frac{1}{2} > \frac{P}{N}.\]
	Now Bob is in a good shape and he could transmit at the $ \paren{ P, \paren{\frac{2P}{2P-N}}^2\overline{N} } $-list-decoding capacity: 
	\[ \frac{1}{2}\log{ \overline{\snr} } = \frac{1}{2}\log\paren{1 + \wt\snr} = \frac{1}{2}\log\paren{\frac{1}{2} + \frac{P}{N}} > \frac{1}{2}\log\frac{P}{N}. \]
    \item In \cite{jaggi-langberg-2017-two-way}, it is claimed that when $R>1/2$, random codes also achieve capacity. This is not true for the quadratically constrained case. No matter how large the $\snr$ is, we cannot use a random spherical/ball code. We have to use codes with linear structures. This is because if codewords are independently and uniformly distributed in $ \cB^n\paren{\vzero,\sqrt{nP}} $, then given $ \vbfz = \vbfx_1 + \vbfx_2 $, with probability 1, $ (\vbfx_1,\vbfx_2) $ is the unique pair of codewords that sum up to $ \vbfz $. Then James knows $ \vbfx_1 $ and $ \vbfx_2 $ and is hence omniscient. 
    \item In general, suppose that Alice and Bob have power constraints $P_A $ and $ P_B $, respectively, and the noise vectors to them are subject to power constraints $N_A $ and $ N_B $, respectively. 
    Assume that $N_B<P_A+P_B $, otherwise $C_A = C_B = 0 $, obviously.
    Consider Bob.  Following exactly the same proof, in the high-$\snr$ regime, the rate given by the following optimization can be achieved.
    \begin{align}
    \begin{array}{rl}
    \maximize & \frac{1}{2}\log\paren{1+\frac{(1-\alpha^2)P_A }{N_B - \alpha^2(P_A+P_B) }} \\
    \subto & 0 \le \alpha \le \sqrt{\frac{N_B}{P_A+P_B}}.
    \end{array}
    \notag
    \end{align}
    Solving it, we have the maximizer $\alpha_* = \frac{N_B}{P_A+P_B} $ and the maxima is 
    \begin{align}
     C_B =& \frac{1}{2}\log\paren{ \frac{P_B}{P_A+P_B} + \frac{P_A}{N_B} }. \notag
    \end{align} 
    Exactly the same optimization also shows up in the scale-and-babble converse. Hence the above expression is the capacity of user Bob in the high-$\snr$ regime.

    Similarly, if we consider Alice, by the same calculations, we get  the capacity for user Alice 
    \begin{align}
    C_A =& \frac{1}{2}\log\paren{ \frac{P_A}{P_A+P_B} + \frac{P_B}{N_A} }. \notag
    \end{align}
\end{enumerate}

\section{Achievability}\label{sec:achievability}
\subsection{Code design}\label{sec:code_design}
Let $\Lf$ be a lattice obtained by lifting random linear codes $\cC'$ over $\bF_q$ via Construction-A. 
Specifically, let $\bfG\sim\bF_q^{n\times k}$ be a uniformly ranodm matrix.
The field size $q$ and dimension $k$ will be fixed later. 
Define the random linear code generated by $\bfG$ as $\cC' = \bfG\bF_q^k$.
Define  $\Lf =  \frac{1}{q}\Phi(\cC') + \bZ^n $, where $ \Phi\colon \bF_q\to\bZ $ is the natural embedding which maps  any field element $ j\in\bF_q $ to an integer $ j\in\bZ $. 
One can easily check that $\Lf$ is indeed a lattice.
Our lattice code is finally defined as $\cC \coloneqq \Lf\cap\cB^n\paren{\vzero,\sqrt{nP}} $.
It was proved in \cite{erez-2005-lattices-goodfor-everything} that the above ensemble of lattices is good for covering w.h.p.
\begin{lemma}[Theorem 2, \cite{erez-2005-lattices-goodfor-everything}]\label{lem:coveringgoodness}
Let $ \Lf $ be a lattice randomly drawn from the ensemble defined above whose parameters are restricted as follows.
Let $q$ and $k$ be such that
\begin{align}
q^k =& \frac{1}{\cB^n\paren{\reff(\Lf)}} \asymp \sqrt{n\pi}\paren{ \frac{n}{2\pi \reff(\Lf)^2} }^{n/2}. \notag
\end{align}
Fix $ \reff(\Lf) $ to a constant. 
Let $ k\le(1-c)n $ for some constant $ c\in(0,1) $ and $ k = \omega(\log^2n) $. 
This in turn imposes constraints on $ q $, $ q = \omega(\sqrt{n}) $ and $ \log q = o(n/\log n) $. Define $ d(n) \coloneqq \frac{\sqrt{n}}{2q} $. Then  $ \Lf $ is good for covering w.h.p., 
\begin{align}
\probover{\bfG}{ \frac{\rcov(\Lf)}{\reff(\Lf)} \le f(n) }\ge& (1-2^{-\nu(n)})(1-2^{-\nu(n)+1})^{\log n+\log\log q}, \label{eqn:coveringgoodness}
\end{align}
where $f(n)>1$ is defined as 
\begin{align}
f(n) \coloneqq& \paren{\frac{\rcov(\Lf)}{\rcov(\Lf) - 2d(n)}}n^{\lambda/n}2^{(\log n+\log\log q+1)\frac{\log q}{n}}, \notag
\end{align}
for some fixed constant $ \lambda>0 $;
in the RHS of Eqn. \eqref{eqn:coveringgoodness}, $ \nu(n) $ is defined as 
\begin{align}
\nu(n) \coloneqq& 2\log(\log n+\log\log q). \notag 
\end{align}
\end{lemma}

\begin{remark}
As $n$ approaches infinity,  by the choice of $ d(n)\xrightarrow{n\to\infty}0 $, for any constant $\lambda>0$ and by the choice of $ q $, respectively, we have
\begin{align}
\frac{\rcov(\Lf)}{\rcov(\Lf) - 2d(n)}\to& 1, \notag \\
n^{\lambda/n}\to&1, \notag \\
2^{(\log n+\log\log q+1)\frac{\log q}{n}}\to&1. \notag
\end{align}
Hence $ f(n)\xrightarrow{n\to\infty}1 $. That is, $ \rcov(\Lf)/\reff(\Lf) = 1+o_n(1) $ and $ \Lf $ is good for covering.
Also, note that, by the choice of $ \nu(n) $ and $p$, the RHS of Eqn. \eqref{eqn:coveringgoodness} approaches 1 from left as $ n\to\infty $,
\begin{align}
(1-2^{-\nu(n)})(1-2^{-\nu(n)+1})^{\log n+\log\log q} \to& 1. \notag
\end{align}
Hence the covering goodness property holds w.h.p. 
\end{remark}

\begin{remark}
Under the above choices of parameters, there are \emph{superexponentially} many lattice points in the unit cube $ [0,1]^n $ (and any of its integer translation $ [0,1]^n+\va $ where $ \va\in\bZ^n $).
For the purpose of coding, it is desirable to have exponentially many lattice points to keep the rate fixed. 
Indeed, we will scale $ \Lf $ properly momentarily.
\end{remark}

For the convenience of future calculations, define $ \rcov(\Lf) \coloneqq \sqrt{n\omega} $ and $ \reff(\Lf) \coloneqq \sqrt{n\tau} $. 

Scale $ \Lf $ properly so that 
\begin{align}
\card{\Lf\cap\cB^n\paren{\vzero,\sqrt{nP}}} \ge& 2^{\frac{1}{2}\log\paren{\frac{1}{2} + \frac{P}{N}} - \beta_0}, \label{eqn:rate_requirement}
\end{align}
where $ \beta_0 \coloneqq\beta + \beta_2 + \beta_3 $ and $ \beta, \beta_2,\beta_3 $ will be defined in Sec. \ref{sec:settingrate} (see Eqn. \eqref{eqn:apply_log_ratio_ineq}, \eqref{eqn:def_beta2} and \eqref{eqn:def_beta3}).
Since 
\begin{align}
\card{\Lf\cap\cB^n\paren{\vzero,\sqrt{nP}}} \ge& \frac{\cB^n\paren{ \vzero, \sqrt{nP} - \rcov(\Lf) }}{\cB^n\paren{ \reff(\Lf) }} \notag \\
=& \paren{ \frac{\sqrt{nP} - \rcov(\Lf)}{\reff(\Lf)} }^n \notag \\
=& \paren{ \sqrt{\frac{P}{\tau}} - \sqrt{\frac{\omega}{\tau}} }^n, 
\end{align}
the above requirement (Eqn. \eqref{eqn:rate_requirement}) translates to 
\begin{align}
\log\paren{ \sqrt{\frac{P}{\tau}} - \sqrt{\frac{\omega}{\tau}} } =& \frac{1}{2}\log\paren{ \frac{1}{2} + \frac{P}{N} } - \beta_0. \label{eqn:rcovreff_requirement}
\end{align}
For large $ \snr $ and covering-good $ \Lf $ (such that $ \omega\approx\tau $), Eqn. \eqref{eqn:rcovreff_requirement} implies that $ \tau\approx N $, i.e., $ \reff(\Lf)\approx\sqrt{nN} $.

Note that scaling does not change covering goodness since $ \rcov(\Lf) $ and $ \reff(\Lf) $ (and $ \rpack(\Lf) $) are scaling homogeneous, i.e., $ \rcov(a\Lf) = a\reff(\Lf) $, $ \reff(a\Lf) = a\reff(\Lf) $ for any $ a>0 $. 

Let $ \eps_n'>0 $ be a function such that $ \eps_n\xrightarrow{n\to\infty}0 $ and the decaying speed is lower than that of $ f(n) - 1 $ (by Lemma \ref{lem:coveringgoodness}, we know that $ f(n) - 1>0 $ and $ f(n )- 1\xrightarrow{n\to\infty}0 $).
Then, by Lemma \ref{lem:coveringgoodness}, w.h.p. a random lattice from the above ensemble satisfies
\begin{align}
\frac{\rcov(\Lf)}{\reff(\Lf)} = \frac{\sqrt{n\omega}}{\sqrt{n\tau}} = \sqrt{\omega/\tau} = 1+\eps_n', \notag
\end{align}
or
$
\omega/\tau = 1+\eps_n
$, 
for $ \eps_n \coloneqq2\eps_n'+\eps_n'^2\to0 $.

Over the randomness of picking $q$-ary linear codes, it was shown in \cite{zhang-vatedka-2019-ld-real} that the infinite lattice $ \Lf $ is  list decodable. 
\begin{lemma}[\cite{zhang-vatedka-2019-ld-real}]
Let $ \Lf $ be a lattice randomly drawn from the ensemble defined above whose parameters are restricted as follows. Let 
\begin{align}
q \coloneqq& \frac{\sqrt{P/N}}{2^{\beta /8} - 1}. \notag
\end{align}
Let $ \reff(\Lf)\ge \sqrt{nN}2^{\beta} $. Then $ \Lf $ is $ (\wt P', \wt N', L) $-list decodable w.h.p.,
\begin{align}
\frac{1}{n}\log\probover{\bfG}{ \Lf\text{ is not }(\wt P',\wt N',L)\text{-list decodable} } \le& -\frac{5}{8}\beta\frac{\log L}{\log q} + \log(5q)<0, \notag
\end{align}
where $ \wt P' $ and $ \wt N' $ are given by Eqn. \eqref{eqn:listdec_p} and Eqn. \eqref{eqn:listdec_n}, respectively,
and $ L\coloneqq 2^{ \cO\paren{ \frac{1}{\beta}\log^2\frac{1}{\beta} } } $. 
\end{lemma}

\begin{remark}
If we take $ \reff(\Lf) = \sqrt{nN}2^\beta $, then the density of $\Lf$ is
\begin{align}
R(\Lf) =& \frac{1}{\vol(\Lf)} \notag \\
=& \frac{1}{\cB^n\paren{\reff(\Lf)}} \notag \\
\asymp& \frac{1}{\frac{1}{\sqrt{\pi n}}\paren{2\pi e/n}^{n/2}\reff(\Lf)^n} \notag \\
\asymp& \frac{\sqrt{\pi n}}{\paren{2\pi e N2^{2\beta}}^{n/2}}. \notag
\end{align}
Hence $ (N,L) $-list decodable lattices can achieve NLD
\begin{align}
R(\Lf) =& \frac{1}{n}\log\frac{\sqrt{\pi n}}{\paren{2\pi e N2^{2\beta}}^{n/2}} \notag \\
=& \frac{\log\sqrt{\pi n}}{n} + \frac{1}{2}\log\frac{1}{2\pi eN2^{2\beta}} \notag \\
\xrightarrow{n\to\infty}& \frac{1}{2}\log\frac{1}{2\pi eN} - \beta,  \notag 
\end{align}
and list size $ L = 2^{\cO\paren{\frac{1}{\beta}\log^2\frac{1}{\beta}}} $.
\end{remark}

\begin{remark}
For small constant $ \beta>0 $, $q$ scales as $ \cO(1/\beta) $ and $k$ scales as $ \cO\paren{n/\log\frac{1}{\beta}} $. 
Our choice of parameters falls into the regime specified in Lemma \ref{lem:coveringgoodness} after proper scaling. 
\end{remark}
\begin{remark}
Roughly speaking, 
\begin{align}
{\wt P'}\approx& \frac{(2P-N)^2}{4P},\quad {\wt N'}\approx \frac{2NP-N^2}{2P}. \notag
\end{align}
Note that, by setting $ \alpha = \frac{N}{2P} $, they can be written as  $ \wt P' \approx (1-\alpha)^2P $ and $ \wt N' \approx N - 2\alpha^2P $. 
In fact, the list-decoding capacity $ \frac{1}{2}\log\frac{\wt P'}{\wt N'} - \beta $ happen to equal the two-way adversarial channel capacity $ \frac{1}{2}\log\paren{ \frac{1}{2} + \frac{P}{N} } - \beta_2 - \beta_3 - \beta $ under the above choices of $ \wt P' $ and $ \wt N' $. 
This coincidence matches our intuition in Sec. \ref{sec:beyond_listdec_cap}. 
\end{remark}

By union bound, w.h.p. a random lattice from the above ensemble is simultaneously good for covering and $ (\wt P', \wt N',L) $-list decodable.
Fix $ \Lf $ to be any of such lattice. 

Given two identical copies of $\cC$, independently expurgate them and get $\cC_A$ and $\cC_B$ as Alice's and Bob's codebooks, respectively. Specifically,   each codeword in $ \cC $ is independently picked into $ \cC_A $ with probability $ 2^{-\gamma n} $ for certain sufficiently small constant $ \gamma>0 $. Bob's codebook $ \cC_B $ is obtained in the same manner independently. By Chernoff bound (Corollary \ref{cor:expurgation}), we have that $ |\cC_A| $ and $ |\cC_B| $ are at least $\frac{1}{2}\cdot 2^{ \frac{1}{2}\log\paren{ \frac{1}{2}+\frac{P}{N} } - \beta_0 - \gamma } $ with probability doubly exponentially close to 1. 
The rate incurs essentially no loss if $\gamma$ is sufficiently small.

\begin{remark}
In the proof in subsequent sections, the probability is only taken over message selection $ \bfm_A\sim\cM, \bfm\sim\cW $ and the expurgation process. 
The base lattice $ \Lf $ is fixed throughout the paper. 
\end{remark}

\subsection{Error events}\label{sec:error_events}
Take a $ \sqrt{n\eta'} $-net $ \cS $ of $ \cB^n\paren{\vzero,\sqrt{nN}} $ such that for every $ \vs\in\cB^n\paren{\vzero,\sqrt{nN}} $, there is $ \vs'\in\cS $ satisfying $ \normtwo{ \vs - \vs' }\le \sqrt{n\eta'} $. 
We can take a lattice $ \Lf_{\cS} $ of covering radius $ \sqrt{n\eta'} $. If $ \Lf_\cS $ is good for covering, then the size of the net $\cS$ is at most 
\begin{align}
|\cS| =& \card{ \cB^n\paren{\vzero,\sqrt{nN}}\cap\Lf_\cS } \notag \\
\le& \frac{\vol\paren{ \cB^n\paren{\sqrt{nN} + \sqrt{n\eta'}} }}{\vol\paren{ \cB^n\paren{\sqrt{n\eta'}} }} \notag \\
=& \paren{ \sqrt{\frac{N}{\eta'}} + 1 }^n \notag \\
=& \paren{ \frac{1}{\eta'} }^{\cO(n)}. \notag 
\end{align}
\begin{itemize}
    \item[$\elen$] The transmitted $\vbfx_A$ or $\vbfx_B$ is not close to the surface of the codebook,
    \begin{align}
    \elen\coloneqq&\curbrkt{\normtwo{\vbfx_A}\le\sqrt{nP(1-\zeta_1)}}\cup\curbrkt{\normtwo{\vbfx_B}\le\sqrt{nP(1-\zeta_1)}}.
    \label{eqn:elen_def}
    \end{align}
    \item[$\einprod$] The transmitted codewords $\vbfx_A$ and $\vbfx_B$ are not approximately orthogonal, 
    \begin{align}
    \einprod\coloneqq&\curbrkt{\abs{\inprod{\vbfx_A}{\vbfx_B}}\ge nP\zeta_1}.
    \label{eqn:einprod_def}
    \end{align}
    \item[$\ez$] The sum of transmitted codeword pair $\vbfx_A$ and $\vbfx_B$ has length deviating from its typical value $\sqrt{2nP}$,
    \begin{align}
    \ez\coloneqq&\curbrkt{\normtwo{\vbfz}\notin\sqrt{2nP(1\pm\delta)}}.
    \label{eqn:ez_def}
    \end{align}
    \item[$\ezz$] The norm of $ \normtwo{\vbfsperp } $ deviates from its typical value $ \sqrt{n\wt N} = \sqrt{n(N-2\alpha^2P)} $, 
    \begin{align}
    \ezz\coloneqq&\curbrkt{ \normtwo{\vbfsperp }\notin \sqrt{n(N - 2\alpha^2P(1\pm\delta))} }. 
    \label{eqn:ezz_def}
    \end{align}
    \item[$\cE$] The union of $\elen$, $\einprod$ and $\ez$, i.e., the transmitted $ \vbfx_A $ and $ \vbfx_B $ are not jointly typical,
    \begin{align}
    \cE\coloneqq&\elen\cup\einprod\cup\ez.
    \label{eqn:e_def}
    \end{align}
    \item[$\esumset$] Given James' received $ \vbfz $, there is not a large number of pairs  of codewords  $ (\vbfx_A,\vbfx_B) $ such that $\vbfx_A+\vbfx_B = \vbfz$.
    \begin{align}
    \esumset \coloneqq& \curbrkt{ \card{\curbrkt{(\vx_A,\vx_B)\in\cC_A\times\cC_B\colon\vx_A+\vx_B = \vbfz }} \le 2^{n(F_1-o(1))}} .
    \label{eqn:esumset_def}
    \end{align}
    \item[$\cE_\cT$] Codeword pairs $(\vbfx_A,\vbfx_B)$ which sum up to $\vbfz$ are not in a thin strip $\cT$ which will be defined later,
    \begin{align}
    \cE_\cT\coloneqq&\curbrkt{\vbfx_A\notin\cT} \cup \curbrkt{\vbfx_B\notin\cT}.
    \label{eqn:et_def}
    \end{align}
    \item[$\cE_1 $] Codewords $\vbfx_A $ or $ \vbfx_B $ in the strip have norm  much less than $\sqrt{nP}$,
    \begin{align}
    \cE_1\coloneqq\curbrkt{ \normtwo{\vbfx_A}\le\sqrt{n(P-c_2)} }\cup\curbrkt{ \normtwo{\vbfx_B}\le\sqrt{n(P-c_2)} }.
    \label{eqn:e1_def}
    \end{align}
    \item[$\cE_2$] Codeword pairs  $(\vbfx_A,\vbfx_B)$ in the strip $\cT$ that sum up to James' observation $\vbfz$ are not approximately orthogonal,
    \begin{align}
    \cE_2\coloneqq&\curbrkt{\abs{\inprod{\vbfx_A}{\vbfx_B}}\ge nP\theta}.
    \label{eqn:e2_def}
    \end{align}
    \item[$\cE_3$] Codewords  $\vbfx_A$ or $ \vbfx_B $ in the strip $\cT$ are not approximately orthogonal to $\vbfsperp $,
    \begin{align} 
    \cE_3\coloneqq&\curbrkt{\abs{\inprod{\vbfx_A}{\vbfsperp }}\ge n\zeta}.
    \label{eqn:e3_def}
    \end{align}
    Here $\vbfsperp  = \proj_{\vbfz^\perp}(\vbfs)$ is the projection of $\vbfs$ to the subspace orthogonal to $\vbfz$. 
    \item[$\cE'$] The union of $ \cE_1 $, $ \cE_2 $ and $ \cE_3 $, i.e., codewords $ \vbfx_A$, $\vbfx_B $ in the strip $ \cT $ and any given $ \vbfsperp $ are not jointly typical,
    \begin{align}
    \cE'\coloneq& \cE_1\cup\cE_2\cup\cE_3. \label{eqn:eqn:eprime_def}
    \end{align}
    \item[$\ealpha$] Bob' estimate $\wh\alpha$  is imprecise w.r.t. the true value $\alpha$ used by James,
    \begin{align}
    \ealpha\coloneqq&\curbrkt{\wh\alpha\notin\alpha\pm\xi}.
    \label{eqn:ealpha_def}
    \end{align}
    Here $\alpha$ is the fractional length of $\vbfs$ along the direction of $\vbfz$, i.e.,  $ \normtwo{\proj_{\vbfz}(\vbfs)} = \alpha\normtwo{\vbfz} $
    \item[$\edecrad$] Bob's estimate of decoding radius w.r.t. the effective channel deviates from its typical value $ \sqrt{n\wt N} $.
    \begin{align}
    \edecrad\coloneqq&\curbrkt{ \normtwo{\wh\vbfsperp }\notin\sqrt{n\paren{ N-2\alpha^2P(1\mp\delta)\pm\mu }} }. 
    \label{eqn:edecrad_def}
    \end{align}
    \item[$ \eavgrad $] The (normalized) effective decoding radius deviates from its typical value (averaged over the strip) $ \frac{\wt P\wt N}{\wt P+\wt N} $.
    \begin{align}
    \eavgrad\coloneqq&\curbrkt{ \wh\bfr\notin\frac{\wt P\wt N}{\wt P+\wt N}\pm\nu }. 
    \label{eqn:eavgrad_def}
    \end{align}
\end{itemize}
The dependencies among the above events are plotted in Fig. \ref{fig:event_dependency}, where an arrow $ \cE_i\to\cE_j $ from event $ \cE_i $ to event $ \cE_j $ denotes the inclusion $ \cE_i\subseteq\cE_j $. 
\begin{figure}[htbp]
	\centering
	\includegraphics[width=0.4\textwidth]{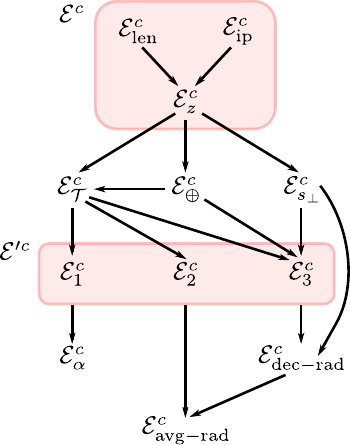}
	\caption{Event dependencies. An arrow $ \cE_i\to\cE_j $  denotes the inclusion $ \cE_i\subseteq\cE_j $. Events are defined in Sec. \ref{sec:error_events}.}
	\label{fig:event_dependency}
\end{figure}
We would like to point out that most ``good'' events are implied purely by $ \ez^c $. 
The proofs mostly follow from geometric arguments. 

In Sec. \ref{sec:sumset}, event $ \elen $ is analyzed in Lemma \ref{lem:lattice_pt_norm}, $ \einprod $ in Lemma \ref{lem:lattice_pt_orthogonal} and \ref{lem:xa_ip_xb}, $ \ez $ in Lemma \ref{lem:length_z}, and $ \ezz $ in Corollary \ref{cor:bound_zprime}.
In Sec. \ref{sec:estimating_alpha}, the event $ \esumset $ is analyzed in Lemma \ref{lem:many_confusing_pairs}, $ \cE_\cT $ in Lemma \ref{lem:many_cw_in_strip}, $ \cE_2 $ in Lemma \ref{lem:bound_cos_in_strip} and Corollary \ref{cor:bound_ip_in_strip}, $ \cE_1 $ in Corollary \ref{cor:conc_norm_cw_in_strip}, $ \cE_3 $ in Lemma \ref{lem:xa_ip_zprime}, and $ \ealpha $ in Lemma \ref{lem:estimate_alpha}. 
In Sec. \ref{sec:estimating_effdecrad}, the event $ \edecrad $ is analyzed in Lemma \ref{lem:estimate_decrad} and Corollary \ref{cor:estimate_decrad}.
In Sec. \ref{sec:compute_avgdecrad}, the event $ \eavgrad $ is analyzed in Lemma \ref{lem:compute_avgdecrad}.
Finally, the average probability of decoding error is bounded in Lemma \ref{lem:bounding_pe} in Sec. \ref{sec:bouding_pe}.

\subsection{Sumset property}\label{sec:sumset}
For notational convenience, we write $ \reff = \reff(\Lf) $ and $ \rcov = \rcov(\Lf) $.

\begin{lemma}\label{lem:lattice_pt_norm}
A lattice point chosen uniformly from $ \cC $ is concentrated within a thin shell near the sphere w.h.p. For any constant $ \zeta\in(0,1) $,
\[\probover{\vbfx\sim\cC}{ \normtwo{\vbfx}\le\sqrt{nP(1-\zeta)} }\le\paren{\frac{ \sqrt{P(1-\zeta)}+\sqrt{\omega} }{\sqrt{P} - \sqrt{\omega}}}^n.\]
\end{lemma}
\begin{proof}
\begin{align*}
    \probover{\vbfx\sim\cC}{ \normtwo{\vbfx}\le\sqrt{nP(1-\zeta)} }=&\frac{\card{ \Lf\cap\cB^n(\vzero,\sqrt{nP(1-\zeta)}) }}{\card{ \Lf\cap\cB^n(\vzero,\sqrt{nP}) }}\\
    \le&\frac{\vol\paren{\cB^n\paren{\vzero,\sqrt{nP(1-\zeta)}+\rcov} }}{\vol\paren{\cB^n\paren{\vzero,\sqrt{nP} - \rcov}}}\\
    =&\paren{\frac{\sqrt{nP(1-\zeta)} + \rcov}{\sqrt{nP} - \rcov}}^n \\
    =& \paren{\frac{ \sqrt{P(1-\zeta)}+\sqrt{\omega} }{\sqrt{P} - \sqrt{\omega}}}^n .
\end{align*}
\end{proof}

\begin{lemma}\label{lem:lattice_pt_orthogonal} 
For  any vector $\vv\in\cS^{n-1}(\vzero,\sqrt{nP})$, a  lattice point uniformly drawn from $ \cC $ is almost orthogonal to $\vv$ w.h.p. For any constant $ \zeta\in(0,1) $,
\begin{align*}
    \probover{\vbfx\sim\cC}{ \abs{ \cos\angle_{{\vbfx},{\vv}} }\ge \zeta }\le& 2\paren{\frac{\sqrt{P(1-\zeta^2)} + \sqrt{\omega}}{\sqrt{P} - \sqrt{\omega}}}^n .
\end{align*}
\end{lemma}
\begin{proof}
Let $\cB\coloneqq\cB^n(\vzero,\sqrt{nP})$. Define
\begin{align}
\cB_1 \coloneqq& \curbrkt{ \vx\in\cB\colon \cos\angle_{\vx,\vv}\ge\zeta }, \notag \\
\cB_1 \coloneqq& \curbrkt{ \vx\in\cB\colon \cos\angle_{\vx,\vv}\le-\zeta }. \notag
\end{align}
Geometrically (Fig. \ref{fig:rand_lattice_pt_in_ball_ortho}), $ \cB_1 $ and $ \cB_2 $ are the blue and pink cones restricted to the ball $ \cB $. 
\begin{figure}[htbp]
    \centering
    \includegraphics[width = .5\textwidth]{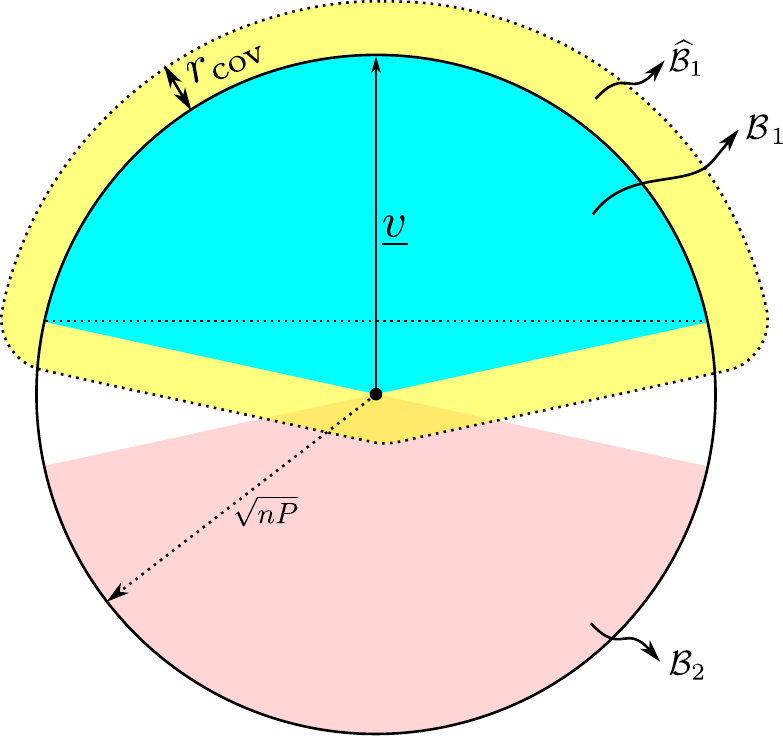}
    \caption{A random lattice point in a ball is approximately orthogonal to any given vector $\vv$ with high probability.}
    \label{fig:rand_lattice_pt_in_ball_ortho}
\end{figure}
Then
\begin{align}
    \prob{ \abs{\cos\angle_{{\vbfx},{\vv}}}\ge \zeta }=\frac{\card{\Lf\cap(\cB_1\cup\cB_2)}}{\card{\Lf\cap\cB}}. \label{eqn:lattice_pt_ortho}
\end{align}
We  apply Lemma \ref{lem:lattice_pt_bd} to upper bound the numerator and lower bound the denominator. To this end, we only need to upper bound the volume of $\wh{\cB_1\cup\cB_2}$ which is at most $2\vol(\wh\cB_1)$. 

\begin{align*}
    2\vol(\wh\cB_1)\le&2\vol\paren{\cB^n\paren{\sqrt{nP(1-\zeta^2)} +{\rcov} }}.
\end{align*}
The probability in Eqn. \eqref{eqn:lattice_pt_ortho} is hence at most
\begin{align*}
    &\frac{2\vol(\wh\cB_1)}{\vol(\wc\cB)}\\
    =&\frac{2\vol(\wh\cB_1)}{\vol\paren{\cB^n\paren{\sqrt{nP} - \rcov}}}\\
    \le&2\paren{\frac{\sqrt{nP(1-\zeta^2)} + \rcov}{\sqrt{nP} - \rcov}}^n \\
    =& 2\paren{\frac{\sqrt{P(1-\zeta^2)} + \sqrt{\omega}}{\sqrt{P} - \sqrt{\omega}}}^n .
\end{align*}
\end{proof}

\begin{lemma}\label{lem:xa_ip_xb}
If $\vbfx_A$ and $\vbfx_B$ are two  lattice points independently and uniformly chosen from $\cC$, then their inner product is  close to  0 w.h.p. For any constant $ \zeta\in(0,1) $,
\begin{align}
\probover{\vbfx_A,\vbfx_B\iid\cC}{\abs{\inprod{\vbfx_A}{\vbfx_B}}\ge nP\zeta}
\le&2\paren{\frac{\sqrt{P(1-\zeta^2)} + \sqrt{\omega}}{\sqrt{P} - \sqrt{\omega}}}^n.\notag
\end{align}
\end{lemma}
\begin{proof}
\begin{align}
\prob{\abs{\inprod{\vbfx_A}{\vbfx_B}}\ge nP\zeta}=&\prob{\normtwo{\vbfx_A}\cdot\normtwo{\vbfx_B}\cdot\abs{\cos\angle_{\vbfx_A,\vbfx_B}}\ge nP\zeta}\notag\\
\le&\prob{nP\cdot\abs{\cos\angle_{\vbfx_A,\vbfx_B}}\ge nP\zeta}\notag\\
=&\prob{\abs{\cos\angle_{\vbfx_A,\vbfx_B}}\ge \zeta}\notag\\
\le&2\paren{\frac{\sqrt{P(1-\zeta^2)} + \sqrt{\omega}}{\sqrt{P} - \sqrt{\omega}}}^n,\notag
\end{align}
where the last inequality is by Lemma \ref{lem:lattice_pt_orthogonal}.
\end{proof}

\begin{lemma}\label{lem:length_z}
Let $\vbfx_A,\vbfx_B$ be  random lattice points sampled uniformly and independently from $\cC$. Then $\vbfz\coloneqq\vbfx_A+\vbfx_B$ has norm approximately $\sqrt{2nP}$ w.h.p. 
For any constant $ \delta\in(0,1) $, $ \lambda\in(0,\delta) $, let $ \zeta\coloneqq\delta-\lambda $. Then
\begin{align*}
    \probover{\vbfx_A,\vbfx_B\iid\cC}{ \normtwo{\vbfz}\notin\sqrt{2nP(1\pm\delta)} }\le&
     2\paren{\frac{ \sqrt{P(1-\lambda)}+\sqrt{\omega} }{\sqrt{P} - \sqrt{\omega}}}^n + 2\paren{\frac{\sqrt{P(1-\zeta^2)} + \sqrt{\omega}}{\sqrt{P} - \sqrt{\omega}}}^n ,
\end{align*}
\end{lemma}
\begin{proof}
\begin{align*}
    &\prob{ \normtwo{\vbfz}\notin\sqrt{2nP(1\pm\delta)} }\\
    =&\prob{ \normtwo{\vbfx_A+\vbfx_B}\notin\sqrt{2nP(1\pm\delta)} }\\
    =&\prob{\normtwo{\vbfx_A}\le\sqrt{nP(1-\lambda)}\text{ or }\normtwo{\vbfx_B}\le\sqrt{nP(1-\lambda)}}\\
    &+\prob{ \normtwo{\vbfx_A+\vbfx_B}\notin\sqrt{2nP(1\pm\delta)},\;\normtwo{\vbfx_A}\in(\sqrt{nP(1-\lambda)},\sqrt{nP}],\;\normtwo{\vbfx_B}\in(\sqrt{nP(1-\lambda)},\sqrt{nP}] }.
\end{align*}
The first term, by Lemma  \ref{lem:lattice_pt_norm},  is at most
\begin{align}
    \prob{ \normtwo{\vbfx_A}\le\sqrt{nP(1-\lambda)} }+\prob{ \normtwo{\vbfx_B}\le\sqrt{nP(1-\lambda)} }  
    \le& 2\paren{\frac{ \sqrt{P(1-\lambda)}+\sqrt{\omega} }{\sqrt{P} - \sqrt{\omega}}}^n. \label{eqn:bound_z_1}
\end{align}
The second term is at most
\begin{align}
    &\prob{ \normtwo{\vbfx_A}^2+\normtwo{\vbfx_B}^2+2\inprod{\vbfx_A}{\vbfx_B}\notin{2nP(1\pm\delta)},\;\normtwo{\vbfx_A}^2\in({nP(1-\lambda)},{nP}],\;\normtwo{\vbfx_B}^2\in({nP(1-\lambda)},{nP}] }\notag\\
    \le& \prob{ 2nP + 2\inprod{\vbfx_A}{\vbfx_B} > 2nP(1+\delta)\text{ or }2nP(1-\zeta) + 2\inprod{\vbfx_A}{\vbfx_B}<2nP(1-\delta) } \notag \\
    \le&\prob{ \inprod{\vbfx_A}{\vbfx_B}>nP\delta\text{ or }\inprod{\vbfx_A}{\vbfx_B}<-nP(\delta-\lambda) }\notag \\
    \le& \prob{ \abs{\inprod{\vbfx_A}{\vbfx_B}}>nP(\delta-\lambda) } \notag \\
    \le&2\paren{\frac{\sqrt{P(1-\zeta^2)} + \sqrt{\omega}}{\sqrt{P} - \sqrt{\omega}}}^n.\label{eqn:bound_z_2}
\end{align}
Combining bounds \eqref{eqn:bound_z_1} and \eqref{eqn:bound_z_2} allows us to conclude
\begin{align}
\prob{ \normtwo{\vbfz}\notin\sqrt{2nP(1\pm\delta)} } \le& 2\paren{\frac{ \sqrt{P(1-\lambda)}+\sqrt{\omega} }{\sqrt{P} - \sqrt{\omega}}}^n + 2\paren{\frac{\sqrt{P(1-\zeta^2)} + \sqrt{\omega}}{\sqrt{P} - \sqrt{\omega}}}^n . \notag 
\end{align}
\end{proof}

\begin{corollary}\label{cor:bound_zprime}
Fix $\vs\in\cS^{n-1}\paren{\vzero,\sqrt{nN }}$. Let $ \vbfx_A $ and $ \vbfx_B $ be two random lattice points independently and uniformly sampled from $ \cC $. 
Let $ \vbfz\coloneqq\vbfx_A+\vbfx_B $ and $ \vbfsperp \coloneqq\proj_{\vbfz^\perp}(\vs) $. 
Then the norm  of $\vbfsperp $ is concentrated around $\sqrt{n(N-2\alpha^2P)}$ w.h.p. For any $ \delta\in(0,1) $,
\begin{align}
\probover{\vbfx_A,\vbfx_B\iid\cC}{ \normtwo{\vbfsperp }\notin\sqrt{ n(N - 2\alpha^2P(1\pm \delta)) }}\le&
2\paren{\frac{ \sqrt{P(1-\lambda)}+\sqrt{\omega} }{\sqrt{P} - \sqrt{\omega}}}^n + 2\paren{\frac{\sqrt{P(1-\zeta^2)} + \sqrt{\omega}}{\sqrt{P} - \sqrt{\omega}}}^n . \notag
\end{align}
\end{corollary}
\begin{proof}
\begin{align}
\probover{\vbfx_A,\vbfx_B\iid\cC}{ \normtwo{\vbfsperp }^2\notin n(N - 2\alpha^2P(1\pm \delta)) } =& \prob{ \normtwo{\vs}^2 - \alpha^2\normtwo{\vbfz}^2 \notin n(N-2\alpha^2P(1\pm\delta)) } \notag \\
=& \prob{ \normtwo{\vbfz}^2\notin2nP(1\pm\delta) } \notag \\
\le& 2\paren{\frac{ \sqrt{P(1-\lambda)}+\sqrt{\omega} }{\sqrt{P} - \sqrt{\omega}}}^n + 2\paren{\frac{\sqrt{P(1-\zeta^2)} + \sqrt{\omega}}{\sqrt{P} - \sqrt{\omega}}}^n . \notag 
\end{align}
\end{proof}

\subsection{Estimating $\alpha$}\label{sec:estimating_alpha}

\begin{lemma}\label{lem:many_confusing_pairs}
For any $\vz\in\cC+\cC$ such that $\normtwo{\vz}\in\sqrt{2nP(1\pm\delta)}$, there is a large number of pairs  $(\vx_A,\vx_B)\in\cC\times\cC$ which sum up to $\vz$,
\begin{align}
\card{\curbrkt{(\vx_A,\vx_B)\in\cC\times\cC\colon \vx_A+\vx_B = \vz}}\ge& C_1\cdot 2^{\frac{n}{2}\paren{ \log\paren{P/2-\comega} + \log\frac{1}{\tau} }} , \notag
\end{align}
where $ C_1 = C_1(P) $ and $ \comega  $ are positive constants to be defined later. In particular $ \comega\xrightarrow{\omega,\delta\to0}0 $.
\end{lemma}

\begin{remark}
For future convenience, 
let $ F_1  >0 $ be the largest constant such that
\begin{align}
2^{n(F_1-o(1))} \le& C_1\cdot 2^{\frac{n}{2}\paren{ \log\paren{P/2-\comega} + \log\frac{1}{\tau} }}. \notag
\end{align}
It suffices to take
\begin{align}
F_1 \coloneqq& \frac{1}{2}\log\paren{\frac{P}{2} - \comega} + \frac{1}{2}\log\frac{1}{\tau}. \notag
\end{align}
Note that $ F_1 \xrightarrow{\tau\to0}\infty $. 
\end{remark}

\begin{remark}
For readers who are familiar with the \emph{myopic} channel model \cite{zhang-quadratic-isit}, we would like to take this opportunity to point out that, as opposed to the myopic case where the uncertain codewords from James' perspective are \emph{approximately} uniformly distributed in his uncertainty set (which was named an \emph{oracle-given set} of thickness only $ \cO((\log n)/n) $), in our case the uncertain codewords are \emph{exactly} uniformly distributed on $\ufo\cap\Lf$ given James' observation $ \vbfz $. 
\end{remark}

\begin{proof}
First note that for each $ \vx_A\in\cC $, there is a unique $ \vx_B\in\cC $ such that $ \vx_A+\vx_B=\vz $. Indeed, such an $\vx_B $ is given by $ \vx_B = \vz - \vx_A $.
Hence, to count 
\begin{align}
{\curbrkt{(\vx_A,\vx_B)\in\cC\times\cC\colon \vx_A+\vx_B = \vz}} , \notag
\end{align}
it is equivalent to count 
\begin{align}
\curbrkt{ \vx_A\in\cC\colon \vz - \vx_A\in\cC } =& \curbrkt{ \vx_A\in\cC\colon \vx_A\in-\cC+\vz } \notag \\
=& \curbrkt{ \vx_A\in\cC\colon \vx_A\in\cC+\vz } \label{eqn:lattice_symm} \\
=& \curbrkt{ \vx_A\in\Lf\colon \vx_A\in\cB^n\paren{ \vzero,\sqrt{nP} }\cap\cB^n\paren{ \vz,\sqrt{nP} } }, \label{eqn:lattice_transl_inv}
\end{align}
where Equality \eqref{eqn:lattice_symm} follows from symmetry of lattice, i.e., $\vx\in\Lf$ iff $ -\vx\in\Lf $ for any $ \vx\in\bR^n $, and Equality \eqref{eqn:lattice_transl_inv} follows from translation invariance of lattice, i.e., $ \vx+\Lf = \Lf $ for any $ \vx\in\Lf $.
See Fig. \ref{fig:sumset}.
\begin{figure}[htbp]
    \centering
    \includegraphics[width = 0.5\textwidth]{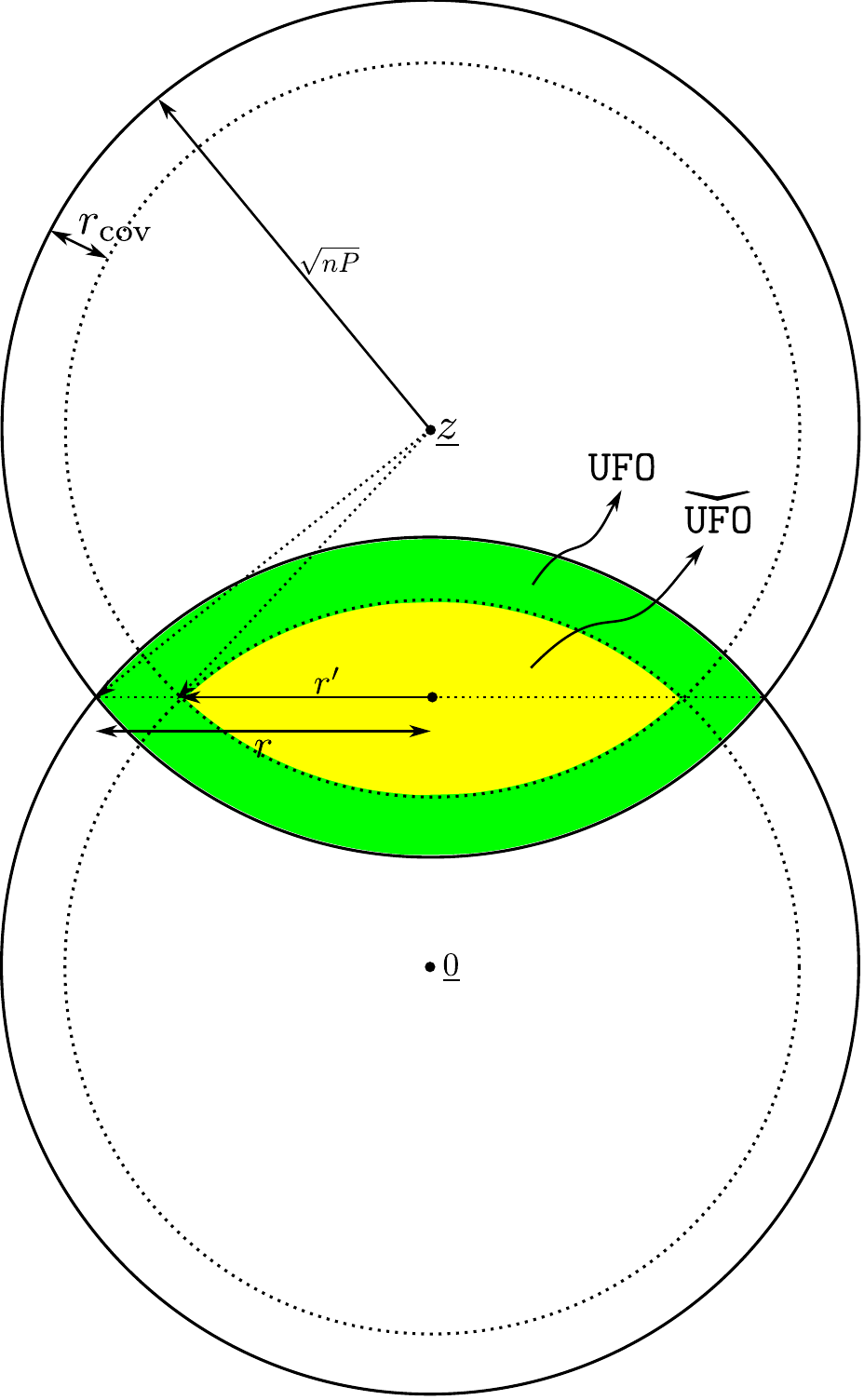}
    \caption{For any $ \vz\in\cC+\cC $, the  pairs of codewords $ (\vx_A,\vx_B)\in\cC\times\cC $ that sum up to $\vz$ lie within a $\ufo$ (the green and yellow regions). If $ \normtwo{\vz}\in\sqrt{2nP(1\pm\delta)} $, then the number of such pairs is exponentially large.}
    \label{fig:sumset}
\end{figure}
Define
\begin{align}
\ufo \coloneqq& \cB^n\paren{ \vzero,\sqrt{nP} }\cap\cB^n\paren{ \vz,\sqrt{nP} } . \notag
\end{align}
See Fig. \ref{fig:ufo} for a real UFO (if it exists).

The number of  $\vx_A$ in the $\ufo$ is at least
\begin{align}
	\frac{\vol\paren{\wc\ufo}}{\vol(\Lf)} 
    =&\frac{1}{\vol(\Lf)}\vol\paren{\wc{\cB^n\paren{\vzero,\sqrt{nP}}}\cap\wc{\cB^n\paren{\vz,\sqrt{nP}}}} \notag \\
    \ge&\frac{1}{\vol(\Lf)}\vol\paren{\cB^{n-1}\paren{r'}} \label{eqn:rp_def} \\
    \asymp&\frac{\sqrt{n}}{\sqrt{2\pi}\reff}\paren{\frac{\sqrt{(\sqrt{nP}-\rcov)^2 - \normtwo{\vz}^2/4}}{\reff}}^{n-1}  \notag \\
    \ge& \frac{\sqrt{n}}{\sqrt{2\pi}\reff}\paren{ \frac{\sqrt{ nP/2 + \rcov^2 - 2\rcov\sqrt{nP}- nP\delta/2 }}{\reff} }^{n-1} \notag \\
    =& \frac{1}{\sqrt{2\pi\tau}}\paren{ \frac{P/2+\omega-2\sqrt{P\omega}-P\delta/2}{\tau} }^{\frac{n-1}{2}} \label{eqn:lowerbound_ufo} \\
    =& C_1\cdot 2^{\frac{n}{2}\paren{ \log\paren{P/2-\comega} + \log\frac{1}{\tau} }} ,  \label{eqn:c1_comega_def} 
\end{align}
where we defined
\begin{align}
r'\coloneqq&\sqrt{\paren{\sqrt{nP}-\rcov}^2-\normtwo{\vz}^2/4}, 
\label{eqn:def_rp} \notag \\
C_1\coloneqq& \frac{1}{\sqrt{2\pi(P/2+\omega-2\sqrt{P\omega}-P\delta/2)}} \notag \\
\le& \frac{1}{\sqrt{7(P/2+1)}} , \notag \\
\comega\coloneqq&2\sqrt{P\omega}-\omega+P\delta/2 \notag
\end{align}
in Eqn. \eqref{eqn:rp_def} and \eqref{eqn:c1_comega_def}, respectively.
\end{proof}

\begin{lemma}\label{lem:many_cw_in_strip}
Fix $\vz\in\cC+\cC$ such that $\normtwo{\vz}\in\sqrt{2nP(1\pm\delta)}$. 
Among those pairs \[\curbrkt{(\vx_A,\vx_B)\in\cC\times\cC\colon \vx_A+\vx_B=\vz},\] most are in a thin strip $\cT$ (to be precisely defined in the proof) of radius approximately $\sqrt{nP/2}$ perpendicular to $\vz$,
\begin{align}
\frac{\card{\Lf\cap\cT}}{\card{\Lf\cap\ufo}}\ge& 1 - p(n)\paren{\frac{\frac{P}{2}(1+\delta)(1-\rho)+c_{3,\omega}}{P/2 - \comega}}^{\frac{n-1}{2}} - C_2\paren{\frac{P/2 - c_\eps + c_{4,\omega}}{P/2 - \comega}}^{\frac{n-1}{2}} , \notag
\end{align}
where 
\begin{align}
p(n) = p_{\eps,\omega}(n) \coloneqq& \sqrt{n\eps}+2\sqrt{n\omega}, \notag \\
C_2 = C_2(P) \le& 6\paren{\sqrt{P/2}+1}, \notag 
\end{align}
and $ c_\eps,c_{3,\omega} $ and $ c_{4,\omega} $ to be defined later satisfy $ c_\eps\xrightarrow{\eps\to0}0 $ and $ c_{3,\omega},c_{4,\omega}\xrightarrow{\omega\to0}0 $. 
\end{lemma}
\begin{remark}
For future convenience, take the largest constant $f_1>0 $ such that 
\begin{align}
2^{-n(f_1-o_n(1))}\ge& p(n)\paren{\frac{\frac{P}{2}(1+\delta)(1-\rho)+c_{3,\omega}}{P/2 - \comega}}^{\frac{n-1}{2}} + C_2\paren{\frac{P/2 - c_\eps + c_{4,\omega}}{P/2 - \comega}}^{\frac{n-1}{2}} . \notag
\end{align}
\end{remark}
\begin{proof}
We will show that the volume of $\ufo$ is concentrated around a thin strip $ \cT $ (to be defined momentarily) on the equator of $\ufo$, so are the lattice points therein.

We slice $\ufo$ into many layers each of height $\sqrt{n\eps}$ for some small constant $\eps>0$. Obviously, the  layer with the largest volume is the one in the middle, denoted by $\cF$. 
Formally $ \cF $ is defined as 
\begin{align}
\cF\coloneqq& \paren{ \frac{\vz}{2} + \curbrkt{ \vx\in\bR\colon \normtwo{\proj_{\vz}(\vx)}\le\sqrt{n\eps}/2 } } \cap \ufo \notag \\
=& \paren{ \frac{\vz}{2} + \curbrkt{ \vx\in\bR\colon \abs{\inprod{\vx}{\vz}} \le \normtwo{\vz} \sqrt{n\eps}/2 } } \cap \ufo . \notag
\end{align}
Note that the disk $\cF$ has radius  
\begin{align}
r\coloneqq&\sqrt{nP-\normtwo{\vz}^2/4},
\label{eqn:def_r}  \\
\in& \sqrt{\frac{nP}{2}(1\pm\delta)}. \notag
\end{align}
We further take a strip $\cT$ around the boundary of $ \cF $,
\begin{align*}
    \cT\coloneqq&\cF\setminus\paren{\cB^{n-1}\paren{\vzero,r\sqrt{1-\rho}}\times\bR},
\end{align*}
where $\cB^{n-1}\paren{\vzero,r\sqrt{1-\rho}}\times\bR$ denotes an infinitely high cylinder of radius $r\sqrt{1-\rho}$, centered around $\vzero$, along the direction of $\vz$. See Fig. \ref{fig:strip} for the construction of $\cF$ and $\cT$.
\begin{figure}[htbp]
    \centering
    \subfloat[Strip $\cT$ is defined as the thin pink band around the equator of $\ufo$. The blue region denotes a thin disk $ \cF $.\label{fig:strip}]{\includegraphics[width=0.45\textwidth]{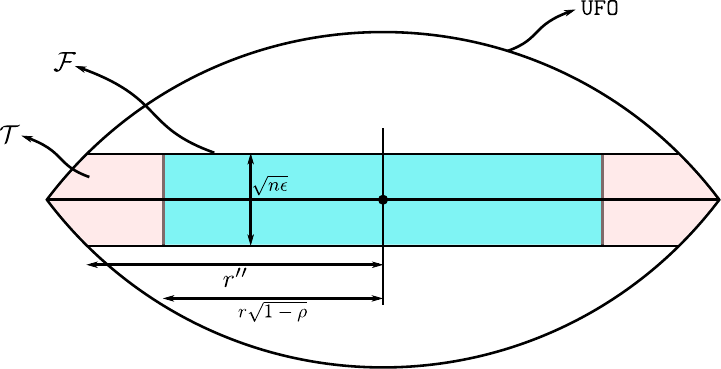}}
    \quad\quad\quad
    \subfloat[A real UFO (if exists).\label{fig:ufo}]{\includegraphics[width=0.45\textwidth]{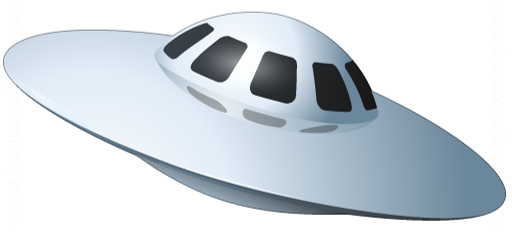}}
    \caption{Strip $\cT$ and UFO.}
    \label{fig:strip_ufo}
\end{figure}
We emphasize that there are two thickness parameters associated to $\cT$: $\rho$ --- (normalized) thickness along the radius of $\cF$; $\eps$ --- (normalized) thickness perpendicular to the radius of $\cF$. We want to show that there is a large (exponential) number of lattice points in $\cT$. 
To this end,  
it suffices to upper bound the number of lattice points in $\ufo$ but outside $\cT$.
Indeed, $\ufo\setminus \cT $ consists of two parts: 
\begin{align}
\cF\setminus\cT 
=& \cF\cap \paren{\cB^{n-1}\paren{\vzero,r\sqrt{1-\rho}}\times\bR}, \notag 
\end{align}
which is the blue disk in Fig. \ref{fig:strip}, and $ \ufo\setminus\cF $ which is the union of the upper and lower door of $\ufo$. 
We now upper bound the volumes of these two parts separately.
For the first part  $ \cF\setminus\cT $, we have
\begin{align}
	\vol\paren{ \wh{\cF\setminus\cT} }\le& \vol\paren{ \cB^{n-1}\paren{r\sqrt{1-\rho}+\rcov} }\paren{ \sqrt{n\eps}+2\rcov } \notag  \\
	=& V_{n-1} \paren{r\sqrt{1-\rho}+\rcov}^{n-1}\paren{\sqrt{n\eps}+2\rcov} \notag \\
	\le& V_{n-1} \paren{\sqrt{\frac{nP}{2}(1+\delta)(1-\rho)}+\sqrt{n\omega}}^{n-1}\paren{\sqrt{n\eps}+2\sqrt{n\omega}} . \label{eqn:upperbound_disk}
\end{align}
For the second part $ \ufo\setminus\cF $, we have
\begin{align}
\vol\paren{\wh{\ufo\setminus\cF}}\le& 2\cB^n\paren{ r''+\rcov } \notag \\
=& 2V_n \paren{r''+\rcov}^n . \label{eqn:upperbound_twocaps} 
\end{align}
where
\begin{align}
    r'' \coloneqq&\sqrt{\paren{\sqrt{nP} - \rcov}^2 - \paren{\frac{\normtwo{\vz}  }{2}+ \frac{\sqrt{n\eps}}{2}-\rcov}^2} .  \label{eqn:def_rpp}
\end{align}
Note that
\begin{align}
r''=& \sqrt{\paren{ nP - \frac{\normtwo{\vz}^2}{4} }
		- \paren{ \frac{n\eps}{4} + \frac{\normtwo{\vz}\sqrt{n\eps}}{2} } 
		- (2\sqrt{nP} - \normtwo{\vz} - \sqrt{n\eps})\rcov  } \notag \\
	=& \sqrt{ r^2 - n(c_\eps +c_\omega) } \notag \\
	\ge& \sqrt{ n \paren{ P/2 - c_1}} , \label{eqn:lowerbound_rpp}
\end{align}
where we define, for notational convenience,
\begin{align}
c_\eps \coloneqq& \frac{\eps}{4} + \frac{\normtwo{\vz}\sqrt{\eps}}{2\sqrt{n}}, \notag \\
\le& \frac{\eps}{4} + \sqrt{\frac{P(1+\delta)\eps}{2}} \notag \\
\le& \frac{\eps}{4} + \sqrt{P\eps} . \notag \\
c_\omega \coloneqq& (2\sqrt{P} - \normtwo{\vz}/\sqrt{n} - \sqrt{\eps})\sqrt{\omega}. \notag \\
\le& \paren{2\sqrt{P}-\sqrt{2P(1-\delta)}-\sqrt{\eps}}\sqrt{\omega} \notag \\
\le& 2\sqrt{P\omega}, \notag \\
c_1\coloneqq& P\delta/2+c_\eps+c_\omega. \notag
\end{align}

Combining  bounds \eqref{eqn:upperbound_disk}, \eqref{eqn:upperbound_twocaps}, \eqref{eqn:lowerbound_rpp} and \eqref{eqn:lowerbound_ufo}, we have that the probability that a uniformly lattice point in $\ufo$ falls outside $\cT$ is given by
\begin{align}
\probover{\vbfx\in\ufo}{\vbfx\notin\cT} =& \frac{\card{ \Lf\cap\paren{\cF\setminus\cT} }}{\card{\Lf\cap\ufo}} + \frac{\card{\Lf\cap\paren{\ufo\setminus\cF}}}{\card{\Lf\cap\ufo}} \notag \\
\le& \frac{\vol\paren{ \wh{\cF\setminus\cT} }}{\vol\paren{\wc{\ufo}}} + \frac{\vol\paren{ \wh{\ufo\setminus\cF} }}{\vol\paren{\wc{\ufo}}} \notag \\
\le& \frac{V_{n-1} \paren{\sqrt{\frac{nP}{2}(1+\delta)(1-\rho)}+\sqrt{n\omega}}^{n-1}\paren{\sqrt{n\eps}+2\sqrt{n\omega}}}{V_{n-1}\sqrt{P/2-\comega}^{n-1}}
+ \frac{2V_n\paren{\sqrt{n(P/2-c_1)}+\sqrt{n\omega}}^n}{V_{n-1}\sqrt{P/2-\comega}^{n-1}} \notag \\
\asymp& \paren{\sqrt{n\eps}+2\sqrt{n\omega}}\paren{\frac{\frac{P}{2}(1+\delta)(1-\rho) + \sqrt{2P(1+\delta)(1-\rho)\omega}+\omega}{P/2-\comega}}^{\frac{n-1}{2}} \label{eqn:bound_part1} \\
&+2\sqrt{2\pi}\paren{\sqrt{P/2-c_1}+\sqrt{\omega}}\paren{\frac{P/2-c_\eps-P\delta/2-c_\omega+2\sqrt{(P/2-c_1)\omega}+\omega}{P/2-\comega}}^{\frac{n-1}{2}} . \label{eqn:bound_part2} 
\end{align}
Observe that in Expression \eqref{eqn:bound_part1}, 
\begin{align}
c_{3,\omega}\coloneqq& \sqrt{2P(1+\delta)(1-\rho)\omega}+\omega \notag \\
\le& 2\sqrt{P\omega}+\omega \notag
\end{align}
in the numerator  vanishes as $ \omega\to0 $; 
in Expression \eqref{eqn:bound_part2}, 
\begin{align}
c_{4,\omega}\coloneqq&  -P\delta/2 - c_\omega + 2\sqrt{(P/2-c_1)\omega}+\omega \notag \\
\le& \sqrt{2P\omega}+\omega \notag
\end{align}
in the numerator  vanishes as $\omega\to0$. 
Also, $ \comega $ in the denominator of both Expression \eqref{eqn:bound_part1} and \eqref{eqn:bound_part2} vanishes as $ \omega,\delta\to0 $. 
By taking $ \rho,\eps\gg \omega,\delta $, we can make the bounds \eqref{eqn:bound_part1} and \eqref{eqn:bound_part2}   exponentially small in total.  This finishes the proof.
\end{proof}

\begin{lemma}\label{lem:bound_cos_in_strip}
Fix $\vz\in\cC+\cC$ such that $\normtwo{\vz}\in\sqrt{2nP(1\pm\delta)}$. 
All pairs of codewords $\vx_A$, $\vx_B$ in $\cT$ that sum up to $\vz$ are almost orthogonal,
\begin{align}
\cos(\angle_{\vx,\vx'})\in\sqrbrkt{ -\delta, \frac{(\rho+\delta)P - \rho P(1+\delta)/2}{(1-\rho)P+\rho P(1+\delta)/2} }. \notag 
\end{align} 
\end{lemma}
\begin{remark}
It is easy to verify that $ \frac{(\rho+\delta)P - \rho P(1+\delta)/2}{(1-\rho)P+\rho P(1+\delta)/2} > \delta $. 
For future convenience, define 
$
\theta \coloneqq \frac{(\rho+\delta)P - \rho P(1+\delta)/2}{(1-\rho)P+\rho P(1+\delta)/2}
$.
The above bound can be relaxed to $ \abs{ \cos\angle_{\vx,\vx'} } \le \theta $. 
\end{remark}

\begin{proof}
We will show that any two points (not necessarily in $\Lf$) $\vx,\vx'$ in $\cT$ are almost orthogonal.  Define
\begin{align}
\angle_{\cT}^{\max}\coloneqq&\max_{\substack{\vx,\vx'\in\cT\cap\Lf\\\vx+\vx'=\vz}}\angle_{\vx,\vx'},\quad\angle_{\cT}^{\min}\coloneqq\min_{\substack{\vx,\vx'\in\cT\cap\Lf\\\vx+\vx'=\vz}}\angle_{\vx,\vx'}.\notag
\end{align}
It turns out that, as depicted in Fig. \ref{fig:extremal_angles}, for any $\vx$ and $\vx'$ in $ \cT $ such that $ \vx+\vx' = \vz $,
\begin{align}
-\cos(\angle^{\min})\le&\cos(\angle_{\cT}^{\max})\le\cos(\angle_{\vx,\vx'})\le\cos(\angle_{\cT}^{\min})\le-\cos(\angle^{\max}), \notag
\end{align}
The extremal angles $ \angle^{\max} $ and $ \angle^{\min} $ are given by 
\begin{align}
    \cos(\angle^{\max})=&\frac{2\sqrbrkt{ \paren{r\sqrt{1-\rho}}^2 + \paren{\normtwo{\vz}/2}^2 } - \normtwo{\vz}^2}{ 2\sqrbrkt{\paren{r\sqrt{1-\rho}}^2 + \paren{\normtwo{\vz}/2}^2} } \notag \\
    =&1-\frac{\normtwo{\vz}^2}{2\sqrbrkt{r^2(1-\rho) + \normtwo{\vz}^2/4}} \notag \\
    =&1-\frac{\normtwo{\vz}^2}{ 2\sqrbrkt{ (1-\rho)\cdot nP + \rho\cdot\normtwo{\vz}^2/4 } } \label{eqn:monotone_fn_in_z} \\
    \in& \sqrbrkt{  1-\frac{P(1+\delta)}{(1-\rho)P+\rho P(1+\delta)/2}, 1-\frac{P(1-\delta)}{(1-\rho)P+\rho P(1-\delta)/2} } \label{eqn:monotonicity_z} \\
    =& \sqrbrkt{  \frac{\rho P(1+\delta)/2-(\rho+\delta)P}{(1-\rho)P+\rho P(1+\delta)/2}, \frac{\rho P(1-\delta)/2-(\rho-\delta)P}{(1-\rho)P+\rho P(1-\delta)/2} }, \label{eqn:cos_max_angle}
\end{align}
and 
\begin{align}
    \cos(\angle^{\min})=&1-\frac{\normtwo{\vz}^2}{2\sqrbrkt{r^2 + \normtwo{\vz}^2/4}} \notag \\
    =&1-\frac{\normtwo{\vz}^2}{2nP} \notag \\
    \in& [-\delta,\delta]. \label{eqn:cos_min_angle}
\end{align}
Eqn. \eqref{eqn:monotonicity_z} follows since Expression \eqref{eqn:monotone_fn_in_z} is a decreasing function in $ \normtwo{\vz} $. 

To see that $\angle^{\max}$ and $\angle^{\min}$ are indeed extremal angles, see Fig. \ref{fig:extremal_angles}. For $\vx$ and $\vx'$ in the strip $\cT$ (i.e., the light pink region in Fig. \ref{fig:extremal_angles}) which sum up to $\vz$, they form a triangle $\Delta OO'A$. We are interested in determining the maximum and minimum possible angles between such $\vx$ and $\vx'$. It is not hard to see that $\angle_{\vx,\vx'} = 180^\circ - \angle OAO'$. Hence $\cos(\angle_{\vx,\vx'}) = -\cos(\angle OAO')$ and it suffices to bound $\angle OAO'$. Let
\begin{align}
\angle^{\max}\coloneqq&\max_{A\in\cT}\angle OAO',\quad\angle^{\min}\coloneqq\min_{A\in\cT}\angle OAO'.\notag
\end{align}
For any apex $A$, all apexes $A'$s which are on the same circle determined by $O$, $O'$ and $A$ give the same angle $\angle OA'O' = \angle OAO'$. Hence, without loss of generality, we focus on $A$ which is on the radius of the strip $\cT$. Thereby, $\Delta OAO'$ is isosceles: $\overline{OA} = \overline{O'A}$. Now it is easy to see that $\angle^{\min}$ and $\angle^{\max}$ are given by {\color{red}$A$} and {\color{blue}$A$}, respectively in Fig. \ref{fig:extremal_angles}.
\begin{figure}[htbp]
    \centering
    \subfloat[ Extremal angles formed by vectors $ \vx,\vx'\in\cT $ such that $ \vx_A + \vx_B = \vz $ for some $ \vz $ with $ \normtwo{\vz}\in\sqrt{2nP(1\pm\delta)} $. The blue pair {\color{blue}$(\vx,\vx')$} forms the minimum angle $ \angle^{\min}_{\cT} $ and the red pair {\color{red}$(\vx,\vx')$} forms the maximum angle $ \angle^{\max}_\cT $. \label{fig:extremal_angles}]{\includegraphics[width=0.54\textwidth]{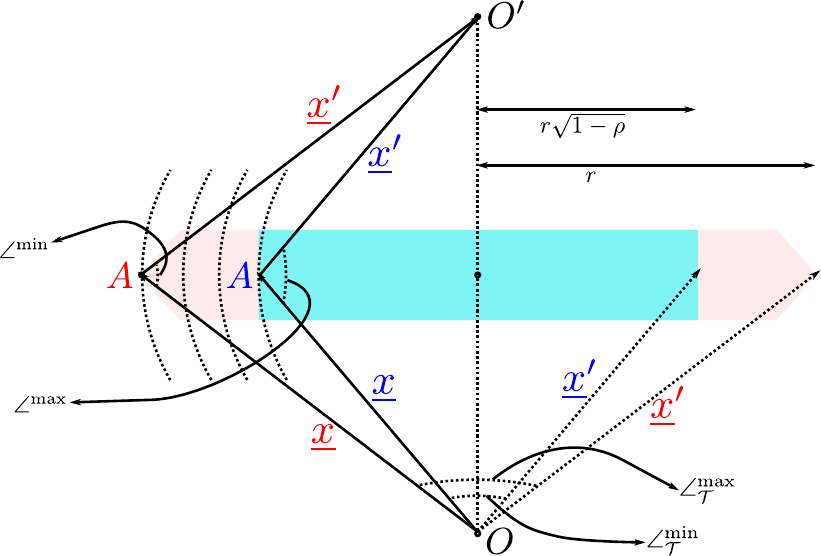}}
    \hfill
    \subfloat[ Extremal angles and extremal lengths in the strip  $\cT$ (the pink region). The blue region denotes $ \cF $. \label{fig:extremal_length}]{\includegraphics[width=0.43\textwidth]{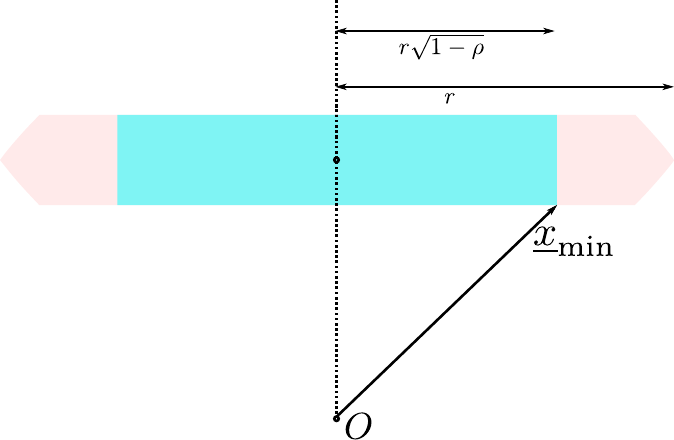}}
    \caption{Extremal angles and extremal lengths in the strip  $\cT$ (the pink region). The blue region denotes $ \cF $.}
    \label{fig:extremal_angle_len}
\end{figure}
\end{proof}

The following corollary follows directly from Lemma \ref{lem:bound_cos_in_strip}.
\begin{corollary}\label{cor:bound_ip_in_strip}
For $ \vz\in\cC+\cC $ such that $ \normtwo{\vz}\in\sqrt{2nP(1\pm\delta)} $, codewords $ \vx_A,\vx_B\in\cT $ such that $ \vx_A+\vx_B = \vz $ satisfy $ \abs{\inprod{\vx_A}{\vx_B}}\le nP\theta $. 
\end{corollary}

\begin{corollary}\label{cor:conc_norm_cw_in_strip}
Fix $\vz\in\cC+\cC$ such that $\normtwo{\vz}\in\sqrt{2nP(1\pm\delta)}$. 
Codewords $ \vx $ in $\cT$ have length close to $\sqrt{nP}$,
\begin{align}
\normtwo{ \vx }\in \sqrbrkt{ \sqrt{n(P-c_2)}, \sqrt{nP} }, \notag 
\end{align}
where $ c_2 = c_2(\rho,\eps) $ satisfies $ c_2\xrightarrow{\rho,\eps\to0}0 $. 
\end{corollary}
\begin{proof}
Define 
\begin{align}
\vx_{\min}\coloneqq& \argmin{\vx\in\cT}\normtwo{\vx}. \notag
\end{align}
From Fig. \ref{fig:extremal_length}, we get
\begin{align}
\normtwo{\vx_{\min}}^2 =& r^2(1-\rho) + \paren{ \frac{\normtwo{\vz}}{2} - \frac{\sqrt{n\eps}}{2} }^2 \notag \\
=& \paren{ nP - \frac{\normtwo{\vz}^2}{4} }(1-\rho) + \frac{ \normtwo{\vz}^2+n\eps - 2\normtwo{\vz}\sqrt{n\eps} }{4} \notag \\
=& nP - nc_2, \label{eqn:norm_x_in_strip}
\end{align}
where
\begin{align}
c_2\coloneqq& c_\rho + c_\eps', \notag \\
c_\rho \coloneqq& \paren{ P - \frac{\normtwo{\vz}^2}{4n} }\rho \notag \\
\le& \frac{P}{2}(1-\delta)\rho \notag \\
\le& \frac{P}{2}\rho , \notag \\
c_\eps' \coloneqq& \frac{\normtwo{\vz}\sqrt{\eps}}{2\sqrt{n}} - \frac{\eps}{4} \notag  \\
\le& \sqrt{\frac{P(1+\delta)}{8}\eps}-\frac{\eps}{4} \notag \\
\le& \frac{\sqrt{P\eps}}{2} . \notag
\end{align}
\end{proof}

\begin{lemma}\label{lem:xa_ip_zprime}
Fix $\vz\in\cC+\cC$ such that $\normtwo{\vz}\in\sqrt{2nP(1\pm\delta)}$. 
Fix $ \vs\in\cS^{n-1}\paren{\vzero, \sqrt{n N}} $. 
If $ \vbfx $ is uniformly distributed in $\cT\cap\Lf$, then it is approximately orthogonal to $ \vs_\perp  $ w.h.p., where $ \vs_\perp \coloneqq\proj_{\vz^\perp}(\vs) $,
\begin{align}
\probover{ \vbfx\sim\cT\cap\Lf }{ \abs{\inprod{\vbfx}{\vs_\perp }}\ge n\zeta }\le& \frac{2p(n)}{1-2^{-n(f_1-o(1))}}\paren{\frac{\frac{P}{2}(1+\delta)(1-\zeta'^2) + c_{3,\omega}}{\frac{P}{2} - \comega} }^{\frac{n-1}{2}} , \notag
\end{align}
where $ \zeta'\coloneqq\frac{\zeta}{PN } $. 
\end{lemma}
\begin{remark}
For future convenience, take the largest constant $ f_2>0 $ such that
\begin{align}
2^{-n(f_2-o(1))}\ge& \frac{2p(n)}{1-2^{-n(f_1-o(1))}}\paren{\frac{\frac{P}{2}(1+\delta)(1-\zeta'^2) + c_{3,\omega}}{\frac{P}{2} - \comega} }^{\frac{n-1}{2}}. \notag
\end{align}
\end{remark}
\begin{proof}
\begin{align}
    \probover{\vbfx\sim\cT\cap\Lf}{\abs{\inprod{\vbfx}{\vs_\perp }}\ge n\zeta}=&\frac{\card{\curbrkt{\vx\in\Lf\cap\cT\colon \abs{\inprod{\vx}{\vs_\perp }}\ge n\zeta}}}{\card{\Lf\cap\cT}}. \label{eqn:bound_pr_x_ip_zp}
\end{align}
The set $\cT'\coloneqq\cT\cap\curbrkt{\vx\in\bR^n\colon \abs{\inprod{\vx}{\vs_\perp }}\ge n\zeta}$ is the intersection of two halfspaces (that are symmetric around $\vz$)  and $\cT$. (In Fig. \ref{fig:geom_xa_ip_zprime}, the {pink} region represents $ \cT $ and the {red} subset of $ \cT $ represents $ \cT' $.) 
The above probability (Expression \eqref{eqn:bound_pr_x_ip_zp}) can be written as 
\begin{align}
\frac{\card{\cT'\cap\Lf}}{\card{\cT\cap\Lf}} =& \frac{\card{\cT'\cap\Lf}}{\card{\ufo\cap\Lf}\paren{1 - \frac{\card{(\ufo\setminus\cT)\cap\Lf}}{\card{\ufo\cap\Lf}}}} \notag \\
\ge& \frac{\vol(\wh\cT')}{\vol(\wc\ufo)(1 - 2^{-n(f_1 - o_n(1))})}. \notag
\end{align}

We already have a lower bound on $\vol(\wc\ufo)$. We now upper bound $\vol(\wh\cT')$.
\begin{align}
    \vol(\wh\cT')\le&2\cdot\cB^{n-1}\paren{r'''+\rcov}\cdot\paren{\sqrt{n\eps}+2\rcov} \notag\\
    =& 2V_{n-1} (\sqrt{n\eps}+2\sqrt{n\omega})(r'''+\sqrt{n\omega})^{n-1} , \label{eqn:vol_tpenlarge}
\end{align}
where $r'''$ is the radius of $\cT'$ given by
\begin{align}
    \frac{r'''}{r}=&\sqrt{1-\cos^2\paren{\max_{\vx\in\cT}\angle_{\vx,\vs_\perp }}} \notag \\
    =&\sqrt{1-\paren{\frac{n\zeta}{\normtwo{\vx}\normtwo{\vs_\perp }}}^2} \notag \\
    \le&\sqrt{1-\zeta'^2}, \label{eqn:def_rppp}
\end{align}
where $ \zeta'\coloneqq\frac{\zeta}{PN} $. 
See Fig. \ref{fig:geom_xa_ip_zprime} for the geometry behind the calculations. 
\begin{figure}[htbp]
    \centering
    \includegraphics[width = 0.5\textwidth]{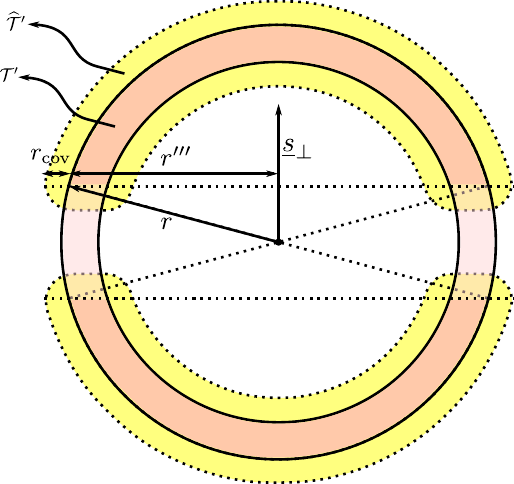}
    \caption{Most codewords in $ \cT $ are almost perpendicular to any $ \vsperp $ which is assumed to be orthogonal to $\vz$. }
    \label{fig:geom_xa_ip_zprime}
\end{figure}
Therefore, by Eqn. \eqref{eqn:lowerbound_ufo}, \eqref{eqn:def_rpp} and  \eqref{eqn:vol_tpenlarge}, \eqref{eqn:def_rppp},
\begin{align}
\frac{\vol(\wh\cT')}{\vol(\wc\ufo)(1 - 2^{-n(f_1 - o_n(1))})} \le& 
\frac{ 2V_{n-1}\paren{\sqrt{n\eps}+2\sqrt{n\omega}}\paren{r\sqrt{1-\zeta'^2}+\sqrt{n\omega}}^{n-1} }{V_{n-1}\sqrt{n\paren{P/2-\comega}}^{n-1}\paren{1-2^{-n(f_1-o(1))}}} \notag \\
=& \frac{2\paren{\sqrt{n\eps}+2\sqrt{n\omega}}}{1-2^{-(f_1-o(1))}}\paren{\frac{\frac{P}{2}(1+\delta)\paren{1-\zeta'^2}+2\sqrt{\frac{P}{2}(1+\delta)\paren{1-\zeta'^2}\omega}+\omega}{\frac{P}{2}-\comega}}^{\frac{n-1}{2}} 
\label{eqn:bound_ip_x_zp}  \\
\le& \frac{2(\sqrt{n\eps}+2\sqrt{n\omega})}{1-2^{-n(f_1-o(1))}}\paren{\frac{\frac{P}{2}(1+\delta)(1-\zeta'^2)+2\sqrt{P\omega}+\omega}{\frac{P}{2} - \comega}}^{\frac{n-1}{2}} \notag \\
=& \frac{2p(n)}{1-2^{-n(f_1-o(1))}}\paren{\frac{\frac{P}{2}(1+\delta)(1-\zeta'^2) + c_{3,\omega}}{\frac{P}{2} - \comega} }^{\frac{n-1}{2}}. \notag
\end{align}
Note that by taking $ \zeta\gg\delta,\omega $, the above bound is exponentially small. 
\end{proof}

\begin{lemma}\label{lem:estimate_alpha}
Fix $\vz\in\cC+\cC$ such that $\normtwo{\vz}\in\sqrt{2nP(1\pm\delta)}$. 
Fix $ \vs\in\cS^{n-1}\paren{\vzero, \sqrt{n N}} $. 
Assume $ \cE'^c $ holds.
Then $ \wh\alpha $ is a good estimate of $ \alpha $, i.e., $ \wh\alpha\in\alpha\pm\xi $, where  $ \xi\ge {\zeta+(1-\alpha)(P\theta+c_2)}/{P} $.
\end{lemma}
\begin{proof}
By definition of estimator $ \wh\alpha\coloneqq 1 - \frac{\inprod{\vbfy_B}{\vbfx_B}}{nP} $,
it suffices to show $ \inprod{\vbfy_B}{\vbfx_B}\approx(1-\alpha)nP $ w.h.p.
Note that
\begin{align}
    \inprod{\vbfy_B}{\vbfx_B}=&\inprod{(1-\alpha)\vbfx_A+(1-\alpha)\vbfx_B+\vbfsperp }{\vbfx_B} \notag \\
    =&(1-\alpha)\inprod{\vbfx_A}{\vbfx_B}+(1-\alpha)\normtwo{\vbfx_B}^2+\inprod{\vbfx_B}{\vbfsperp }. \label{eqn:ip_yb_xb}
\end{align}
By Corollary \ref{cor:bound_ip_in_strip}, if $ \vbfx_A $ and $ \vbfx_B $ that sum up to $ \vbfz $ fall into $ \cT $, then they are approximately orthogonal $ \abs{\inprod{\vbfx_A}{\vbfx_B}}\le nP\theta $. 
Moreover, by Corollary \ref{cor:conc_norm_cw_in_strip}, their norms are concentrated around $ \sqrt{nP} $, i.e.,  $ \normtwo{\vbfx_A}^2\in[n(P-c_2),nP] $ and $ \normtwo{\vbfx_B}^2\in[n(P-c_2),nP] $. 
Also, for any given $ \alpha $ such that $ \proj_{\vbfz}(\vs') = -\alpha\vbfz $ and the induced $ \vbfsperp  = \proj_{\vbfz^\perp}(\vs') $, we have $ \abs{\inprod{\vbfx_B}{\vbfsperp }}\le n\zeta $ w.h.p. by Lemma \ref{lem:xa_ip_zprime}. 


Now we are ready to bound the estimation error of $ \wh\alpha $.
\begin{align}
\curbrkt{ \wh\alpha\in\alpha\pm\xi } =& \curbrkt{ 1 - \frac{\inprod{\vbfy_B}{\vbfx_B}}{nP} \in \alpha\pm\xi } \notag \\
=& \curbrkt{ \inprod{\vbfy_B}{\vbfx_B} \in nP(1-\alpha\pm\xi) } \notag \\
=& \curbrkt{ (1-\alpha)\inprod{\vbfx_A}{\vbfx_B} + (1-\alpha)\normtwo{\vbfx_B}^2 + \inprod{\vbfx_B}{\vbfsperp } \in nP(1-\alpha\pm\xi) } \notag \\
\supset& \cE_1^c\cap\cE_2^c\cap\cE_3^c. \label{eqn:incl_e123} 
\end{align}
The last Inequality \eqref{eqn:incl_e123} follows since $ \cE_1^c\cap\cE_2^c\cap\cE_3^c $ implies 
\begin{align}
-(1-\alpha)nP\theta+(1-\alpha)n(P-c_2)-n\zeta\le&(1-\alpha)\inprod{\vbfx_A}{\vbfx_B} + (1-\alpha)\normtwo{\vbfx_B}^2 + \inprod{\vbfx_B}{\vbfsperp }\le (1-\alpha)nP\theta+(1-\alpha)nP+n\zeta. \notag 
\end{align}
By setting parameters properly, the above interval lies inside the interval $ nP(1-\alpha\pm\xi) $.
Indeed, set $\xi$ such that
\begin{align}
&\left\{\begin{array}{l}
-(1-\alpha)P\theta+(1-\alpha)(P-c_2)-\zeta \ge P(1-\alpha-\xi) \\
(1-\alpha)P\theta+(1-\alpha)P+\zeta \le P(1-\alpha+\xi)
\end{array},\right. \notag 
\end{align}
or
\begin{align}
&\left\{
\begin{array}{l}
\zeta\le P\xi-(1-\alpha)(P\theta+c_2) \\
\zeta\le P\xi-(1-\alpha)P\theta
\end{array}
.\right. \notag
\end{align}
It suffices to set
\begin{align}
\zeta\le& P\xi-(1-\alpha)(P\theta+c_2), \notag 
\end{align}
or 
\begin{align}
\xi\ge& \frac{\zeta+(1-\alpha)(P\theta+c_2)}{P}. \notag
\end{align}
\end{proof}

\subsection{Estimating effective decoding radius}\label{sec:estimating_effdecrad}
The analysis in the previous section implies that $\wh\alpha\coloneqq1-\frac{\inprod{\vbfy_B}{\vbfx_B}}{nP}$ is a good estimate to $\alpha$ used by James. 
It further implies that, from James' perspective, after Bob cancels his own (scaled) signal, Bob effectively receives
\begin{align*}
    \vbfy_B-(1-\wh\alpha)\vbfx_B=&\vbfy_B-\frac{\inprod{\vbfy_B}{\vbfx_B}}{nP}\vbfx_B,
\end{align*}
which is approximately equal to $(1-\alpha)\vbfx_A+\vbfsperp $ w.h.p., where $\wt\vbfx_A\coloneqq(1-\alpha)\vbfx_A$ is uniformly distributed in the strip scaled by $1-\alpha$. Such $\wt\vbfx_A$'s are translated by $\vbfsperp $ (which is perpendicular to $ \vbfz $) and Bob's effective received vector is $\wt\vbfy_B\coloneqq\wt\vbfx_A+\vbfsperp $. 
The geometry of the effective channel is shown in Fig. \ref{fig:eff_dec}.
In fact, the effective decoding radius $ \normtwo{\vbfsperp } $ can also be well estimated by Bob. Indeed, we have,
\begin{align}
    \normtwo{\vbfsperp }^2=&\normtwo{\paren{\vbfy_B-(1-\alpha)\vbfx_B} - (1-\alpha)\vbfx_A}^2 \notag \\
    =&\normtwo{\vbfy_B-(1-\alpha)\vbfx_B}^2 + (1-\alpha)^2\normtwo{\vbfx_A}^2 - 2\inprod{\vbfy_B-(1-\alpha)\vbfx_B}{(1-\alpha)\vbfx_A} \notag \\
    =&\normtwo{\vbfy_B}^2 + (1-\alpha)^2\normtwo{\vbfx_B}^2 - 2(1-\alpha)\inprod{\vbfy_B}{\vbfx_B} + (1-\alpha)^2\normtwo{\vbfx_A}^2 - 2\inprod{(1-\alpha)\vbfx_A+\vbfsperp }{(1-\alpha)\vbfx_A} \notag \\
    \approx&\normtwo{\vbfy_B}^2 + (1-\alpha)^2nP - 2(1-\alpha)\inprod{\vbfy_B}{\vbfx_B} + (1-\alpha)^2nP - 2(1-\alpha)^2nP + 0 \label{eqn:heuristic} \\
    =& \normtwo{\vbfy_B}^2 - 2(1-\alpha)\inprod{\vbfy_B}{\vbfx_B} ,  \notag 
\end{align}
where Eqn. \eqref{eqn:heuristic} heuristically holds w.h.p.
Hence we equip Bob with the following estimator for  $\normtwo{\vbfsperp }$,
\begin{align*}
    \normtwo{\wh\vbfsperp }^2=& \normtwo{\vbfy_B}^2 - 2(1-\wh\alpha)\inprod{\vbfy_B}{\vbfx_B} \\
    =&\normtwo{\vbfy_B}^2 - \frac{2\inprod{\vbfy_B}{\vbfx_B}^2}{nP}.
\end{align*}

\begin{figure}[htbp]
    \centering
    \subfloat[The effective channel to Bob is $ \wt\vbfy_B = \wt\vbfx_A+\vbfsperp $. The pink band denotes the scaled strip $ \wt\cT\coloneq(1-\alpha) \cT$.\label{fig:eff_dec}]{\includegraphics[height=0.45\textwidth]{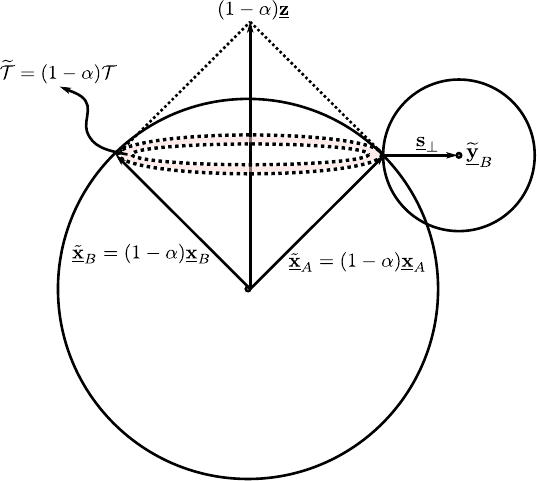}}
    \quad\quad
    \subfloat[The average (over uncertain codewords  $ \wt\vbfx_A\in \wt\cT$) effective decoding radius  of Bob can be computed from the geometry.\label{fig:geom_compute_avg_rad}]{\includegraphics[height=0.45\textwidth]{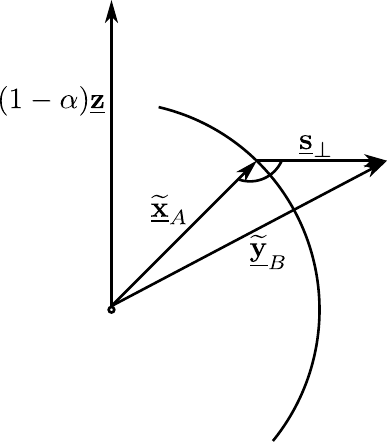}}
	\caption{The geometry of the effective channel and the geometry of computation of average effective decoding radius.}
	\label{fig:effdec_and_avgrad}
\end{figure}
Note that James does not have to use up all his power and thus $\normtwo{\vbfs}$ may be less than $\sqrt{nN}$. However, the worst case is when $\normtwo{\vbfs} = \sqrt{nN}$, which we assume is the case and suffices for upper bounding the decoding error probability.
We now bound the estimation error of $ \normtwo{\wh\vbfsperp } $. 
\begin{lemma}\label{lem:estimate_decrad}
Fix $\vz\in\cC+\cC$ such that $\normtwo{\vz}\in\sqrt{2nP(1\pm\delta)}$. 
Fix $ \vs\in\cS^{n-1}\paren{\vzero, \sqrt{n N}} $.  
Assume $ \cE'^c $ holds.
Then $ \normtwo{\wh\vbfsperp }^2\in\normtwo{\vbfsperp }^2\pm n\mu $, where
\begin{multline}
\mu\ge2\max\left\{ c_2(1-\alpha)^2+\frac{\zeta^2}{P} +2(1-\alpha)\zeta(1+\theta)+P(1-\alpha)^2\theta(3+\theta), \right.\notag \\
 \left. (1-\alpha)\paren{ -\frac{c_2^2(1-\alpha)}{P}+2c_2(1-\alpha)(1-\theta)+3P(1-\alpha)\theta+2\zeta(1+\theta) } \right\}. \notag
\end{multline}
\end{lemma}
\begin{proof}
By definition of the estimator, we have 
\begin{align}
\normtwo{\wh\vbfsperp }^2 =& \normtwo{\vbfy_B}^2 - \frac{2}{nP}\inprod{\vbfy_B}{\vbfx_B}^2 \notag \\
=&\normtwo{(1-\alpha)\vbfx_A+(1-\alpha)\vbfx_B+\vbfsperp }^2 - \frac{2}{nP}\inprod{(1-\alpha)\vbfx_A+(1-\alpha)\vbfx_B+\vbfsperp }{\vbfx_B}^2 \notag \\
=& \normtwo{\vbfsperp }^2 +  (1-\alpha)^2\paren{\normtwo{\vbfx_A}^2 + \normtwo{\vbfx_B}^2} + 2(1-\alpha)^2\inprod{\vbfx_A}{\vbfx_B} \notag \\
& -\frac{2}{nP} \left( (1-\alpha)^2\paren{\inprod{\vbfx_A}{\vbfx_B}^2 + \normtwo{\vbfx_B}^4 +2\inprod{\vbfx_A}{\vbfx_B}\normtwo{\vbfx_B}^2}  +2(1-\alpha)\paren{\inprod{\vbfx_A}{\vbfx_B}+\normtwo{\vbfx_B}^2}\inprod{\vbfx_B}{\vbfsperp } + \inprod{\vbfx_B}{\vbfsperp }^2 \right). \notag
\end{align}
To simplify notation in the following calculations, define $ B\coloneqq \inprod{\vbfx_B}{\vbfsperp }/n $. 
Then
we claim that
\begin{align}
\cE_1^c\cap\cE_2^c\cap\cE_3^c\subset\curbrkt{ \normtwo{\wh\vbfsperp }^2\in\normtwo{\vbfsperp }^2\pm n\mu }. \label{eqn:hatzz_close_to_zz}
\end{align}
Indeed, $ \cE_1^c\cap\cE_2^c\cap\cE_3^c $ implies 
\begin{align}
&\normtwo{\wh\vbfsperp }^2 - \normtwo{\vbfsperp }^2 \notag \\
\in& \left[ 2(1-\alpha)^2n(P-c_2)-2(1-\alpha)^2nP\theta-\frac{2}{nP}\paren{ (1-\alpha)^2\paren{(nP\theta)^2+(nP)^2+2nP\theta\cdot nP} + 2(1-\alpha)(nP\theta+nP)n\zeta+(n\zeta)^2 } , \right. \notag \\
&\left. 2(1-\alpha)^2nP+2(1-\alpha)^2nP\theta-\frac{2}{nP}\paren{ (1-\alpha)^2\paren{(n(P-c_2))^2-2nP\theta\cdot n(P-c_2)} - 2(1-\alpha)(nP\theta+nP)n\zeta } \right]. \notag
\end{align}
By taking proper values of parameters, the above interval is inside the interval $ [-n\mu,n\mu] $.
Indeed, we take $\mu$ such that
\begin{align}
\left\{
\begin{array}{l}
2(1-\alpha)^2(P-c_2)-2(1-\alpha)^2P\theta-\frac{2}{P}\paren{ (1-\alpha)^2(P^2\theta^2+P^2+2P^2\theta)+2(1-\alpha)(P\theta+P)\zeta+\zeta^2 }\ge-\mu \\
2(1-\alpha)^2P+2(1-\alpha)^2P\theta-\frac{2}{P}\paren{ (1-\alpha)^2((P-c_2)^2 - 2P\theta(P-c_2))-2(1-\alpha)(P\theta+P)\zeta} \le\mu
\end{array},
\right.
\end{align}
or
\begin{multline}
\mu\ge2\max\left\{ c_2(1-\alpha)^2+\frac{\zeta^2}{P} +2(1-\alpha)\zeta(1+\theta)+P(1-\alpha)^2\theta(3+\theta), \right.\notag \\
 \left. (1-\alpha)\paren{ -\frac{c_2^2(1-\alpha)}{P}+2c_2(1-\alpha)(1-\theta)+3P(1-\alpha)\theta+2\zeta(1+\theta) } \right\}. \notag
\end{multline}
\end{proof}
Note that 
\begin{align}
\ezz^c\cap\curbrkt{ \normtwo{\wh\vbfsperp }^2\in\normtwo{\vbfsperp }^2\pm n\mu }\subset& \curbrkt{ \normtwo{\wh\vbfsperp }^2\in n\paren{ N-2\alpha^2P(1\mp\delta)\pm\mu } }.\notag
\end{align}
Hence by Corollary \ref{cor:bound_zprime} and Lemma \ref{lem:estimate_decrad}, we immediately have the following corollary.
\begin{corollary}\label{cor:estimate_decrad}
Fix $\vz\in\cC+\cC$ such that $\normtwo{\vz}\in\sqrt{2nP(1\pm\delta)}$. 
Fix $ \vs\in\cS^{n-1}\paren{\vzero, \sqrt{n N}} $. 
Assume $ \cE'^c $ and $ \ezz^c $ hold.
Then James' estimate $ \normtwo{\wh\vbfsperp } $ is concentrated around the typical value $ \sqrt{n\wt N} $ of the correct decoding radius $ \normtwo{\vbfsperp } $, i.e., 
$
\normtwo{\wh\vbfsperp }\in\sqrt{n(N - 2\alpha^2P(1\mp\delta)\pm\mu)}
$.
\end{corollary}


\subsection{Setting the rate}\label{sec:settingrate}
Note that $ \vbfy_B - (1-\alpha)\vbfx_B = (1-\alpha)\vbfx_A + \vbfsperp $. Hence
the effective channel is essentially $ \wt \vbfy_B = \wt\vbfx_A+\vbfsperp  $, where $ \wt\vbfx_A\coloneqq(1-\alpha)\vbfx_A $. 
To decode, Bob computes 
\begin{align}
\vbfy_B-(1-\wh\alpha)\vbfx_B=&\vbfy_B-\frac{\inprod{\vbfy_B}{\vbfx_B}}{nP}\vbfx_B \notag \\
=& (1-\alpha)\vbfx_A+(1-\alpha)\vbfx_B+\vbfsperp  - \frac{1}{nP} \inprod{(1-\alpha)\vbfx_A+(1-\alpha)\vbfx_B+\vbfsperp }{\vbfx_B}\vbfx_B \notag \\
=& (1-\alpha)\vbfx_A+\vbfsperp  - \paren{ \frac{1-\alpha}{nP}\inprod{\vbfx_A}{\vbfx_B} + (1-\alpha)\paren{\frac{\normtwo{\vbfx_B}^2}{nP} - 1} + \frac{\inprod{\vbfsperp }{\vbfx_B}}{nP} }\vbfx_B. \notag
\end{align}
The error terms are bounded as follows.
\begin{align}
\normtwo{\frac{1-\alpha}{nP}\inprod{\vbfx_A}{\vbfx_B}\vbfx_B} \le& \frac{1-\alpha}{P}P\theta\sqrt{nP} \label{eqn:othererror1} \\
=& (1-\alpha)\theta\sqrt{P}\sqrt{n}, \notag \\
\normtwo{(1-\alpha)\paren{\frac{\normtwo{\vbfx_B}^2}{nP} - 1}\vbfx_B}\le& (1-\alpha)(c_2/P)\sqrt{nP} \label{eqn:othererror2} \\
=& \frac{(1-\alpha)c_2}{\sqrt{P}}\sqrt{n}, \notag \\
\normtwo{\frac{\inprod{\vbfsperp }{\vbfx_B}}{nP}\vbfx_B} \le& \frac{\zeta}{P}\sqrt{nP} \label{eqn:othererror3} \\
=& \frac{\zeta}{\sqrt{P}}\sqrt{n}. \notag
\end{align}
Bob scales $ \cC_A $ by $1-\wh\alpha$ and sets his (normalized) decoding radius to  
\begin{align}
{\wt N}' \coloneqq \paren{ \sqrt{N  - 2\alpha^2P(1-\delta)+\mu} + (1-\alpha)\theta\sqrt{P} + \frac{(1-\alpha)c_2}{\sqrt{P}} + \frac{\zeta}{\sqrt{P}} }^2. \label{eqn:listdec_n}
\end{align}
The power of $\wt\vbfx_A'\coloneqq(1-\wh\alpha)\vbfx_A $ is at least 
\begin{align}
 (1-\wh\alpha)^2(P-c_2)\ge (1-\alpha-\xi)^2(P-c_2) \eqcolon \wt P' . \label{eqn:listdec_p}
\end{align}

Define
\begin{align}
\beta_1\coloneqq& (1-\alpha)\theta\sqrt{P} + \frac{(1-\alpha)c_2}{\sqrt{P}} + \frac{\zeta}{\sqrt{P}} \notag \\
\le& \sqrt{P}\theta + \frac{c_2+\zeta}{\sqrt{P}}. \notag 
\end{align}
Note that since $ c_2\to0 $ as $ \rho,\eps \to0$, we get that $ \beta_1 $ vanishes as $ \theta,\rho,\eps $ and $ \zeta $ all approach 0. 
Let $ \cC_A $ operate at rate 
\begin{align}
R_A =& \frac{1}{2}\log\frac{\wt P'}{\wt N'} - \beta \notag \\
=& \frac{1}{2}\log\frac{(1-\alpha)^2P - \xi(2(1-\alpha)-\xi)P - c_2(1-\alpha-\xi)^2}{(N-2\alpha^2P)+2\alpha^2P\delta+\mu + \beta_1^2 + 2\sqrt{N  - 2\alpha^2P(1-\delta)+\mu}\beta_1 } - \beta \notag \\
\ge& \frac{1}{2}\log\frac{(1-\alpha)^2P}{N-2\alpha^2P} - \beta_2 - \beta_3 - \beta, \label{eqn:apply_log_ratio_ineq}
\end{align}
where Inequality \eqref{eqn:apply_log_ratio_ineq} follows from Corollary \ref{cor:log_ratio_ineq} by setting $ \eps $ and $ \delta $ in the corollary to
\begin{align}
\eps\leftarrow& \xi(2(1-\alpha)-\xi)P + c_2(1-\alpha-\xi)^2, \notag \\
\delta\leftarrow& 2\alpha^2P\delta+\mu + \beta_1^2 + 2\sqrt{N  - 2\alpha^2P(1-\delta)+\mu}\beta_1. \notag
\end{align}
In Inequality \eqref{eqn:apply_log_ratio_ineq}, we also defined 
\begin{align}
\beta_2 \coloneqq& 2\frac{\xi(2(1-\alpha)-\xi)P + c_2(1-\alpha-\xi)^2}{(1-\alpha)^2P} \notag \\
\le& \frac{2P\xi + 2c_2}{(1-\alpha)^2 P}  , \label{eqn:def_beta2} \\
\beta_3 \coloneqq& 2\frac{2\alpha^2P\delta+\mu + \beta_1^2 + 2\sqrt{N  - 2\alpha^2P(1-\delta)+\mu}\beta_1}{N-2\alpha^2P} \notag  \\
\le& 2\frac{2P\delta+\mu+\beta_1^2+2\sqrt{N +\mu}\beta_1}{N-2\alpha^2P}. \label{eqn:def_beta3}
\end{align}
Since $ c_2\xrightarrow{\rho,\eps\to0}0 $ and $ \beta_1\xrightarrow{\theta,\rho,\eps,\zeta}0 $, we have that $ \beta_2 $ vanishes as $ \xi,\rho $ and $ \eps $ approach 0, and $ \beta_3 $ vanishes as $ \delta,\mu,\theta,\rho,\eps $ and $\zeta$ all approach 0. 

For future convenience, let 
\begin{align}
C_\alpha\coloneqq& \frac{1}{2}\log\frac{(1-\alpha)^2P}{N-2\alpha^2P}. \notag
\end{align}

By the above configuration of parameters and by the choice of $\Lf$ in Sec. \ref{sec:code_design},  $ \Lf $, thereby $ \cC $, is $ \paren{ \wt P',\wt N', L } $-list decodable, where $ L = {2^{\cO\paren{\frac{1}{\beta}\log^2\frac{1}{\beta}}}} $.

\subsection{Computing average effective decoding radius}\label{sec:compute_avgdecrad}
In this section, we argue that, for any   $\vs\in\cS^{n-1}(\vzero,\sqrt{nN }) $, the radius  of the decoding region is concentrated around its typical value w.h.p. over James' uncertainty in the strip $ \cT $. 

Define random variable $ \bfr $ such that 
\begin{align}
\text{radius}\paren{ \cB^n\paren{ \wt\vbfy_B, \normtwo{\vbfsperp} }\cap\cB^n\paren{\vzero,\sqrt{n\wt P}} } = \sqrt{n\bfr}. \notag 
\end{align}
As shown in Fig. \ref{fig:geom_compute_avg_rad}, 
from the geometry, we have, on the one hand,
\begin{align*}
     \cos\paren{\angle_{\wt\vbfx_A,\wt\vbfy_B}} = &\frac{\inprod{\wt\vbfy_B}{\wt\vbfx_A}}{\normtwo{\wt\vbfy_B}\normtwo{\wt\vbfx_A}}\\
     =&\frac{\inprod{\wt\vbfx_A+\vbfsperp }{\wt\vbfx_A}}{\normtwo{\wt\vbfx_A+\vbfsperp }\normtwo{\wt\vbfx_A}};
\end{align*}
on the other hand,
\begin{align*}
    \sin\paren{\angle_{\wt\vbfx_A,\wt\vbfy_B}} =& \sqrt{n\bfr}/\normtwo{\wt\vbfx_A}.
\end{align*}
By $(\cos\theta)^2+(\sin\theta)^2 = 1$, we obtain
\begin{align}
    \bfr =& \frac{\normtwo{\wt\vbfx_A}^2}{n}\paren{ 1 - \frac{\paren{\normtwo{\wt\vbfx_A}^2+\inprod{\wt\vbfx_A}{\vbfsperp }}^2}{\paren{\normtwo{\wt\vbfx_A}^2+\normtwo{\vbfsperp }^2+2\inprod{\wt\vbfx_A}{\vbfsperp }}\normtwo{\wt\vbfx_A}^2} }. \notag
\end{align}
Note that heuristically, w.h.p. $ \bfr $ approximately equals
\begin{align}
\bfr \approx& {\wt P}\paren{ 1 - \frac{\paren{n\wt P+0}^2}{\paren{n\wt P+n\wt N+0}n\wt P} } \notag \\
=& {\wt P}\paren{1 - \frac{\wt P^2}{\paren{\wt P+\wt N}\wt P}}\notag \\
=& \frac{\wt P\wt N}{\wt P+\wt N}. \notag 
\end{align}

However,in reality, Bob does not have direct access to the parameters of the effective channel. 
From Bob's perspective, the input of the effective channel is $ \wt\vbfx_A' = (1-\wh\alpha)\vbfx_A $ of power $ \wh P' $ and the effective channel noise is $ \wh\vbfsperp' $ which is perpendicular to $ \vbfz $ of power $ \wt N' $. 
Let $ \wt\vbfy_B'\coloneqq\wt\vbfx_A'+\wh\vbfsperp' $. 
Define random variable $ \wh\bfr $ such that
\begin{align}
\text{radius}\paren{ \cB^n\paren{ \wt\vbfy_B', \normtwo{\wh\vbfsperp'} }\cap\cB^n\paren{\vzero,\sqrt{n\wt P'}} } = \sqrt{n\wh\bfr}, \notag 
\end{align}
which is a robust version of $ \bfr $ that takes estimation errors into account. 
We then argue that the above  channel parameters are close to the underlying typical values w.h.p. 

\begin{lemma}\label{lem:compute_avgdecrad}
Fix $\vz\in\cC+\cC$ such that $\normtwo{\vz}\in\sqrt{2nP(1\pm\delta)}$. 
Fix $ \vs\in\cS^{n-1}\paren{\vzero, \sqrt{n N}} $. 
Assume $ \cE'^c $ holds.
Then Bob's estimate of the (normalized) average effective decoding radius $ \wh\bfr $ is concentrated around the underlying typical value $ \frac{\wt P\wt N}{\wt P+\wt N} $, i.e., $ \wh\bfr\in \frac{\wt P\wt N}{\wt P+\wt N}\pm\nu $, where 
\begin{align}
\nu \ge \max\left\{  e_x + \frac{(\wt P+e'-e_x)^2}{\wt P+\wt N - \cdelta - e' + e_s - e_x - 2\mu} - \frac{\wt P^2}{\wt P+\wt N}, \frac{(\wt N+\cdelta+e_s)(\wt P - e_x)+3(\wt P-e_x)e' - e'^2}{\wt P+\wt N + \cdelta+e'+e_s-e_x} - \frac{\wt P\wt N}{\wt P+\wt N} \right\}.\notag
\end{align}
Here 
\begin{align}
\begin{array}{rlrl}
\cdelta\coloneqq&2\alpha^2P\delta+\mu\to0, & \text{as}&\delta,\mu\to0, \notag \\
e_x =& e_x(\rho,\eps,\xi)\to0, & \text{as}&\rho,\eps,\xi\to0, \notag \\
e_s =& e_s(\theta,\rho,\eps,\zeta,\mu)\to0, & \text{as}&\theta,\rho,\eps,\zeta,\mu\to0, \notag \\
e' =& e'(\zeta,\theta,\rho,\eps,\xi)\to0, & \text{as}&\zeta,\theta,\rho,\eps,\xi\to0. 
\end{array}
\notag 
\end{align}
\end{lemma}
\begin{proof}
Let $ \wh\vbfsperp'  = \vbfsperp+\es $ where $ \es $ is an estimation error vector. 
To bound the norm of $ \es $, note that 
on the one hand 
\begin{align}
\frac{1}{n}\abs{ \normtwo{\wh\vbfsperp'}^2 - \normtwo{\vbfsperp}^2 }\le& \frac{1}{n}\abs{ \normtwo{\wh\vbfsperp'}^2 - \normtwo{\wh\vbfsperp}^2 } + \frac{1}{n}\abs{ \normtwo{\wh\vbfsperp}^2 - \normtwo{\vbfsperp}^2 }\notag \\
\le& \paren{\beta_1^2 + 2\sqrt{N  - 2\alpha^2P(1-\delta)+\mu}\beta_1} + \mu \notag \\
\le& \beta_1^2 + 2\sqrt{N +\mu}\beta_1+\mu \notag \\
\eqcolon& e_s. \notag 
\end{align}
On the other hand, the largest possible difference between $ \normtwo{\vbfsperp+\es}^2 $ and $ \normtwo{\vbfsperp}^2 $ is
\begin{align}
\abs{ \normtwo{\vbfsperp + \es}^2 - \normtwo{\vbfsperp}^2 }=& \abs{\normtwo{ \vbfsperp }^2 + \normtwo{\es}^2 + 2\normtwo{\vbfsperp}\normtwo{\es} - \normtwo{\vbfsperp}^2} \notag \\
=& \normtwo{\es}^2+2\normtwo{\vbfsperp}\normtwo{\es}. \notag 
\end{align}
Therefore we have
\begin{align}
\normtwo{\es}^2+2\normtwo{\vbfsperp}\normtwo{\es}\le& e_s, \notag 
\end{align}
or $ \normtwo{\es}\le\sqrt{ne_s} $. 
Similarly, since 
\begin{align}
\frac{1}{n}\paren{ \normtwo{\wt\vbfx_A}^2 - \normtwo{\wt\vbfx_A'}^2 } =& {\wt P} - {\wt P'} \notag \\
=& (1-\alpha)^2P -  (1-\alpha-\xi)^2(P-c_2) \notag \\
=& c_2(1-\alpha-\xi)^2 + P\xi(2-2\alpha-\xi) \notag \\
\le& c_2 + 2P\xi \notag \\
\eqcolon& e_x, \notag 
\end{align}
if we write $ \wt\vbfx_A' = \wt\vbfx_A + \ex $, then $ \normtwo{\ex}\le\sqrt{ne_x} $. 

The average decoding radius computed w.r.t. Bob's estimated channel parameters is 
\begin{align}
{\wh\bfr} =& \frac{\normtwo{\wt\vbfx_A'}^2}{n}\paren{ 1 - \frac{\paren{\normtwo{\wt\vbfx_A'}^2+\inprod{\wt\vbfx_A'}{\wh\vbfsperp' }}^2}{\paren{\normtwo{\wt\vbfx_A'}^2+\normtwo{\wh\vbfsperp' }^2+2\inprod{\wt\vbfx_A'}{\wh\vbfsperp' }}\normtwo{\wt\vbfx_A'}^2} } \notag \\
=& \frac{\normtwo{\wt\vbfx_A ' }^2}{n}\paren{ 1 - \frac{\paren{\normtwo{\wt\vbfx_A' }^2+\inprod{\wt\vbfx_A }{\vbfsperp } + \inprod{\wt\vbfx_A }{\es} + \inprod{\ex}{\vbfsperp}+\inprod{\ex}{\es} }^2}{\paren{\normtwo{\wt\vbfx_A'}^2+\normtwo{\wh\vbfsperp '}^2+2\inprod{\wt\vbfx_A }{\vbfsperp }+2\inprod{\wt\vbfx_A }{\es}  + 2\inprod{\ex}{\vbfsperp}+2\inprod{\ex}{\es} }\normtwo{\wt\vbfx_A '}^2} }. \notag
\end{align}
By Cauchy--Schwarz inequality, 
\begin{align}
\abs{\inprod{\wt\vbfx_A}{\es}}\le& n\sqrt{(1-\alpha)Pe_s}\le n\sqrt{Pe_s}, \notag \\
\abs{\inprod{\ex}{\vbfsperp}}\le& n\sqrt{(N -2\alpha^2P(1-\delta)+\mu)e_x}\le n\sqrt{N e_x}, \notag \\
\abs{\inprod{\ex}{\es}}\le& n\sqrt{e_xe_s}. \notag
\end{align}
Therefore 
\begin{align}
{\wh\bfr}\in& {\wt P'}\paren{ 1 - \frac{\paren{n\wt P'\pm n\zeta \pm n \sqrt{P e_s}\pm n \sqrt{N  e_x} \pm n\sqrt{e_se_x}}^2}{ \paren{n\wt P' + \normtwo{\wh\vbfsperp}^2+n(e_s-\mu) \mp n\zeta \mp 2 n \sqrt{P e_s}\mp n \sqrt{N  e_x} \mp n\sqrt{e_se_x} } n\wt P' }  } \label{eqn:sphatp_vs_sphat} \\
\subseteq& {\wt P'}\paren{ 1 - \frac{\paren{n\wt P'\pm n\zeta \pm n \sqrt{P e_s}\pm n \sqrt{N  e_x} \pm n\sqrt{e_se_x}}^2}{ \paren{n\wt P' + \normtwo{\vbfsperp}^2+n(e_s-\mu\mp\mu) \mp n\zeta \mp 2 n \sqrt{P e_s}\mp n \sqrt{N  e_x} \mp n\sqrt{e_se_x} } n\wt P' }  } \label{eqn:sphat_vs_sp} \\
\subseteq& {\wt P'}\paren{ 1 - \frac{\paren{n\wt P'\pm n\zeta \pm n \sqrt{P e_s}\pm n \sqrt{N  e_x} \pm n\sqrt{e_se_x}}^2}{ \paren{n\wt P' + n\wt N+n(e_s-\mu\mp\mu\mp\cdelta) \mp n\zeta \mp 2 n \sqrt{P e_s}\mp n \sqrt{N  e_x} \mp n\sqrt{e_se_x} } n\wt P' }  } \label{eqn:sp_vs_nt} \\
=& \left[ \paren{\wt P-e_x}\paren{ 1 - \frac{\paren{\wt P - e_x + e'}^2}{ \paren{ \wt P + \wt N - e_x + e_s - 2\mu-\cdelta-e'}\paren{\wt P - e_x} } } , \right. \label{eqn:def_ep} \\
&\left. \paren{\wt P-e_x}\paren{ 1 - \frac{\paren{\wt P - e_x - e'}^2}{ \paren{ \wt P + \wt N - e_x + e_s +\cdelta+e'}\paren{\wt P - e_x} } } \right]. \notag
\end{align}
Eqn. \eqref{eqn:sphatp_vs_sphat} follows since 
\begin{align}
\frac{1}{n}\paren{\normtwo{\wh\vbfsperp'}^2 - \normtwo{\wh\vbfsperp}^2}\le& \beta_1^2 + 2\sqrt{N  - 2\alpha^2P(1-\delta)+\mu}\beta_1 \notag \\
=& e_s - \mu. \notag
\end{align}
Eqn. \eqref{eqn:sphat_vs_sp} follows since $ \normtwo{\wh\vbfsperp}^2 = \normtwo{\vbfsperp}^2\pm n\mu $. 
Eqn. \eqref{eqn:sp_vs_nt} follows since $ \normtwo{\vbfsperp} \in \sqrt{n\paren{\wt N\pm\cdelta}} $. 
In Eqn. \eqref{eqn:def_ep}, we defined $ e'\coloneqq\zeta+\sqrt{Pe_s}+\sqrt{N e_x}+\sqrt{e_se_x} $. 

We set $ \nu $ such that the above interval is a subinterval of $ \sqrbrkt{ \frac{\wt P\wt N}{\wt P+\wt N}-\nu,\frac{\wt P\wt N}{\wt P+\wt N}+\nu } $. 
Indeed it suffices to take
\begin{align}
\left\{
\begin{array}{l}
\frac{\wt P\wt N}{\wt P+\wt N}-\nu\le \paren{\wt P-e_x}\paren{ 1 - \frac{\paren{\wt P - e_x + e'}^2}{ \paren{ \wt P + \wt N - e_x + e_s - 2\mu-\cdelta-e'}\paren{\wt P - e_x} } } \\
\frac{\wt P\wt N}{\wt P+\wt N}+\nu \ge \paren{\wt P-e_x}\paren{ 1 - \frac{\paren{\wt P - e_x - e'}^2}{ \paren{ \wt P + \wt N - e_x + e_s +\cdelta+e'}\paren{\wt P - e_x} } }
\end{array},
\right.
\end{align}
or
\begin{align}
\nu \ge \max\left\{  e_x + \frac{(\wt P+e'-e_x)^2}{\wt P+\wt N - \cdelta - e' + e_s - e_x - 2\mu} - \frac{\wt P^2}{\wt P+\wt N}, \frac{(\wt N+\cdelta+e_s)(\wt P - e_x)+3(\wt P-e_x)e' - e'^2}{\wt P+\wt N + \cdelta+e'+e_s-e_x} - \frac{\wt P\wt N}{\wt P+\wt N} \right\}.\notag
\end{align}
\end{proof}

\subsection{Expurgation}
By now, all lemmas are proved w.r.t. $ \cC $ without expurgation.
All bounds are only over the randomness of message selection.
However, Lemma \ref{lem:expurgation} shows that, if the expurgation parameter $ \gamma $ and the packing/covering radius parameters $ \tau $ and $ \omega $ are sufficiently small, properties shown in previous sections continue to hold with probability doubly exponentially close to 1 over the expurgation process of $ \cC_A $ and $ \cC_B $.
Specifically, invoking Lemma \ref{lem:expurgation},  we have the following post-expurgation versions of the lemmas we have proved so far.
We state them without proof.
\begin{lemma}\label{eqn:postexpurgation}
Suppose that $ \gamma,\tau,\omega $ are all sufficiently small.
Then the following bounds hold. They are post-expurgation analogs of (pre-expurgation) bounds on probability (over message selection) of $ \elen,\einprod,\ez$, and $ \ezz $.
\begin{align}
&\probover{\cC_A}{\probover{\vbfx\sim\cC_A}{ \normtwo{\vbfx}\le\sqrt{nP(1-\zeta)} }>3\paren{\frac{ \sqrt{P(1-\zeta)}+\sqrt{\omega} }{\sqrt{P} - \sqrt{\omega}}}^n} 
\le 2^{-2^{\Omega(n)}} , \notag \\
&\probover{\cC_A,\cC_B}{\probover{\substack{\vbfx_A\sim\cC_A\\\vbfx_B\sim\cC_B}}{\abs{\inprod{\vbfx_A}{\vbfx_B}}\ge nP\zeta} >6\paren{\frac{\sqrt{nP(1-\zeta^2)} + \rcov}{\sqrt{nP} - \rcov}}^n } 
\le 2^{-2^{\Omega(n)}} , \notag \\
&\probover{\cC_A,\cC_B}{ \probover{\substack{\vbfx_A\sim\cC_A\\\vbfx_B\sim\cC_B}}{ \normtwo{\vbfz}\notin\sqrt{2nP(1\pm\delta)} }> 6\paren{\frac{ \sqrt{P(1-\lambda)}+\sqrt{\omega} }{\sqrt{P} - \sqrt{\omega}}}^n + 6\paren{\frac{\sqrt{P(1-\zeta^2)} + \sqrt{\omega}}{\sqrt{P} - \sqrt{\omega}}}^n  }
\le 2^{-2^{\Omega(n)}} , \notag \\
&\probover{\cC_A,\cC_B}{ \probover{\substack{\vbfx_A\sim\cC_A\\\vbfx_B\sim\cC_B}}{ \normtwo{\vbfsperp }\notin\sqrt{ n(N - 2\alpha^2P(1\pm \delta)) }}>
6\paren{\frac{ \sqrt{P(1-\lambda)}+\sqrt{\omega} }{\sqrt{P} - \sqrt{\omega}}}^n + 6\paren{\frac{\sqrt{P(1-\zeta^2)} + \sqrt{\omega}}{\sqrt{P} - \sqrt{\omega}}}^n  }
\le 2^{-2^{\Omega(n)}}. \notag 
\end{align}
Fix $\vz\in\cC_A+\cC_B $ such that $\normtwo{\vz}\in\sqrt{2nP(1\pm\delta)}$. 
Fix $ \vs\in\cS^{n-1}\paren{\vzero, \sqrt{n N}} $. 
Then the following bounds hold. They are post-expurgation analogs of bounds on $ \esumset,\cE_\cT $ and $ \cE_3 $. 
Events $ \cE_1^c$ and $\cE_2^c $ are geometric consequences of the  construction of the strip $\cT$ and will not be affected by expurgation. 
\begin{align}
&\probover{ \cC_A,\cC_B }{ \card{\curbrkt{(\vbfx_A,\vbfx_B)\in\cC_A\times\cC_B\colon \vbfx_A+\vbfx_B = \vz}}<\frac{1}{2}\cdot 2^{n(F_1-2\gamma-o(1))} }
\le 2^{-2^{\Omega(n)}} , \notag \\
&\probover{\cC_A,\cC_B}{ { \frac{\card{\curbrkt{(\vbfx_A,\vbfx_B)\in\cC_A\times\cC_B\colon\vbfx_A+\vbfx_B = \vz,\;\vbfx_A\in\cT,\;\vbfx_B\in\cT}}}{\card{\curbrkt{(\vbfx_A,\vbfx_B)\in\cC_A\times\cC_B\colon\vbfx_A+\vbfx_B = \vz}}} } > 3\cdot 2^{-n(f_1 - o(1))} }
\le 2^{-2^{\Omega(n)}} , \notag \\
&\probover{\cC_A}{ \probover{ \vbfx_A\sim\cC_A\cap\cT }{ \abs{\inprod{\vbfx_A}{\vs_\perp }}\ge n\zeta }>3\cdot2^{-n(f_2-o(1))} }
\le 2^{-2^{\Omega(n)}} . \notag 
\end{align}
Events $\ealpha^c,\edecrad^c $ and $ \eavgrad^c $ follow from $ \cE'^c $ and $ \ezz^c $ and will be not affected by expurgation as long as $ \cE'^c $ and $ \ezz^c $ hold after expurgation. 
\end{lemma}

\subsection{Bounding  probability of error}\label{sec:bouding_pe}
Let $\Tgood$ ($\Tbad \coloneqq \cT\setminus\Tgood $) denote the subset of $\cT$ in which codewords induce typical (atypical) radii of decoding regions under $\wh\vbfsperp '$ assuming these codewords were transmitted.
The probability that the transmitted $\vbfx_A$ falls into $\Tbad$ is exponentially small. For those $\vbfx_A$ in $\Tgood$, by list decodability, the number of codewords i n  balls centered around any $\vbfx_A\in\Tgood$ of radius $\sqrt{n\wt N'}$ is at most $L$. 
After expurgation with probability $1 - 2^{-\gamma n}$, in expectation, the number of codewords in the decoding ball is at most $L2^{-\gamma n}$. 
To get doubly exponential concentration (which admits a union bound over $ \vs'\in\cS $), we invoke McDiarmid's inequality and show that with  probability $ 1-2^{-\Omega(2^n)} $ over expurgation, the fraction of codewords in $ \Tgood $ that suffer decoding errors (i.e., there exists another codeword in the decoding ball) is exponentially small, or, in $ 1-2^{-\Omega(n)} $ fraction of decoding balls induced by codewords in $ \Tgood $, there will be no codeword other than the transmitted one  that survived the expurgation. 
The analysis is similar to that in \cite{jaggi-langberg-2017-two-way}.

For any vector $ \vx $, define $ \wt\vx\coloneqq(1-\wh\alpha)\vx $. 
For any set $ \cV $, let $ \wt\cV\coloneqq(1-\wh\alpha)\cV $.
Let $ \curbrkt{\vx_i}_{i = 1}^M $  denote $ \cC $. 

\begin{lemma}\label{lem:bounding_pe}
Fix $\vz\in\cC+\cC$ such that $\normtwo{\vz}\in\sqrt{2nP(1\pm\delta)}$. 
Fix $ \vs\in\cS^{n-1}\paren{\vzero, \sqrt{n N}} $. 
Then the fraction of codewords codewords in $ \Tgood $ that may suffer decoding errors is exponentially small with probability doubly exponentially close to 1 over expurgation,
\begin{align}
&\probover{\cC_A,\cC_B}{ \probover{ \wt\vbfx_A'\sim\wt\Tgood\cap\wt\cC_A }{ \wt\vbfx_A'\text{ suffers a decoding error} } > 3L2^{-2n\gamma} } \notag \\
\le& \exp\paren{ -\frac{(1-1/L)^2}{4(1+(L-1)2^{-n\gamma})^2}2^{n(2F_1 - C_{\alpha} - 6\gamma +\beta_2 + \beta_3 + \beta - o(1))} } + \exp\paren{-\frac{1}{12}2^{n(F_1-\gamma-o(1))}}, \notag
\end{align}
where the outer expectation is taken over expurgation and the inner one is taken over uniform distribution on $ \wt\Tgood\cap\wt\cC_A $. 
\end{lemma}

\begin{proof}
For $\vz\in\cC+\cC$ such that $ \normtwo{\vz}\in\sqrt{2nP(1\pm\delta)} $ and $\wh\vsperp'\in\cB^n\paren{\vzero,\sqrt{n\wt N'}} $, consider a directed graph $\cG(\cC,\cC_A,\cC_B,\vz,\wh\vsperp')$ with vertices $\cV=\cC$. There is an edge $\vx_A\to\vx_A'$ for $\vx_A\ne\vx_A'$ iff 
\begin{enumerate}
    \item $\vx_A\in\Tgood\cap\cC_A$ ($\vx_A$ is not expurgated in the construction of $\cC_A$);
    \item 
    $\vx_B \coloneqq \vz - \vx_A\in\cC_B$ ($\vx_B$ is not expurgated in the construction of $\cC_B$); (Note that $ \vx_B $ is guaranteed to be inside $ \Tgood $ if $ \vx_A\in\Tgood $.)
    \item $\vx_A'\in\cC_A$ and $ \normtwo{\wt\vy_B - \wt\vx_A'} = \normtwo{\wt\vx_A+\wh\vsperp'-\wt\vx_A'}\le\sqrt{n\wt N'}$ (there is an  $\wt\vx_A'\in\wt\cC_A$ confusable with $\wt\vx_A$).
\end{enumerate}
For $ \vbfx_A $ uniformly distributed in $ \Tgood\cap\cC_A $, the probability that it incurs a decoding error is given by the following ratio 
\begin{align*}
    \frac{\card{\curbrkt{\vx_A\in\Tgood\cap\cC_A\text{ that suffers a decoding error}}}}{\card{\Tgood\cap\cC_A}}\le&\frac{f(\cG)}{\card{\Tgood\cap\cC_A}},
\end{align*}
where 
\[f(\cG)\coloneqq\card{\curbrkt{\vx_A\in \cG\colon \outdeg(\vx_A)>0}}.\]
We first bound the denominator. Before expurgation, by Lemma 
\ref{lem:many_confusing_pairs}, \ref{lem:many_cw_in_strip} and 
\ref{lem:xa_ip_zprime}, we have
\begin{align}
\card{\Tgood\cap\cC} =& \card{ \cT\cap\cC }\paren{ 1 - 2^{-n(f_2-o(1))} } \notag  \\
\ge& \card{\ufo\cap\cC}\paren{ 1 - 2^{-n(f_1-o(1))} }\paren{ 1 - 2^{-n(f_2-o(1))} } \notag \\
\ge& 2^{n(F_1-o(1))} \paren{ 1 - 2^{-n(f_1-o(1))} }\paren{ 1 - 2^{-n(f_2-o(1))} } \label{eqn:size_tgood_intersection_c}  \\
\doteq& 2^{nF_1}. \notag
\end{align}
Since 
\begin{align}
\exptover{\cC_A}{\card{\Tgood\cap\cC_A}} =& \card{\Tgood\cap\cC}2^{-n\gamma}, \notag
\end{align}
by Chernoff bound (Corollary \ref{cor:expurgation}), we have
\begin{align}
\probover{\cC_A}{ \card{ \Tgood\cap\cC_A }\le\frac{1}{2}\card{\Tgood\cap\cC}2^{-n\gamma} }\le&  \exp\paren{ - \frac{1}{12}\card{\Tgood\cap\cC}2^{-n\gamma} } \notag \\
\le& \exp\paren{-\frac{1}{12}2^{n(F_1-\gamma-o(1))}}. \label{eqn:conc_tgood_intersection_c}
\end{align}

We then bound $ f(\cG) $. To this end, let us compute the expected value of $ f(\cG) $.
Note that 
\begin{align}
f(\cG) =& \sum_{\vx_A\in\Tgood\cap\cC}\indicator{\vx_A\in\cC_A} \indicator{\vz-\vx_A\in\cC_B} \indicator{ \exists\wt\vx_A'\ne\wt\vx_A,\;\wt\vx_A'\in\cB^n\paren{ \wt\vx_A + \wh\vsperp' ,\sqrt{n\wt N'} }\cap\wt\cC_A }. \notag
\end{align}
Now
\begin{align}
\expt{f(\cG)} =& \sum_{\vx_A\in\Tgood\cap\cC} \probover{\cC_A,\cC_B}{ \curbrkt{\vx_A\in\cC_A}\cap\curbrkt{\vz - \vx_A\in\cC_B}\cap\curbrkt{\exists\wt\vx_A'\ne\wt\vx_A,\;\wt\vx_A'\in\cB^n\paren{ \wt\vx_A + \wh\vsperp',\sqrt{n\wt N'} }\cap\wt\cC_A} } \label{eqn:lin_expt} \\
=& \sum_{\vx_A\in\Tgood\cap\cC}\probover{\cC_A}{\vx_A\in\cC_A} \probover{\cC_B}{\vz - \vx_A\in\cC_B} \probover{\cC_A}{ \exists\wt\vx_A'\ne\wt\vx_A,\;\wt\vx_A'\in\cB^n\paren{ \wt\vx_A + \wh\vsperp',\sqrt{n\wt N'} }\cap\wt\cC_A } \label{eqn:indep} \\
=& \card{ \Tgood\cap\cC }2^{-n\gamma}2^{-n\gamma}\paren{ 1 - \paren{1 - 2^{-n\gamma}}^{L-1} } \notag \\
\le& \card{\Tgood\cap\cC}2^{-2n\gamma}\paren{1 - \paren{1 - L2^{-n\gamma}}} \label{eqn:apply_exp_ineq} \\
=& L\card{\Tgood\cap\cC}2^{-3n\gamma}, \label{eqn:ub_on_exptf} 
\end{align}
where Equality \eqref{eqn:lin_expt} is by linearity of expectation, Equality \eqref{eqn:indep} follows since $ \cC_A $ and $ \cC_B $ are obtained by independent expurgation and each codeword is expurgated independently. 
Inequality \eqref{eqn:apply_exp_ineq} is by Fact \ref{fact:exp_ineq}. 

Using Fact \ref{fact:exp_ineq}, we can also get a lower bound on $ \expt{f(\cG)} $,
\begin{align}
\expt{f(\cG)}=& \card{ \Tgood\cap\cC }2^{-2n\gamma}\paren{ 1 - \paren{1 - 2^{-n\gamma}}^{L-1} } \notag \\
\ge& \card{ \Tgood\cap\cC }2^{-2n\gamma} \paren{ 1 - \frac{1}{1+(L - 1)2^{-n\gamma}} } \notag \\
=& (L-1)\card{ \Tgood\cap\cC } \frac{2^{-3n\gamma}}{1+(L-1)2^{-n\gamma}}. \label{eqn:lb_on_exptf}
\end{align}

We next  argue that $f$ is actually Lipschitz. Think the expurgation process as picking each codeword in $ \cC $ independently into $\cC_A$ and $\cC_B$ with probability $ 2^{-n\gamma} $. 
Define, for $i\in[M]$,
\begin{align}
X_i\coloneqq&\indicator{\vx_i\in\cC_A}\sim\bern\paren{2^{-n\gamma}}, \notag \\
Y_i\coloneqq&\indicator{\vx_i\in\cC_B}\sim\bern\paren{2^{-n\gamma}}. \notag 
\end{align}
Note that all $X_i$'s and $Y_i$'s are independent. Now $f$ can be written as
\begin{align}
f(X_1,\cdots,X_{M}, Y_1,\cdots,Y_M) =& \sum_{\substack{i\in[M]\\{\vx}_i\in\cT}}\indicator{X_i = 1} \indicator{Y_{j_i} = 1} \indicator{\exists i'\ne i,\;\wt\vx_{i'}\in\cB^n\paren{\wt\vx_i+ \wh\vsperp',\sqrt{n\wt N'}},\; X_{i'} = 1} \label{eqn:def_ji} \\
=& \sum_{\substack{i\in[M]\\{\vx}_i\in\cT}}\indicator{X_i = 1} \indicator{Y_{j_i} = 1} \indicator{\bigcup_{\substack{i'\ne i\\ {\wt\vx}_{i'}\in\cB^n\paren{\wt\vx_i + \wh\vsperp',\sqrt{n\wt N'}} }} \curbrkt{X_{i'} = 1}} \notag \\
=& \sum_{\substack{i\in[M]\\{\vx}_i\in\cT}} X_i\AND Y_{j_i}\AND \paren{\mathop{\OR}_{\substack{i'\ne i\\ {\wt\vx}_{i'}\in\cB^n\paren{\wt\vx_i,\sqrt{n\wt N'}}}} X_{i'}}, \label{eqn:or_mod2} 
\end{align}
where in Equality \eqref{eqn:def_ji},  $j_i\in[M]$ is such that $\vx_{j_i}=\vz-\vx_i$, and in Equality \eqref{eqn:or_mod2}, $\AND$ and $ \OR $  are taken over $\bF_2$, but the summation is still taken over $ \bZ $ as usual. For any $i$, if we flip $X_i$, $f$ can change by at most
\begin{align}
    \abs{f(X_1,\cdots,X_i = 0,\cdots,X_M, Y_1,\cdots,Y_M)-f(X_1,\cdots,X_i = 1,\cdots,X_M, Y_1,\cdots,Y_M)}\le&L, \notag
\end{align}
since $\vx_i$ can lie in the lists of radius $ \sqrt{n\wt N'} $ of at most $L$ codewords, corresponding to the third factor of the summand of Eqn. \eqref{eqn:or_mod2}. 
For any $i$, if we flip $Y_i$, $f$ can change by at most
\begin{align*}
    \abs{f(X_1,\cdots,X_M,Y_1,\cdots,Y_i = 0,\cdots,Y_M)-f(X_1,\cdots,X_M,Y_1,\cdots,Y_i = 1,\cdots, Y_M)}\le&1
\end{align*}
since it only appears as the second factor in the summand of Eqn. \eqref{eqn:or_mod2}.
Therefore, $ \lip(f)\le L $. 

Now we can apply McDiarmid's inequality (Lemma \ref{lem:mcdiarmid_ineq}) to get a doubly exponential concentration bound on  $f$.
\begin{align}
	\probover{\cC_A,\cC_B}{ f(\cG)\ge \frac{3}{2}L\card{\Tgood\cap\cC}2^{-3n\gamma} } 
	\le& \probover{\cC_A,\cC_B}{ f(\cG)\ge\frac{3}{2}\expt{f(\cG)} } \label{eqn:apply_ub_on_exptf} \\
    \le&\exp\paren{ -\frac{ 2(1/2)^2 \expt{f}^2}{2ML^2} } \notag \\
    \le& \exp\paren{ -\frac{\paren{(L-1)\card{ \Tgood\cap\cC } \frac{2^{-3n\gamma}}{1+(L-1)2^{-n\gamma}}}^2}{4\cdot2^{n(C_{\alpha} - \beta_2 - \beta_3 - \beta)}L^2} } \label{eqn:apply_lb_on_exptf} \\
    =& \exp\paren{ -\frac{(L-1)^2\card{\Tgood\cap\cC}^22^{-6n\gamma}}{4L^2(1+(L-1)2^{-n\gamma})^22^{n(C_{\alpha} - \beta_2 - \beta_3 - \beta)}} } \notag \\
    \le&  \exp\paren{ -\frac{(1-1/L)^2}{4(1+(L-1)2^{-n\gamma})^2}2^{n(2F_1 - C_{\alpha} - 6\gamma +\beta_2 + \beta_3 + \beta - o(1))} } , \label{eqn:apply_size_of_tgood_intersection_c} 
\end{align}
where Inequalities \eqref{eqn:apply_ub_on_exptf} and  \eqref{eqn:apply_lb_on_exptf} are by  Inequalities \eqref{eqn:ub_on_exptf} and \eqref{eqn:lb_on_exptf}, respectively; 
Inequality \eqref{eqn:apply_size_of_tgood_intersection_c} is by Inequality \eqref{eqn:size_tgood_intersection_c}. 
The exponent of bound \eqref{eqn:apply_size_of_tgood_intersection_c} can be made exponentially large by taking sufficiently small $ \tau $. 
Indeed, observe that the exponent is at least
\begin{align}
2F_1 - C_{\alpha} - 6\gamma +\beta_2 + \beta_3 + \beta \ge& 2F_1 - C_{\alpha} - 6 \notag \\
=& \log\paren{\frac{P}{2} - \comega} + \log\frac{1}{\tau} - C_{\alpha} - 6.  \notag 
\end{align}
To make the exponent $ 2F_1 - C_{\alpha} - 6\gamma +\beta_2 + \beta_3 + \beta $ positive, it suffices to take 
\begin{align}
\tau < 2^{-\paren{ C_{\alpha} + 6 - \log\paren{\frac{P}{2} - \comega} }}. \notag
\end{align}

Finally, combining Inequalities \eqref{eqn:conc_tgood_intersection_c} and \eqref{eqn:apply_size_of_tgood_intersection_c}, we have
\begin{align}
\probover{\cC_A,\cC_B}{ \frac{f(\cG)}{\card{\Tgood\cap\cC_A}} \ge 3L2^{-2n\gamma} }\le& \probover{\cC_A,\cC_B}{ f(\cG)\ge \frac{3}{2}L\card{\Tgood\cap\cC}2^{-3n\gamma} }  + \probover{\cC_A}{ \card{ \Tgood\cap\cC_A }\le\frac{1}{2}\card{\Tgood\cap\cC}2^{-n\gamma} } \notag \\
\le& \exp\paren{ -\frac{(1-1/L)^2}{4(1+(L-1)2^{-n\gamma})^2}2^{n(2F_1 - C_{\alpha} - 6\gamma +\beta_2 + \beta_3 + \beta - o(1))} } + \exp\paren{-\frac{1}{12}2^{n(F_1-\gamma-o(1))}}. \notag
\end{align}
\end{proof}
The proof of achievability can be finished by taking  a union bound over $ \vs'\in\cS $ where $ |\cS| = 2^{\cO(n)} $.




\subsection{Improved analysis for  sumset property}\label{sec:improved_sumset}
In this section, we show that one can get rid of the technical condition for sumset property that the covering radius of the underlying lattice is small. 
We prove high-probability bounds over random lattice construction and message selection. 

We use a random nested Construction-A lattice pair with fine lattice $\Lf$ lifted from a $q$-ary $k$-dimensional random linear code and a coarse lattice $\Lc$ that is good for covering. 
Specifically, fix a coarse lattice $\Lc$ with $ \rcov(\Lc) = \sqrt{nP} $ such that 
\begin{align}
\frac{\rcov(\Lc)}{\reff(\Lc)} = \frac{\sqrt{nP}}{\reff(\Lc)} = 1+ \eps_n',\notag
\end{align}
for some $ \eps_n'\xrightarrow{n\to\infty}0 $.
Choose $k$ such that $q^k = 2^{nR}$.
Let $ \bfG' $ be a random matrix uniformly distributed in $ \bF_q^{n\times k} $. 
Define linear code $ \cC' $ generated by $ \bfG' $ as $ \cC'\coloneq\bfG'\bF_q^{k} $.
Define the lattice $ \Lf' $ lifted from $ \cC' $ via Construction-A as $ \Lf'\coloneq\frac{1}{q}\Phi(\cC') + \bZ^n $ where $\Phi$ denotes the natural embedding from $ \bF_q $ to $\bZ$.
Rotate $ \Lf' $ using any generator matrix  $ \bfG_0 $ of $ \Lc $ and obtain the fine lattice $ \Lf \coloneq \bfG_0\Lf' $. 
Finally, define the nested Construction-A lattice code $ \cC $ as $\cC \coloneq \Lf\cap\cV(\Lc)$. 
Let $\phi$ denote the encoding map associated to $\cC$.
\begin{lemma}\label{lem:lattice_pt_norm_improved}
For any $\zeta \in(0,3/4) $ and $m\in\curbrkt{1,2,\cdots,2^{nR}}$, 
\begin{align}
\probover{\cC}{\normtwo{\phi(m)}\le\sqrt{nP(1-\zeta)}}\le& 2^{-n\paren{ \zeta/2 - 4/q - \eps_n' }} .\notag
\end{align}
\end{lemma}
\begin{proof}
Let $\vbfx \coloneq \phi(m)$. By code construction, $\vbfx$ is uniformly distributed in $\frac{1}{q}\Lc\cap\cV(\Lc)$. Therefore,
\begin{align}
\probover{\cC}{\normtwo{\vbfx}\le\sqrt{nP(1-\zeta)}}=&\frac{\card{\frac{1}{q}\Lc\cap\cB^n\paren{\vzero,\sqrt{nP(1-\zeta)}}\cap\cV(\Lc)}}{q^n}\notag\\
\le&\frac{\card{\frac{1}{q}\Lc\cap\cB^n\paren{\vzero,\sqrt{nP(1-\zeta)}}}}{q^n}\notag\\
\le&\frac{V_n}{\vol(\Lc)}\paren{\sqrt{nP(1-\zeta)}+\frac{\rcov(\Lc)}{q}}^n\label{eqn:apply_counting_bd}\\
=&\frac{V_n\rcov(\Lc)^n}{\vol(\Lc)}\paren{\sqrt{1-\zeta}+\frac{1}{q}}^n\notag\\
=&\frac{\vol\paren{\cB^n\paren{\rcov}}}{\vol\paren{\cB^n\paren{\reff}}}\paren{\sqrt{1-\zeta}+\frac{1}{q}}^n\notag\\
=&\paren{\frac{\rcov}{\reff}}^n\paren{\sqrt{1-\zeta}+\frac{1}{q}}^n\notag\\
=&(1+\eps_n')^n\paren{\sqrt{1-\zeta}+\frac{1}{q}}^n\notag\\
=&2^{-n\paren{\log\frac{1}{\sqrt{1-\zeta}+1/q}-\log(1+\eps_n')}}\notag \\
\le& 2^{-n\paren{ \zeta/2 - 4/q - \eps_n' }} ,\label{eqn:apply_alg_ineq}
\end{align}
where Inequality \eqref{eqn:apply_counting_bd} is by Lemma \ref{lem:improved_lattice_pt_bd}.
The last inequality \eqref{eqn:apply_alg_ineq} follows since
\begin{align}
-\log\frac{1}{\sqrt{1-\zeta}+1/q}+\log(1+\eps_n') =& \log\paren{ \sqrt{1-\zeta} + 1/q } + \log(1+\eps_n') \notag \\
\le& \frac{1}{2}\log\paren{1-\zeta} + \log\paren{ 1+\frac{1}{q\sqrt{1-\zeta}} } + 2\eps_n' \label{eqn:apply_alg_bd1} \\
\le& -\zeta/2 + \frac{2}{q\sqrt{1-\zeta}} + \eps_n' \label{eqn:apply_alg_bd2} \\
\le& -\zeta/2 + 4/q + \eps_n', \label{eqn:choose_zeta}
\end{align}
where Inequalities \eqref{eqn:apply_alg_bd1} and \eqref{eqn:apply_alg_bd2} follows from Fact \ref{fact:log_ineq} and Inequality \eqref{eqn:choose_zeta} is by $ \zeta<3/4 $.
\end{proof}

\begin{lemma}\label{lem:lattice_pt_orthogonal_improved}
Fix any $\zeta \in(0,3/4) $. If $\bfm_1,\bfm_2$ are two uniform messages from $\curbrkt{1,2,\cdots,2^{nR}}$, then
\begin{align}
\probover{\bfm_1,\bfm_2,\cC}{\abs{\cos\angle_{\phi(\bfm_1),\phi(\bfm_2)}}\ge \zeta}\le&2\cdot2^{-n\paren{ \zeta^2/2 - 4/q - \eps_n' }}.\notag
\end{align}
\end{lemma}
\begin{proof}
Let $\vbfx_1 \coloneq \phi(\bfm_1),\vbfx_2 \coloneq \phi(\bfm_2)$. By the choice of $\bfm_1,\bfm_2$ and the code design, $\vbfx_1$ and $\vbfx_2$ are independent and uniformly distributed in $\frac{1}{q}\Lc\cap\cV(\Lc)$. For any $\vx\in\cV$, define a cone $\cT_{\vx}$ as
\begin{align}
\cT_{\vx}\coloneqq&\curbrkt{\vv\in\bR^n\colon \abs{\angle_{\vx,\vv}}\ge \zeta}.\notag
\end{align}
Now,
\begin{align}
\probover{\vbfx_1,\vbfx_2,\cC}{\abs{ \cos\paren{\angle_{\vbfx_1,\vbfx_2}} }\ge \zeta}=&\exptover{\vbfx_1}{\frac{\card{\frac{1}{q}\Lc\cap\cT_{\vbfx_1}\cap\cV(\Lc)}}{q^n}}\notag\\
\le&{\frac{\card{\frac{1}{q}\Lc\cap\cT_{\vx_1}\cap\cB^n\paren{\vzero,\sqrt{nP}}}}{q^n}}\label{eqn:improved_fig1}\\
\le&\frac{2\card{\frac{1}{q}\Lc\cap\cB^n\paren{\sqrt{nP(1-\zeta^2)}}}}{q^n}\label{eqn:improved_fig2}\\
\le&\frac{2V_n}{\vol(\Lc)}\paren{\sqrt{nP(1-\zeta^2)}+\frac{\rcov(\Lc)}{q}}^n\notag\\
\le&2\cdot2^{-n\paren{ \zeta^2/2 - 4/q - \eps_n' }},\notag
\end{align}
where Inequalities \eqref{eqn:improved_fig1} and \eqref{eqn:improved_fig2} are illustrated in Fig. \ref{fig:improved_sumset} and $ \vx_1 $ in Inequality \eqref{eqn:improved_fig1} can be taken to be any vector in $ \frac{1}{q}\Lc \cap \cV(\Lc) $. 
\begin{figure}[htbp]
	\centering
	\includegraphics[width=0.5\textwidth]{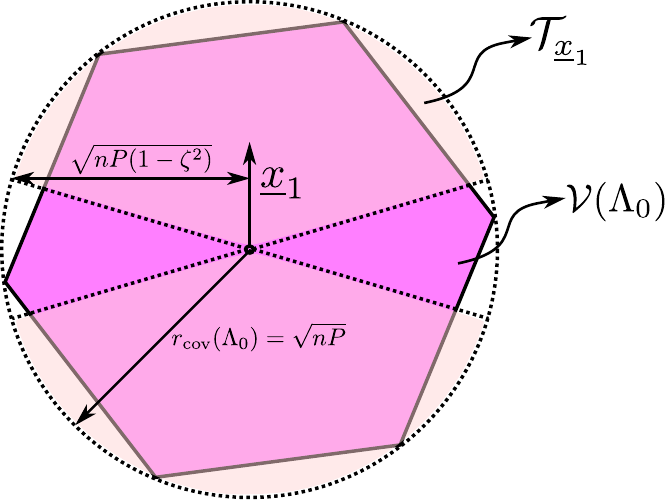}
	\caption{In a random nested Construction-A lattice code, two random codewords are approximately orthogonal to each other w.h.p. over lattice construction and message selection.}
	\label{fig:improved_sumset}
\end{figure}
\end{proof}

Similar to Lemma \ref{lem:xa_ip_xb}, we immediately get the following corollary.
\begin{corollary}\label{cor:xa_ip_xb_improved}
Fix any $\zeta \in(0,3/4) $. If $\bfm_1,\bfm_2$ are two uniform messages from $\curbrkt{1,2,\cdots,2^{nR}}$, then
\begin{align}
\probover{\bfm_1,\bfm_2,\cC}{\abs{\inprod{\phi(\bfm_1)}{\phi(\bfm_2)}}\ge nP\zeta}\le&2\cdot2^{-n\paren{ \zeta^2/2 - 4/q - \eps_n' }}.\notag
\end{align}
\end{corollary}

Similar to Lemma \ref{lem:length_z}, we get sumset property using the Lemma \ref{lem:lattice_pt_norm_improved} and \ref{cor:xa_ip_xb_improved}.
\begin{lemma}\label{lem:length_z_improved}
Fix any $ \delta\in(0,3/4) $, $ \lambda\in(0,\delta) $, let $ \zeta\coloneqq\delta-\lambda $.
If $\bfm_1,\bfm_2$ are two uniform messages from $\curbrkt{1,2,\cdots,2^{nR}}$, then
\begin{align}
\probover{\bfm_1,\bfm_2,\cC}{\normtwo{\phi(\bfm_1)+\phi(\bfm_2)}\notin\sqrt{2nP(1\pm\delta)}}\le& 2\cdot2^{-n\paren{ \lambda/2 - 4/q - \eps_n' }} + 2\cdot2^{-n\paren{ \zeta^2/2 - 4/q - \eps_n' }}.\notag
\end{align}
\end{lemma}

\section{Converse}\label{sec:converse}

\subsection{Scale-and-babble strategy}\label{sec:scale_babble_strategy}
Our converse  works even against stochastic codes which are defined as follows.
\begin{definition}
A stochastic code $ \cC $ is a code which can map a message to different codewords with certain probability. 
Formally, the (stochastic) encoder of $ \cC $ is identified with a conditional distribution: for any $ m\in\cM $ and any $ \vx\in\cB^n\paren{\vzero,\sqrt{nP}} $, 
\begin{align}
\prob{ \enc(m) = \vx } =& P_{\vbfx|\bfm}(\vx|m). \notag
\end{align}
\end{definition}
It suffices to design a jamming strategy for James under which  no rate larger than $ C_A $ is achievable. 
As we shall see, the strategy we are going to design and analyze will turn the adversarial channel into an AWGN channel of certain $\snr$. AWGN channels are defined below.
\begin{definition}
An AWGN$(P,N)$ channel is a channel in which the channel input $\vx\in\bR_n $ satisfies $\normtwo{\vx}\le\sqrt{nP} $ and the channel output is  $\vbfy = \vx + \vbfg$ where $\vbfg\sim\cN\paren{\vzero,N\bfI_n}$. 
\end{definition}

Given any stochastic codebook pair $ (\cC_A,\cC_B) $ for a $ (P,N) $ quadratically constrained two-way adversarial  channel with vanishing probability of error 
$
P_{\e,\avg}(\cC_A) \le \delta_n 
$ and $ P_{\e,\avg}(\cC_B)\le \delta_n $ where $ \delta_n = o_n(1) $, 
we equip James with the following jamming strategy which we call the \emph{scale-and-babble} strategy.

For notational brevity, we write $ \vx_A\coloneqq\vx_{m_A} $ and $ \vx_B \coloneqq\vx_{m_B} $.
Let $ \cB\coloneqq\cB^n\paren{\vzero,\sqrt{nP}} $. 

Given James' received vector $ \vbfz $, define $\wt\vbfs\coloneqq-\alpha\vbfz+\vbfg=-\alpha(\vbfx_A+\vbfx_B)+\vbfg$ for some $ \alpha $ to be optimized later and
\begin{align*}
    \vbfs=&\begin{cases}
    {\wt\vbfs},&\text{if }\normtwo{\wt\vbfs}\le\sqrt{nN}\\
    \sqrt{nN}\frac{\wt\vbfs}{\normtwo{\wt\vbfs}},&\text{otherwise}
    \end{cases},
\end{align*}
where $\vbfg\sim\cN(\vzero,\gamma^2\bfI_n)$ and $\gamma^2=N-2\alpha^2P(1+2\eps)$ for some small constant $\eps>0$. Further define $\cE\coloneqq\curbrkt{\normtwo{\wt\vbfs}>\sqrt{nN}}$,  $Q\coloneqq\prob{\cE}$ and $\bfe\coloneqq\one{\cE}$.
We will reveal the value of $ \alpha $ to Bob and argue that even with such extra information available at decoder, any $ (\cC_A,\cC_B) $ (possibly stochastic) is not able to achieve rate larger than $ \frac{1}{2}\log\paren{\frac{1}{2} + \frac{P}{N}} $.

Under the above jamming strategy, when $\bfe=0$, the channel to Bob is 
\begin{align*}
    \vbfy=&\vbfx_A+\vbfx_B+\vbfs\\
    =&(1-\alpha)(\vbfx_A+\vbfx_B)+\vbfg.
\end{align*}
Since Bob is assumed to know $ \alpha $, he scales and cancels out his signal $ (1-\alpha)\vbfx_B $, and gets effectively $\wt\vbfy=(1-\alpha)\vbfx_A+\vbfg$.

Note that we could assume that 
\begin{align}
\exptover{\vbfx_A\sim\cC_A}{\vbfx_A} =& \frac{1}{M}\sum_{m_A\in\cM}\intgover{\cB}P_{\vbfx_A|\bfm_A}(\vx_A|m_A)\vx_A\diff\vx_A = \vzero, \notag \\
\exptover{\vbfx_B\sim\cC_B}{\vbfx_B} =& \frac{1}{W}\sum_{m_B\in\cW}\intgover{\cB}P_{\vbfx_B|\bfm_B}(\vx_B|m_B)\vx_B\diff\vx_B = \vzero, \notag 
\end{align}
where the expectations are taken over distribution $ \unif(\cM)\times P_{\vbfx_A|\bfm_A} $ and $ \unif(\cW)\times P_{\vbfx_B|\bfm_B} $, respectively. 
Otherwise, assume $\expt{\vbfx_A}=\va\ne\vzero$ and $\expt{\vbfx_B}=\vb\ne\vzero$. Hence every codeword can be decomposed as
\[\vbfx_A=\vbfx_A'+\va,\quad\vbfx_B=\vbfx_B'+\vb.\]
Note that $ \expt{\vbfx_A'} = \expt{\vbfx_B} = \vzero $.
Since $\cC_A$ and $\cC_B$, in particular $\va$ and $\vb$, are known to every party, James could  set $\wt\vbfs\coloneqq -\alpha(\vbfz - \va - \vb) + \vbfg  = -\alpha(\vbfx_A'+\vbfx_B')+\vbfg$. Conditioned on $\cE^c$, Bob receives
\begin{align*}
    \vbfy=&\vbfx_A+\vbfx_B+\vbfs\\
    =&(1-\alpha)(\vbfx_A'+\vbfx_B')+\vbfg+\va+\vb.
\end{align*}
He cancels out $\va$, $\vb$ and $(1-\alpha)\vbfx_B'$ and the effective channel becomes $\wt\vbfy=(1-\alpha)\vbfx_A'+\vbfg$, where $\expt{\vbfx_A'}=\vzero$ and $\vbfg$ is a Gaussian, which is identical to the previous case. 


\subsection{Analysis}
\begin{lemma}
Under the scale-and-babble strategy defined in Sec. \ref{sec:scale_babble_strategy}, no code $ (\cC_A,\cC_B) $ (possibly stochastic) with vanishing average probability of error for a $ (P,N) $ quadratically constrained two-way adversarial  channel can have rate larger than $ \frac{1}{2}\log\paren{1 + \frac{P}{N}} $. That is $ C_A\le\frac{1}{2}\log\paren{1 + \frac{P}{N}} $ and $ C_B\le \frac{1}{2}\log\paren{1 + \frac{P}{N}} $.
\end{lemma}
\begin{proof}
To get an upper bound on $R_A$, we decompose $nR_A$ using standard information (in)equalities.
\begin{align}
    nR_A=&H(\bfm_A) \label{eqn:msg_unif} \\
    =&H(\bfm_A|\bfe) + I(\bfm_A;\bfe) \notag \\
    \le&\prob{\cE^c}H(\bfm_A|\cE^c)+\prob{\cE}H(\bfm_A|\cE) + 1\label{eqn:e_binary}\\
    =&\overline{Q}\paren{H(\bfm_A|\wt\bfy,\cE^c) + I(\bfm_A;\wt\bfy|\cE^c)}+Q H(\bfm_A|\cE) + 1\label{eqn:ytilde_def}\\
    \le&\overline{Q}\paren{n\eps_n + I(\vbfx_A;\wt\vbfy|\cE^c)  } + Q nR_A + 1\label{eqn:fano_dpi}\\
    =&\overline{Q}\paren{n\eps_n + H(\wt\vbfy|\cE^c) - H(\wt\vbfy|\vbfx_A,\cE^c)  } + Q nR_A + 1\notag\\
    =&\overline{Q}\paren{n\eps_n + H(\wt\vbfy|\cE^c) - H(\vbfg|\cE^c)  } + Q nR_A + 1,\label{eqn:entropy_y_equals_entropy_g}
\end{align}
In the above chain of (in)equalities,
\begin{enumerate}
	\item Equality \eqref{eqn:msg_unif} follows since $ \bfm_A $ is uniformly distributed on $ [2^{nR_A}] $. 
	\item Inequality \eqref{eqn:e_binary} follows since $ \bfe $ is a binary random variable and $ I(\bfm_A;\bfe)\le H(\bfe)\le1 $.
	\item In Eqn. \eqref{eqn:ytilde_def}, $ \wt\vbfy $ denotes $ (1-\alpha)\vbfx_A + \vbfg $.
	\item Inequality \eqref{eqn:fano_dpi} is by   Fano's inequality and data processing inequality, since conditioned on $\cE^c$ the effective channel to Bob is an AWGN channel.
	We can take $ \eps_n \coloneqq R_A\delta_n+1/n = o_n(1) $.
	\item Equality \eqref{eqn:entropy_y_equals_entropy_g} is justified below, 
	\begin{align}
	H(\wt\vbfy|\vbfx_A,\cE^c) =& H( (1-\alpha)\vbfx_A + \vbfg|\vbfx_A,\cE^c ) \notag \\
	=& H(\vbfg|\vbfx_A,\cE^c) \notag \\
	=& H(\vbfg|\cE^c), \label{eqn:g_indep}
	\end{align}
	where Equality \eqref{eqn:g_indep} follows since $ \vbfg $ is a white Gaussian noise independent of everything else. 
\end{enumerate}

In what follows, we upper bound $ H(\wt\vbfy|\cE^c) $ and lower bound $ H(\vbfg|\cE^c) $ separately.
To bound $H(\wt\vbfy|\cE^c)$, note that, by subadditivity of entropy,
\begin{align}
H(\wt\vbfy|\cE^c) \le& \sum_{i = 1}^nH(\wt\vbfy(i)|\cE^c). \notag
\end{align}
Each $H\paren{\wt\vbfy(i)|\cE^c}$ can be bounded using the principle of maximum entropy. 
Observe that  $ \cE^c= \curbrkt{ \normtwo{\wt\vbfs}\le\sqrt{nN} }  $ truncates $ \wt\vbfs $ at the boundary of the ball $ \cB^n\paren{\vzero,\sqrt{nN}} $, hence conditioning on $ \cE^c $ will not increase the variance of $ \vbfx_A+\vbfx_B+\wt\vbfs $. Since $ \vbfy = (1-\alpha)(\vbfx_A+\vbfx_B)+\vbfg $ and $ \wt\vbfy = (1-\alpha)\vbfx_A+\vbfg $ can be computed by Bob only when $ \cE^c $ happens, we have
\begin{align}
\expt{\left.\wt\vbfy(i)^2\right|\cE^c}\le& \expt{\wt\vbfy(i)^2} . \notag
\end{align} 
Now we compute 
\begin{align}
    \expt{\wt\vbfy(i)^2}=&\expt{\paren{(1-\alpha)\vbfx_A(i)+\vbfg(i)}^2} \notag \\
    =&(1-\alpha)^2\expt{\vbfx_A(i)^2}+\gamma^2. \label{eqn:second_mmt_of_y}
\end{align}
Equality \eqref{eqn:second_mmt_of_y} follows since $ \vbfg(i) $ is independent of $ \vbfx_A(i) $ and has mean 0, variance $ \gamma^2 $. 
Now, by the entropy vs. variance bound (Lemma \ref{lem:entropy_vs_var}),
\begin{align}
    H\paren{\wt\vbfy(i)}\le&\frac{1}{2}\log\paren{2\pi e\var{\wt\vbfy(i)}} \notag \\
    =&\frac{1}{2}\log\paren{2\pi e\expt{\wt\vbfy(i)^2}} \label{eqn:y_zero_mean} \\
    =&\frac{1}{2}\log\paren{2\pi e\paren{(1-\alpha)^2\expt{\vbfx_A(i)^2}+\gamma^2}}.  \notag 
\end{align}
Equality \eqref{eqn:y_zero_mean} follows since 
$
\expt{\wt\vbfy(i)} = (1-\alpha)\expt{ \vbfx_A(i)   }+ \expt{\vbfg(i)}
= 0
$. 
Therefore,
\begin{align*}
    H\paren{\wt\vbfy}\le&\sum_{i=1}^n\frac{1}{2}\log\paren{2\pi e\paren{(1-\alpha)^2\expt{\vbfx_A(i)^2}+\gamma^2}}\\
    =&\frac{1}{2}\log\paren{\prod_{i=1}^n2\pi e\paren{(1-\alpha)^2\expt{\vbfx_A(i)^2}+\gamma^2}}.
\end{align*}
Since $\sum_{i=1}^n\vbfx_A(i)^2\le nP$ with probability 1, the above bound is maximized when each $\vbfx_A(i)^2$ is equal to $P$. We have
\begin{align}
    H\paren{\wt\vbfy|\cE^c}\le & \frac{1}{\overline{Q}} H\paren{\wt\vbfy}\le \frac{1}{\overline{Q}}\frac{n}{2}\log\paren{2\pi e\paren{(1-\alpha)^2P+\gamma^2}}. \label{eqn:bound_ent_y_given_ec}
\end{align}

The term $H(\vbfg|\cE^c)$ can be bounded in a similar manner. 
\begin{align}
H(\vbfg|\cE^c) \le& \sum_{i = 1}^n H(\vbfg(i)|\cE^c) \notag \\
\le& \sum_{i = 1}^n\frac{1}{2}\log\paren{ 2\pi e\expt{\vbfg(i)^2|\cE^c} } \notag \\
\le& \sum_{i = 1}^n\frac{1}{2}\log\paren{ 2\pi e\expt{\vbfg(i)^2} } \label{eqn:conditioning_reduces_var} \\
=& \frac{n}{2}\log(2\pi e\gamma^2), \label{eqn:bound_ent_g_given_ec}
\end{align}
where Inequality \eqref{eqn:conditioning_reduces_var} follows by noting 
\begin{align}
\cE^c =& \curbrkt{ \normtwo{{\wt\vbfs}}\le\sqrt{nN} } \notag \\
=& \curbrkt{ \normtwo{-\alpha(\vbfx_A+\vbfx_B) + \vbfg}\le\sqrt{nN} } \notag \\
=& \curbrkt{ \vbfg\in\cB^n\paren{ \alpha(\vbfx_A+\vbfx_B),\sqrt{nN} } }, \notag
\end{align}
and hence $ \cE^c $ restricts $ \vbfg $ to a (random) ball $ \cB^n\paren{ \alpha(\vbfx_A+\vbfx_B),\sqrt{nN} }  $ in which the variance of $ \vbfg $ can only be no larger. 

Finally, combining the bounds \eqref{eqn:entropy_y_equals_entropy_g}, \eqref{eqn:bound_ent_y_given_ec} and \eqref{eqn:bound_ent_g_given_ec}, we have
\begin{align*}
    R_A\le & \overline{Q}\eps_n + \frac{1}{2}\log(2\pi e((1-\alpha)^2P + \gamma^2)) - \overline{Q}\paren{\frac{1}{2}\log(2\pi e\gamma^2) - \frac{1}{n\overline{Q}}} + Q  R_A + \frac{1}{n}\\
    =&\frac{1}{2}\log\paren{1+\frac{(1-\alpha)^2P}{\gamma^2}} + \frac{Q}{2}\log(2\pi e\gamma^2) + \overline{Q}\eps_n  + Q R_A + \frac{2}{n}.
\end{align*}
Rearranging terms, we have
\begin{align}
R_A\le&\frac{1}{\overline{Q}}\frac{1}{2}\log\paren{1+\frac{(1-\alpha)^2P}{\gamma^2}}+ \frac{Q}{\overline{Q}} \frac{1}{2}\log(2\pi e\gamma^2) + \eps_n+\frac{2}{n\overline{Q}}. \notag
\end{align}
As shown in Sec. \ref{sec:bound_q}, 
$Q=o_n(1)$.
Substituting it back, we get
\begin{align}
R_A\le&\frac{1}{2(1-o_n(1))}\log\paren{1+\frac{(1-\alpha)^2P}{\gamma^2}}+\frac{o_n(1)}{1-o_n(1)}\frac{1}{2}\log(2\pi e\gamma^2)+\eps_n+\frac{2}{n(1-o_n(1))}. \notag
\end{align}
Taking the limit as $n\to\infty$, we have
\begin{align}
R_A\asymp& \log\paren{ 1+\frac{(1-\alpha)^2P}{\gamma^2} } \notag \\
=& \frac{1}{2}\log\paren{ 1+ \frac{(1-\alpha)^2P}{N - 2\alpha^2P(1+2\eps)} } \notag \\
=& \frac{1}{2}\log\paren{ 1 + \frac{(1-\alpha)^2P}{N - 2\alpha^2P} + \frac{(2\alpha(1-\alpha)P)^2\eps}{(N-2\alpha^2P)(N-2\alpha^2P(1+2\eps))} }. \notag
\end{align}
Optimizing over admissible $\alpha$ and sending $\eps$ to 0 finishes the proof. 
\end{proof}

\subsection{Bounding $Q$}\label{sec:bound_q}
\begin{lemma}
$ Q = o_n(1) $.
\end{lemma}
\begin{proof}
By definition of $Q$,
\begin{align*}
{Q} =&\prob{\cE}\\
= &\prob{\normtwo{\wt\vbfs}>\sqrt{nN}}\\
    =&\prob{\normtwo{-\alpha(\vbfx_A+\vbfx_B)+\vbfg}>\sqrt{nN}}\\
    \le&\prob{\cE_1}+\prob{\cE_2}+\prob{\curbrkt{\alpha^2\normtwo{\vbfx_A+\vbfx_B}^2+\normtwo{\vbfg}^2-2\inprod{\alpha(\vbfx_A+\vbfx_B)}{\vbfg}>{nN}}\cap\cE_1^c\cap\cE_2^c},
\end{align*}
where 
\begin{align*}
    \cE_1\coloneqq&\curbrkt{\inprod{-\alpha(\vbfx_A+\vbfx_B)}{\vbfg}>n\eta_1},\\
    \cE_2\coloneqq&\curbrkt{\normtwo{\vbfg}^2> n\gamma^2(1+\eta_2)}.
\end{align*}

The first two terms are easy to bound. By Gaussian tail bound (Lemma \ref{lem:gaussian_tail}), the first one is at most
\begin{align*}
    \prob{\cE_1}=&\prob{\bfg'>n\eta_1}\\
    \le&\exp\paren{-\frac{(n\eta_1)^2}{2\alpha^2\normtwo{\vbfx_A+\vbfx_B}^2\gamma^2}}\\
    \le&\exp\paren{-\frac{n^2\eta_1^2}{2\alpha^2\cdot4nP\cdot\gamma^2}}\\
    =&\exp\paren{-\frac{n\eta_1^2}{8\alpha^2P\gamma^2}},
\end{align*}
where $\bfg'\coloneqq\inprod{-\alpha(\vbfx_A+\vbfx_B)}{\vbfg}\sim\cN\paren{0,\alpha^2\normtwo{\vbfx_A+\vbfx_B}^2\gamma^2}$. 
As to the second term, by the standard tail bound of $\chi^2$-distributions (Lemma \ref{lem:chisquared_tail}), 
\begin{align*}
    \prob{\cE_2}\le&\exp\paren{-\frac{\eta_2^2}{4}n}.
\end{align*}

The last term is at most
\begin{align}
    &\prob{\normtwo{\vbfx_A+\vbfx_B}^2>n\paren{\frac{N-\gamma^2(1+\eta_2)-2\eta_1}{\alpha^2}}} \notag \\
    =& \prob{\normtwo{\vbfx_A+\vbfx_B}^2>2nP(1+\eps)} \label{eqn:set_param_for_eps} \\
    =& \prob{ \normtwo{\vbfx_A}^2 + \normtwo{\vbfx_B}^2 + 2\inprod{\vbfx_A}{\vbfx_B} > 2nP(1+\eps) } \notag \\
    \le& \prob{ {\inprod{\vbfx_A}{\vbfx_B}} > nP\eps }, \label{eqn:prob_from_empirical} 
\end{align}
where in Eqn. \eqref{eqn:set_param_for_eps} we take $\eta_1 = \alpha^2P\eps/2$, $\eta_2 = \alpha^2P\eps/\gamma^2 = \frac{\alpha^2P\eps}{N-2\alpha^2P(1+2\eps)}$.
The probability in Eqn. \ref{eqn:prob_from_empirical} is $o_n(1) $ by setting $\eta = P\eps$ in Lemma \ref{lem:empirical_indep_awgn} as shown in Sec. \ref{sec:empirical_properties_awgn}.

All in all, we have
\begin{align*}
    \prob{\cE}\le&\exp\paren{-\frac{n\eta_1^2}{8\alpha^2P\gamma^2}} + \exp\paren{-\frac{\eta_2^2}{4}n} + o_n(1)\\
    =&\exp\paren{-\frac{n\alpha^2P\eps^2}{32\gamma^2}}+\exp\paren{-\frac{\alpha^4P^2\eps^2n}{4\gamma^4}}+o_n(1)\\
    =&o_n(1) .
\end{align*}
That is, $Q=o_n(1) $ as promised. 
\end{proof}

\subsection{Empirical properties of AWGN-good codes} \label{sec:empirical_properties_awgn}
To bound the probability \eqref{eqn:prob_from_empirical}, we will prove certain empirical property that is universal for any capacity-achieving code for an AWGN channel. To this end, we first define  AWGN-goodness.

\begin{definition}
An infinite sequence of (possibly stochastic) codes $\{\cC_n\}_n$, where $\cC_n\subset\cB^n\paren{\vzero,\sqrt{nP}} $ is equipped with encoder $ \enc_n $ and decoder $\dec_n $, is said to be good for  AWGN$(P,N)$ channels  if 
\begin{itemize}
\item  for an arbitrarily small constant $\delta>0$ and for all $n$, $R(\cC_n)\ge \frac{1}{2}\log(1+P/N)-\delta$;  and
\item $P_{\e,\avg}(\cC_n) = o_n(1)$.
\end{itemize}
\end{definition}

We then prove the following lemma which provides an $o_n(1)$ bound on the probability \eqref{eqn:prob_from_empirical}.
\begin{lemma}\label{lem:empirical_indep_awgn}
Given any two (possibly stochastic) codes $\cC_1$ and $\cC_2$ that are good for   AWGN$(P,N)$ channels,  for any constant $\eta\in(0,1)$, it holds that
\[\limsup_{n\to\infty}\probover{\substack{\vbfx_1\sim\cC_1\\\vbfx_2\sim\cC_2}}{{\inprod{\vbfx_1}{\vbfx_2}}>n\eta} = 0,\]
where the probability is taken over $\vbfx_1$ and $\vbfx_2$ that are chosen according to the encoders of $\cC_1$ and $\cC_2$, respectively.
\end{lemma}
\begin{proof}
Suppose $R(\cC_1) = \frac{1}{2}\log(1+P/N) - \delta_1 $ and $R(\cC_2) =\frac{1}{2}\log(1+P/N) - \delta_2 $ for arbitrarily small constants $\delta_1>0 $ and $\delta_2>0 $. 
Suppose $\cC_1 = \curbrkt{\vx_i}_{i\in[M_1]}$ and $\cC_2=\curbrkt{\vx_j}_{j\in[M_2]}$ have probability of error $\eps_n$ and $\delta_n$ under their decoders $\dec_1 $ and $\dec_2 $, respectively, when used over an AWGN$(P,N)$ channel. 
Since $\cC_1$ and $\cC_2$ are good, $\eps_n\xrightarrow{n\to\infty}0$ and $\delta_n\xrightarrow{n\to\infty}0$.
Let $M_1$ and $M_2$ denote $|\cC_1|$ and $|\cC_2|$, respectively.

Assume, towards a contradiction, that there exists some constant $\eps>0$ such that
\begin{align*}
    \limsup_{n\to\infty}\probover{\substack{\vbfx_1\sim\cC_1\\\vbfx_2\sim\cC_2}}{{\inprod{\vbfx_1}{\vbfx_2}}>n\eta}=2\eps.
\end{align*}
Hence for infinitely many $n$ that are sufficiently large, we have
\begin{align*}
 \eps \le& \probover{\substack{\vbfx_1\sim\cC_1\\\vbfx_2\sim\cC_2}}{{\inprod{\vbfx_1}{\vbfx_2}}>n\eta} \notag \\
 =& \exptover{\vbfx_1\sim\cC_1}{\probover{\vbfx_2\sim\cC_2}{{\inprod{\vbfx_1}{\vbfx_2}}>n\eta}} \notag \\
 =& \exptover{ \bfm_1\sim[M_1] }{ \exptover{\vbfx_1\sim P_{\vbfx_1|\bfm_1}}{ \probover{ \vbfx_1\sim\cC_2 }{ \inprod{\vbfx_1}{\vbfx_2}>n\eta } } } \notag \\
 =& \frac{1}{M_1}\sum_{i\in[M_1]} \intgover{\cB^n\paren{\vzero,\sqrt{nP}}} P_{\vbfx_1|\bfm_1}(\vzeta_i|i) {\probover{\vbfx_2\sim\cC_2}{{\inprod{\vzeta_i}{\vbfx_2}}>n\eta}} \diff\vzeta_i . \notag 
\end{align*}

By Markov's inequality, there exists an $i_0\in[M_1]$ such that 
\begin{align}
\intgover{\cB^n\paren{\vzero,\sqrt{nP}}} P_{\vbfx_1|\bfm_1}(\vzeta_{i_0}|i_0) \probover{\vbfx_2\sim\cC_2}{{\inprod{\vzeta_{i_0}}{\vbfx_2}}>n\eta} \diff\vzeta_{i_0} \ge\eps. \label{eqn:extraction_by_markov}
\end{align}
Since $ P_{\vbfx_1|\bfm_1}(\vzeta_{i_0}|i_0) \ge0 $,  by the first mean value theorem (Lemma \ref{lem:mean_value_thm}) for integral, there exists an $ \vx_{i_0} $ such that the integral \eqref{eqn:extraction_by_markov} equals
\begin{align}
\probover{ \vbfx_2\sim\cC_2 }{ \inprod{\vx_{i_0}}{\vbfx_2} > n\eta } \intgover{\cB^n\paren{\vzero,\sqrt{nP}}}  P_{\vbfx_1|\bfm_1}(\vzeta_{i_0}|i_0) \diff\vzeta_{i_0} =& \probover{ \vbfx_2\sim\cC_2 }{ \inprod{\vx_{i_0}}{\vbfx_2} > n\eta } \ge \eps. \label{eqn:extraction_by_mean_value}
\end{align}
Define a halfspace 
\begin{align}
\cH = \cH_{\vx_{i_0}, \eta} \coloneqq& \curbrkt{\vx\in\bR^n\colon \inprod{\vx_{i_0}}{\vx}>n\eta}. \notag
\end{align}
Define subcode $ \cC_2' $  as
$\cC_2'\coloneqq \cC_2\cap\cH$. 
Note that $ \cC_2' $ is a subcode contained in the pink cap as shown in Fig. \ref{fig:awgn_extract_subcode_with_large_corr}. 

For each $ j\in[M_2] $, define
\begin{align}
Z_j \coloneqq& \intgover{\cB^n\paren{\vzero,\sqrt{nP}}} P_{\vbfx_2|\bfm_2}(\vxi_j|j) \one{\cH}(\vxi_j) \diff\vxi_j. \notag 
\end{align}
Note that $ Z_j\le1 $ for every $j$.
It is not hard to see that $ \cC_2' $ can also be written as
\begin{align}
\cC_2' \coloneqq \bigcup_{j\in[M_2]\colon Z_j>0} \curbrkt{\vx_j\in\cC_2\colon \vx_j\in\cH }. \notag
\end{align}
The encoder of $ \cC_2' $ is identified with the following conditional distribution: for every $ j $,
\begin{align}
P_{\vbfx_2|\bfm_2}'(\vx_j|j) =& \frac{1}{Z_j}P_{\vbfx_2|\bfm_2}(\vx_j|j)\one{\cH}(\vx_j) , \notag
\end{align}
Let $K $ be the size of message set of $\cC_2'$. Note that $ K = \curbrkt{j\in[M_2]\colon Z_j>0} $. By Eqn. \eqref{eqn:extraction_by_mean_value}, 
\begin{align}
\eps\le& \frac{1}{M_2}\sum_{j\in[M_2]}\intgover{\cB^n\paren{\vzero,\sqrt{nP}}} P_{\vbfx_2|\bfm_2}(\vxi_j|j)\one{\cH}(\vxi_j)\diff\vxi_j \notag \\
=& \frac{1}{M_2}\sum_{j\in[M_2]} Z_j \notag \\
=& \frac{1}{M_2}\sum_{j\in[M_2]} Z_j\indicator{Z_j>0} \notag \\
\le& \frac{1}{M_2}\sum_{j\in[M_2]}\indicator{Z_j>0} \notag \\
=& K/M_2, \notag
\end{align}
i.e., $ K\ge M_2\eps $.

Let $\vx_j^*\in\cB^n\paren{\vzero,\sqrt{nP}} $ be such that ${\inprod{\vx_{i_0}}{\vx_j^*}} = \eta $. 
(Note that $\vx_j^* $ is on the boundary of the cap but it may not be a codeword in $\cC_2$.) 
Define $\theta = \angle_{\vx_{i_0},\vx_j^*} $. Then, as shown in Fig. \ref{fig:awgn_extract_subcode_with_large_corr}, we have
\begin{align}
\cos\theta =& \frac{{\inprod{\vx_{i_0}}{\vx_j^*}}}{\normtwo{\vx_{i_0}}\normtwo{\vx_j^*}} \notag \\
>& \frac{n\eta}{\sqrt{nP}\sqrt{nP}} \notag \\
=& \eta/P. \notag
\end{align}
Hence the radius $\sqrt{nP'}$ the the cap can be computed as follows.
\begin{align}
\paren{\frac{\sqrt{nP'}}{\sqrt{nP}}}^2 =& (\sin\theta)^2  \notag \\
=& 1-(\cos\theta)^2 \notag \\
<& 1-(\eta/P)^2. \notag 
\end{align}
We get $P'< P-\eta^2/P<P$.

Now move the  cap (together with codewords in it) so that its center becomes the origin. 
We get a new code $\cC_2'' $ of the same cardinality $K$ as $\cC_2' $. Every codeword $\vx_j''\in\cC_2'' $ satisfies $\normtwo{\vx_j''}\le\sqrt{nP'}$.



\begin{figure}[htbp]
	\centering
	\includegraphics[width=0.95\textwidth]{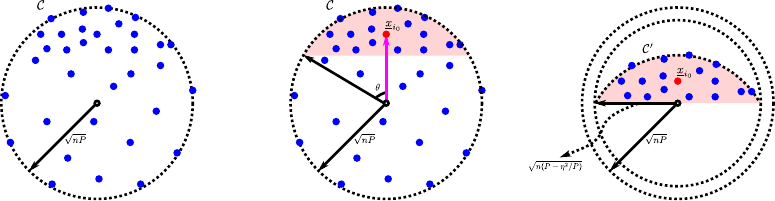}
	\caption{Extracting subcode with large correlation. If an AWGN-capacity-achieving code is highly correlated on average, then we can find a spherical cap among which all codewords have large correlation with the center of the cap. Furthermore, this cap contains a constant fraction of codewords. Moving such a cap to the origin, we get a code of the same rate but lower power. Used over the same channel, it still has vanishing error probability, which violates the fundamental limits.}
	\label{fig:awgn_extract_subcode_with_large_corr}
\end{figure}

Equip $\cC_2'' $ with the same decoder $\dec_2 $ as $\cC_2 $. We claim that when used over an AWGN$(P',N)$ channel, $\cC_2'' $ also has vanishing average probability of error.  Indeed, first note that translating codewords does not change the pairwise distance, hence $P_{\e,\avg}(\cC_2'') = P_{\e,\avg}(\cC_2') $. (Here we use the same decoder $\dec_2 $ for $\cC_2'$ as well.) It suffices to bound $P_{\e,\avg}(\cC_2') $. To this end, define, for every $ m $ and $ \vx_m $,
\begin{align}
P_{\e}(m,\vx_m)\coloneqq& \prob{\wh\bfm\ne m|\bfm=m,\;\vbfx = \vx_m} \notag \\
=& \probover{\vbfg\sim\cN(\vzero,N\bfI_n)}{ \dec(\vx_m + \vbfg) \ne m }. \notag
\end{align}
Then
\begin{align}
    M_2\delta_n=&M_2P_{\e,\avg}(\cC_2) \notag \\
    =&\sum_{j\in[M_2]} \intgover{\cB^n\paren{\vzero,\sqrt{nP}}} P_{\vbfx_2|\bfm_2}(\vxi_j|j) P_{\e}(j,\vxi_j)\diff\vxi_j \notag \\
    \ge&\sum_{j\in[M_2]\colon Z_j>0 }  \intgover{\cB^n\paren{\vzero,\sqrt{nP}}} Z_jP_{\vbfx_2|\bfm_2}'(\vxi_j|j) P_{\e}(j,\vxi_j)\diff\vxi_j \notag \\
    \ge& Z_* \sum_{j\in[M_2]\colon Z_j>0 }  \intgover{\cB^n\paren{\vzero,\sqrt{nP}}} P_{\vbfx_2|\bfm_2}'(\vxi_j|j) P_{\e}(j,\vxi_j)\diff\vxi_j  . \label{eqn:bound_pe_subcode_tobecont}
\end{align}
where in Inequality \eqref{eqn:bound_pe_subcode_tobecont}, we let
\begin{align}
Z_* \coloneqq& \min_{j\in[M_2]\colon Z_j>0}Z_j. \notag
\end{align}
Note that $ Z_*>0 $ is a constant independent of $ n $. 
From Eqn. \eqref{eqn:bound_pe_subcode_tobecont}, we get
\begin{align}
    P_{\e,\avg}(\cC') =& \frac{1}{K} \sum_{j\in[M_2]\colon Z_j>0} \intgover{\cB^n\paren{\vzero,\sqrt{nP}}} P_{\vbfx_2|\bfm_2}'(\vxi_j|j) P_{\e}(j,\vxi_j) \diff\vxi_j \notag \\
    \le&\frac{M_2}{KZ_*}\delta_n  \notag \\
    \le&\frac{\eps}{Z_*}\delta_n \notag \\ 
    =&o_n(1).  \notag 
\end{align}
Moreover,  $\cC_2''$ achieves essentially the same rate as $\cC_2 $ which achieves the capacity of AWGN$(P,N)$ channels. 
\begin{align}
R(\cC_2'') =& R(\cC_2') \notag \\
=&\frac{1}{n}\log K \notag\\
\ge& \frac{1}{n}\log(M_2\eps)\notag \\
=&R(\cC_2) + \frac{\log\eps}{n}\notag \\
\xrightarrow{n\to\infty}& R(\cC_2)\notag \\
\xrightarrow{\delta_2\to0}& \frac{1}{2}\log\paren{1+\frac{P}{N}}. \notag
\end{align}
However, the AWGN$(P',N)$ that $\cC_2'' $ is used over has capacity $\frac{1}{2}\log(1+P'/N)< \frac{1}{2}\log(1+P/N) $.
This violates the fundamental channel coding theorem by Shannon and finishes the proof. 
\end{proof}


Finally, we list several straightforward corollaries of Lemma \ref{lem:empirical_indep_awgn} that may be useful elsewhere. 
\begin{corollary}
Given any (possibly stochastic) codes $\cC_1$, $\cC_2$ and $\cC$ that are good for   AWGN$(P,N)$ channels,  for any constant $\eta\in(0,1)$ and $k\in\bZ_{\ge2} $, it holds that
\begin{align}
\limsup_{n\to\infty}\probover{\substack{\vbfx_1\sim\cC_1\\\vbfx_2\sim\cC_2}}{{\inprod{\vbfx_1}{\vbfx_2}}<-n\eta} &= 0, \label{eqn:cor_other_side} \\
\limsup_{n\to\infty}\probover{\substack{\vbfx_1\sim\cC_1\\\vbfx_2\sim\cC_2}}{\abs{\inprod{\vbfx_1}{\vbfx_2}}>n\eta} &= 0, \label{eqn:cor_two_sided} \\
\limsup_{n\to\infty}\probover{\substack{\vbfx,\vbfx'\iid\cC}}{\abs{\inprod{\vbfx}{\vbfx'}}>n\eta} &= 0. \label{eqn:cor_one_codebook} \\
\limsup_{n\to\infty}\probover{\substack{\vbfx_1,\cdots,\vbfx_k\iid\cC}}{\bigcup_{\substack{i,j\in[k]\\i\ne j}}\curbrkt{\abs{\inprod{\vbfx_i}{\vbfx_j}}>n\eta}} &= 0. \label{eqn:cor_many_cw} 
\end{align}
\end{corollary}
\begin{proof}
Eqn. \eqref{eqn:cor_other_side} follows from the same argument as Lemma \ref{lem:empirical_indep_awgn}. Eqn. \eqref{eqn:cor_two_sided} follows from a union bound that combines Lemma \ref{lem:empirical_indep_awgn} and Eqn. \eqref{eqn:cor_other_side}. Eqn. \eqref{eqn:cor_one_codebook} follows by setting $\cC_2 = \cC_1 =\cC $ in Eqn. \eqref{eqn:cor_two_sided}. Finally, since there are $\binom{k}{2} = \cO_n(1) $ many distinct $(i,j)$ pairs,  Eqn. \eqref{eqn:cor_many_cw} follows  from Eqn. \eqref{eqn:cor_one_codebook} and a union bound.
\end{proof}

Using similar ideas, we prove another empirical property that is universal to all AWGN-good codes, thought it is not used in our main proof. 
\begin{lemma}
Given any  (possibly stochastic) code $\cC$  that is good for   AWGN$(P,N)$ channels,  for any constant $\eta\in(0,1)$, it holds that
\begin{align}
\limsup_{n\to\infty}\probover{\vbfx\sim\cC}{\normtwo{\vbfx}\le\sqrt{nP(1-\eta)} } =& 0. \notag
\end{align}
\end{lemma}
\begin{proof}
Suppose towards a contradiction that for  some constant $\eps>0$
\begin{align}
\eps =& \probover{\vbfx\sim\cC}{\normtwo{\vbfx}\le\sqrt{nP(1-\eta)} }   \notag \\
=& \frac{1}{M}\sum_{j\in[M]} Z_j, \notag
\end{align}
where we defined 
\begin{align}
Z_j \coloneqq& \intgover{\cB^n\paren{\vzero,\sqrt{nP}}} P_{\bfx|\bfm}(\vxi_j|j)\indicator{\normtwo{\vxi_j}<\sqrt{nP(1-\eta)}} \diff\vxi_j , \notag
\end{align}
for each $ j\in[M] $.
Now define
\[\cC'\coloneqq\curbrkt{ \vx\in\cC\colon \normtwo{\vx}\le\sqrt{nP(1-\eta)} }.\]
By the same considerations as in Lemma \ref{lem:empirical_indep_awgn}, we have
\begin{itemize}
	\item on the one hand, $\cC'$ has $o_n(1) $ average probability of error when used over AWGN$(P(1-\eta),N)$ channels which have capacity $\frac{1}{2}\log(1+P(1-\eta)/N)<\frac{1}{2}\log(1+P/N) $;
	\item on the other hand, the number of messages that $\cC'$ encodes is  $M\eps$, in particular, $\cC'$ achieves rate arbitrarily close to $\frac{1}{2}\log(1+P/N) $,
\end{itemize}
which is a contradiction. 
\end{proof}

\subsection{\texorpdfstring{\underline{$\mathbf{z}$}-aware symmetrization}{z}}
\label{sec:z-aware_symm}
\begin{lemma}
For a $ (P,N) $ quadratically constrained two-way adversarial channel, assume $ N = 3P(1+\eps)/4 $ for some constant $ \eps>0 $. Then any codebook pair $ (\cC_A,\cC_B) $ of sizes $ |\cC_A|\ge\frac{\eps}{2(1+\eps)} $ and $ |\cC_B|\ge\frac{\eps}{2(1+\eps)} $ has average error probabilities $ P_{\e,\avg,A}\ge\frac{\eps}{4(1+\eps)} $ and $ P_{\e,\avg,B}\ge\frac{\eps}{4(1+\eps)} $. 
\end{lemma}

\begin{proof}
Given any  codebooks $\cC_A$ and $\cC_B$ of positive rate, by similar considerations, we can assume without loss of generality that $\expt{\vbfx_A} = \expt{\vbfx_B} = \vzero$, where the expectation is over $\vbfx_A$ and $\vbfx_B$ that are randomly chosen from $\cC_A$ and $\cC_B$, respectively.

Define $\wt\vbfs = -\frac{1}{2}(\vbfz - \vbfx_A') = -\frac{1}{2}(\vbfx_A+\vbfx_B-\vbfx_A')$, where $\vbfx_A'$ is a  random codeword from $\cC_A.$
Define $\vbfs$ as follows.
\begin{align}
\vbfs = &\begin{cases}
{\wt\vbfs},&\normtwo{\wt\vbfs}\le\sqrt{nN}\\
\sqrt{nN}\frac{\wt\vbfs}{\normtwo{\wt\vbfs}},&\ow
\end{cases}.\notag
\end{align}
Define error events  
\begin{align}
	\cE_1 \coloneqq& \curbrkt{ \normtwo{\wt\vbfs}>\sqrt{nN} }, \notag \\
	\cE_2 \coloneqq& \curbrkt{ \vbfx_A = \vbfx_A' }. \notag
\end{align}
Under the above jamming strategy, Bob receives
\begin{align}
\vbfy_B =& \vbfx_A+\vbfx_B - \frac{1}{2}(\vbfx_A+\vbfx_B - \vbfx_A')\notag\\
=&\frac{1}{2}(\vbfx_A+\vbfx_A')+\frac{1}{2}\vbfx_B.\notag
\end{align}
If $ \wt\vbfs $ satisfies  power constraint, , cancelling his own signal, Bob effectively receives $\wt\vbfy_B = \frac{1}{2}(\vbfx_A+\vbfx_A')$. If neither $\cE_1 $ nor $ \cE_2 $  happens,  then Bob has no way to distinguish between $\vbfx_A$ and $\vbfx_A'$ and the decoding error probability is at least $1/2$ under any decoding rule. 

We now formally lower bound the probability of error under such a jamming strategy.
\begin{align}
P_{\e,B} =& \prob{\wh\bfm_A \ne \bfm_A}\notag \\
\ge& \prob{\curbrkt{\wh\bfm_A \ne \bfm_A} \cap \cE_1^c\cap\cE_2^c } \notag \\
=& \prob{\cE_1^c\cap\cE_2^c}\prob{\wh\bfm_A \ne \bfm_A | \cE_1^c\cap\cE_2^c} \notag \\
\ge& \frac{1}{2}(1 - \prob{\cE_1} - \prob{\cE_2} ). \notag
\end{align}

First note that $\prob{\cE_2} = 1/ |\cC_A| $ which is at most $ \frac{\eps}{2(1+\eps)} $ as long as $ |\cC_A|\ge\frac{2(1+\eps)}{\eps} $.

We next  upper bound $\prob{\cE_1}$. Suppose $N = \frac{3}{4}P(1+\eps)$. 
By Markov's inequality, 
\begin{align}
\prob{\cE} =& \prob{\normtwo{\wt\vbfs} > \sqrt{nN}} \notag \\
\le& \frac{\expt{\normtwo{\wt\vbfs}^2}}{nN}. \notag 
\end{align}
It suffices to upper bound $\expt{\normtwo{\wt\vbfs}^2}$.
\begin{align}
\expt{\normtwo{\wt\vbfs}^2} =& \expt{\normtwo{ -\frac{1}{2}(\vbfx_A+\vbfx_B - \vbfx_A') }^2} \notag \\
=&\frac{1}{4}\paren{ \expt{\normtwo{\vbfx_A}^2} + \expt{\normtwo{\vbfx_B}^2} + \expt{\normtwo{\vbfx_A'}^2} + 2\expt{\inprod{\vbfx_A}{\vbfx_B}} - 2\expt{\inprod{\vbfx_A}{\vbfx_A'}} - 2\expt{\inprod{\vbfx_B}{\vbfx_A'}} } \notag \\
\le&\frac{1}{4}(nP+nP+nP+0-0-0) \notag \\
=&3nP/4. \notag
\end{align}
Then we get that
\begin{align*}
\prob{\cE}\le&\frac{3nP/4}{nN} \notag \\
=&\frac{1}{1+\eps}. \notag
\end{align*}


Substituting the above bound back, we have 
\begin{align}
P_{\e,B} \ge& \frac{1}{2}\paren{1- \frac{1}{1+\eps} - \frac{\eps}{2(1+\eps)}} \notag \\
=& \frac{\eps}{4(1+\eps)}. \notag
\end{align}
\end{proof}

\subsection{Some remarks}
\begin{enumerate}
    \item 
    Using tools from \cite{polyanskiy-verdu-2014-empirical-distr}, we are able to get a satisfactory bound on $Q$ under \emph{maximum}  probability of error criterion. 
    However, such a criterion makes our problem much harder and less interesting. 
    Indeed, by symmetry, let us consider Bob. 
    To make the maximum error probability large, James only needs to focus on one message. 
    Said differently, we can assume that James knows the message corresponding to the transmitted codeword. 
    Under \emph{deterministic} encoding, this means that he knows the actual codeword $\vbfx_B $ from Bob. 
    Given his observation $\vbfz = \vbfx_A+\vbfx_B $, he also knows $\vbfx_A $. Since Bob aims to decode the message corresponding to $\vbfx_A $, James is essentially omniscient in this case.
    The problem of determining the channel capacity of Bob collapses to  the long-standing sphere packing problem.
    In fact \cite{zhang-quadratic-isit}, even  \emph{stochastic} encoding does not help beat the sphere packing bound. 
    As long as James knows the transmitted message, there is a reduction from stochastic encoding to deterministic encoding which turns James omniscient again.
    \item The effective channel to Bob who aims to decode $ \vbfx_A $ is like a myopic adversarial channel if we treat $\vbfx_B $ as noise to James. One difference is that  the noise to James is known to Bob, which is usually not assumed in the myopic model. 
	\item In the general asymmetric case where $P_A $ and $P_B $ can differ, and $N_A $ and $N_B $ can also differ, following exactly the same proof as in Sec. \ref{sec:z-aware_symm}, we get that 
	\begin{itemize}
		\item $C_B =0 $ if $N_B>\frac{2P_A+P_B}{4} $;
		\item $C_A =0 $ if $N_A>\frac{2P_B+P_A}{4} $.
	\end{itemize}
	\item Empirical properties of good codes are not applicable in Sec. \ref{sec:z-aware_symm}. 
	If the channel is symmetrizable, the capacity is zero and any code has subexponential size.  It does not make sense to talk about capacity-achieving distributions, letting alone empirical properties w.r.t. such distributions. 
\end{enumerate}

\section{Concluding remarks and open problems}\label{sec:concl_rk_open_problems}
This paper studies fundamental limits to a two-party message exchange problem over a two-way channel controlled by a malicious adversary who has access to the sum of transmitted signals. 
We conclude the paper with some final remarks and open questions for future research. 
\begin{itemize}
	\item Only in the high-rate regime, our upper bound due to scale-and-babble attack matches our lower bound based on expurgated lattice codes and estimation-type decoder. 
	Specifically, we require $ \snr $ to be a function ($g(\delta)$) of the gap-to-capacity $\delta$.
	We believe such a technical requirement can be relaxed to a condition that $ \snr>K $ for certain universal constant $K$ (independent of $\delta$).
	The can be potentially proved by bounding the error probability also over the random lattice construction, e.g., via Construction-A.\footnote{Indeed, one of the ingredients of achievability, the sumset property, has already been proved in Sec. \ref{sec:improved_sumset} without imposing the constraint $\snr\ge g(\delta) $. However, we have trouble finishing the rest of the proof. The main challenge is due to dependencies among random lattice points inherently caused by linearity.}
	In this way, the technical requirement on $ \snr $ will be replaced by a  large field size $q$ of the based code in Construction-A, which we are fine to afford. 
	\item We do not believe that the constraint on $ \snr $ can be completely removed. 
	Instead, we believe that in low-$\snr$ regime the capacity is strictly less than $ \frac{1}{2}\log\paren{\frac{1}{2}+\frac{P}{N}} $.
	The intuition comes from our symmetrization result. 
	The bound $ \frac{1}{2}\log\paren{\frac{1}{2}+\frac{P}{N}} $ is only valid when  $ \snr\ge1/2 $ since otherwise it is negative. 
	However, our $\vbfz$-aware symmetrization attack shows that no positive rate can be achieved as long as $ \snr\le4/3 $. The threshold $4/3$ is larger than $1/2$ at which the bound $ \frac{1}{2}\log\paren{\frac{1}{2}+\frac{P}{N}} $ is still strictly positive. 
	Such a gap suggests that our bound may not be tight in the low-$\snr$ regime. 
	Understanding the behaviour of capacity in the low-$\snr$ regime remains an intriguing open question.
\end{itemize}

\section{Acknowledgement}
SJ would like to thank Bobak Nazer and Or Ordentlich for helpful discussions at the early stage of this work when he visited Boston University on sabbatical. 

\appendices
\section{Lattice primer}\label{sec:primer_lattices}

For a tutorial introduction to lattices and their applications, see the book by Zamir~\cite{zamir2014latticebook} or the notes by Barvinok~\cite{barvinok2013math669}.

If $ \vv_1,\ldots,\vv_k $  are linearly independent vectors in $ \bR^n $, then the set of all integer linear combinations of $ \vv_1,\ldots,\vv_k $ is called the lattice generated by the vectors $ \vv_1,\ldots,\vv_k $, i.e.,
\[
\Lf \coloneqq \curbrkt{\sum_{i=1}^{k}a_i\vv_i: a_i\in\bZ }.
\]
If  $ \bfG=[\vv_1\cdots \vv_k] $, then we can write $ \Lf = \bfG\bZ^k $. The matrix $ \bfG $ is called a generator matrix for $ \Lf $. The generator matrix of a lattice is not unique.
The integer $ k $ is invariant for a lattice and is called the rank of $ \Lf $. In this  paper, we only consider lattices in $ \bR^n $ having rank $ n $.
It is obvious that $ \Lf $ is a discrete subgroup of $ \bR^n $ under vector addition. It is also a fact that every discrete subgroup of $ \bR^n $ is a lattice~\cite{barvinok2013math669}.

For any lattice $ \Lf $, it is natural to define the quantizer $ Q_{\Lf} $ which maps every point in $ \bR^n $ to the closest lattice point, i.e., for every $ \vx\in\bR^n $,
\begin{equation}
    \label{eqn:lattice_quant}
    Q_{\Lf}(\vx) \coloneqq \argmin{\vy\in\Lf}\Vert \vy-\vx \Vert,
\end{equation}
where we assume that ties (in computing the closest lattice point) are resolved according to some arbitrary but fixed rule. Associated with the quantizer is the quantization error
\[
[\vx]\bmod \Lf \coloneqq \vx - Q_{\Lf}(\vx).
\]

For every lattice $ \Lf $, we define the following parameters:
\begin{itemize}
	\item The set $$ \cP(\Lf)\coloneqq \{ \bfG\vx:\vx\in [0,1 )^n \}, $$
	where $ \bfG $ is a generator matrix of $ \Lf $, is called the fundamental parallelepiped of $ \Lf $.
	\item The fundamental Voronoi region $ \cV(\Lf) $ is the set of all points in $ \bR^n $ which are closest to the zero lattice point. In other words,
	\[
	\cV(\Lf) \coloneqq \{ \vx\in\bR^n: Q_{\Lf}(\vx) = \vzero \}.
	\]
	Any set $ \cS\subset  \bR^n  $ such that the set of translates of $ \cS $ by lattice points, i.e., $ \{\cS+\vx :\vx\in\Lf \} $ form a partition of $ \bR^n $, is called a fundamental region of $ \Lf$. It is a fact that every fundamental region of $ \Lf $ has the same volume equal to $\det\Lf\coloneqq |\det(\bfG)| $, where $ \bfG $ is any generator matrix of $ \Lf $. The quantity $\det \Lf $ is called the determinant or covolume of $ \Lf $ (also denoted by $ \vol(\Lf) $).
	It is a fact that $ \det\Lf = \vol(\cV(\Lf)) $. 
	\item The covering radius $ \rcov(\Lf) $ is the radius of the smallest closed ball in $ \bR^n $ which contains $ \cV(\Lf) $. It is also equal to the length of the largest vector within $ \cV(\Lf) $.
	\item The packing radius $ \rpack(\Lf) $ is the radius of the largest open ball which is contained within $ \cV(\Lf) $. Equivalently, it is half the minimum distance between two lattice points.
	\item The effective radius $ \reff(\Lf) $ is equal to the radius of a ball having volume equal to $ \vol(\cV(\Lf)) $.
\end{itemize}
Clearly, we have $ \rpack(\Lf)\leq \reff(\Lf)\leq \rcov(\Lf) $.

In the context of power-constrained communication over Gaussian channels, a  lattice code is typically the set of all lattice points within a convex compact subset of $ \bR^n $, i.e., $ \cC = \Lf\cap\cB $ for some set $ \cB\subset \bR^n $. Usually $ \cB $ is taken to be $ \cB^n(\vzero ,\sqrt{nP}) $ or $ \cV(\Lc) $ for some lattice $ \Lc $ constructed so as to satisfy the power constraint.

If $ \Lc,\Lf $ are two lattices in $ \bR^n $ with the property that $ \Lc\subsetneq\Lf $, then $ \Lc $ is said to be nested within (or, a sublattice of) $ \Lf $. A nested lattice code with a fine lattice $ \Lf $ and coarse lattice $ \Lc\subsetneq \Lf $ is the lattice code $ \Lf\cap\cV(\Lc) $. 

Lattices have been extensively used for problems of packing, covering and communication over Gaussian channels. For many problems of interest, we want to construct high-dimensional lattices $ \Lf $ such that $ \rpack(\Lf)/\reff(\Lf) $ is as large as possible, and $ \rcov(\Lf)/\reff(\Lf) $ is as small as possible. A class of lattices that has these properties is the class of Construction-A lattices, which we describe next.

Let $ q $ be a prime number, and $ \cC_{\mathrm{lin}} $ be an $ (n,k) $ linear code over $ \bF_q $. The Construction-A lattice obtained from $ \cC_{\mathrm{lin}} $ is defined to be
\[
\Lf(\cC_{\mathrm{lin}}) \coloneqq \{ \vv\in\bZ^n: [\vv]\bmod (q\bZ^n)\in\Phi(\cC) \},
\] 
where $ \Phi $ denotes the natural embedding of $ \bF_q^n $ in $ \bR^n $.
An equivalent definition is that $ \Lf(\cC_{\mathrm{lin}}) = \Phi(\cC_{\mathrm{lin}})+q\bZ^n $.
We make use of the following result to choose our coarse lattices:
\begin{theorem}[\cite{erez-2005-lattices-goodfor-everything}]
	For every $ \delta>0 $,
	there exist sequences of prime numbers $ q_n $ and positive integers $ k_n $ such that if $ \cC_{\mathrm{lin}} $ is a randomly chosen linear code\footnote{The $ (n,k_n) $ random code is obtained by choosing an $ n\times k_n $ generator matrix uniformly at random over $ \bF_q $.} over $ \bF_{q_n} $, then
	\[
	\Pr\left[ \frac{\rpack(\Lf(\cC_{\mathrm{lin}}))}{\reff(\Lf(\cC_{\mathrm{lin}}))}<\frac{1}{2}-\delta \text{ or } \frac{\rcov(\Lf(\cC_{\mathrm{lin}}))}{\reff(\Lf(\cC_{\mathrm{lin}}))}>1+\delta\right] = o(1).
	\] 
\end{theorem}

\bibliographystyle{alpha}
\bibliography{IEEEabrv,ref}

\end{document}